%% file: paper.tex
\pdfoutput=1

\documentclass[10pt, conference, compsocconf]{IEEEtran}

\usepackage{amsmath,scalefnt}
\usepackage{amsfonts}
\usepackage{amssymb}
\usepackage{graphicx}
\usepackage{xspace}
\usepackage{wasysym}

\usepackage{tikz}
\usetikzlibrary{shapes,arrows,calc}

\newtheorem{theorem}{Theorem}[section]
\newtheorem{lemma}[theorem]{Lemma}
\newtheorem{definition}[theorem]{Definition}
\newtheorem{corollary}[theorem]{Corollary}
\newtheorem{condition}[theorem]{Condition}

\newcommand{\N}{\mathbb{N}}
\newcommand{\R}{\mathbb{R}}
\newcommand{\E}{\mathbb{E}}
\newcommand{\Ss}{\mathbb{S}}
\newcommand{\BO}{\mathcal{O}}
\newcommand{\sr}[2]{\stackrel{\eqref{#1}}{#2}}

\newcommand{\acc}{\emph{accept}\xspace}
\newcommand{\slp}{\emph{sleep}\xspace}
\newcommand{\rdy}{\emph{ready}\xspace}
\newcommand{\prop}{\emph{propose}\xspace}
\newcommand{\rec}{\emph{recover}\xspace}
\newcommand{\init}{\emph{init}\xspace}
\newcommand{\join}{\emph{join}\xspace}
\newcommand{\none}{\emph{none}\xspace}
\newcommand{\res}{\emph{resync}\xspace}
\newcommand{\supp}{\emph{supp}\xspace}
\newcommand{\dorm}{\emph{dormant}\xspace}
\newcommand{\srw}{\ensuremath{\emph{sleep}\rightarrow\emph{waking}}\xspace}
\newcommand{\srr}{\ensuremath{\emph{supp}\rightarrow\emph{resync}}\xspace}
\newcommand{\act}{\emph{active}\xspace}
\newcommand{\pass}{\emph{passive}\xspace}
\newcommand{\wake}{\emph{waking}\xspace}

\newcommand{\wait}{\emph{wait}\xspace}
\newcommand{\darts}{\mbox{\sc{darts}}\xspace}

\newcommand{\Next}[1][]{\mbox{\sc{next}\ensuremath{#1}}\xspace}
\newcommand{\accp}{\emph{accept\ensuremath{^+}}\xspace}
\newcommand{\rdyp}{\emph{ready\ensuremath{^+}}\xspace}
\newcommand{\propp}{\emph{propose\ensuremath{^+}}\xspace}

\DeclareMathOperator{\Mem}{Mem}
\DeclareMathOperator{\Time}{Time}

\newcommand{\namedref}[2]{\hyperref[#2]{#1~\ref*{#2}}}
\newcommand{\sectionref}[1]{\namedref{Section}{#1}}

\newcommand{\theoremref}[1]{\namedref{Theorem}{#1}}
\newcommand{\defref}[1]{\namedref{Definition}{#1}}
\newcommand{\figureref}[1]{\namedref{Figure}{#1}}
\newcommand{\figurerefs}[1]{\namedref{Figures}{#1}}

\newcommand{\lemmaref}[1]{\namedref{Lemma}{#1}}

\newcommand{\corollaryref}[1]{\namedref{Corollary}{#1}}

\newcommand{\conditionref}[1]{\namedref{Condition}{#1}}
\newcommand{\equalityref}[1]{\hyperref[#1]{Equality~\eqref{#1}}}
\newcommand{\inequalityref}[1]{\hyperref[#1]{Inequality~\eqref{#1}}}

\newcommand{\emn}[1]{{\em #1}\/}
\newcommand{\LRa}{\Leftrightarrow}

\newcommand{\bfno}[1]{{\noindent\bf #1}\/}

\newcommand{\theterm}{(R_1+4\Delta_g+T_1+(8\vartheta+16)d)}

\begin{document}

\title{FATAL$^+$: A Self-Stabilizing Byzantine Fault-tolerant Clocking Scheme
for SoCs}

\author{\IEEEauthorblockN{Danny Dolev\IEEEauthorrefmark{1},
Matthias F\"ugger\IEEEauthorrefmark{2},
Christoph Lenzen\IEEEauthorrefmark{3},
Markus Posch\IEEEauthorrefmark{2},
Ulrich Schmid\IEEEauthorrefmark{2},
and Andreas Steininger\IEEEauthorrefmark{2}}
\IEEEauthorblockA{\IEEEauthorrefmark{1}
Hebrew University of Jerusalem\\
Jerusalem, Israel\\
Email: dolev@cs.huji.ac.il}
\IEEEauthorblockA{\IEEEauthorrefmark{2}
Vienna University of Technology\\
Vienna, Austria\\
Email: \{fuegger,mposch,s,steininger\}@ecs.tuwien.ac.at}
\IEEEauthorblockA{\IEEEauthorrefmark{3}
Weizmann Institute of Science\\
Rehovot, Israel\\
Email: clenzen@cs.huji.ac.il}}

\maketitle

\input{abstract}


\input{intro}
\input{model}

\input{algo}
\input{analysis}
\input{generalizations}
\input{application}

\input{implementation}
\input{evaluation}

\input{conclusions}

\bibliographystyle{IEEEtran}

\bibliography{pulse}

\end{document}

%% file: abstract.tex
\begin{abstract}
We present concept and implementation of a self-stabilizing Byzantine
fault-tolerant distributed clock generation scheme for multi-synchronous GALS
architectures in critical applications. It combines a variant of a recently
introduced self-stabilizing algorithm for generating low-frequency, low-accuracy
synchronized pulses with a simple non-stabilizing high-frequency,
high-accuracy clock synchronization algorithm. We provide thorough correctness
proofs and a performance analysis, which use methods from fault-tolerant distributed
computing research but also addresses hardware-related issues like
metastability. The algorithm, which consists of several concurrent communicating
asynchronous state machines, has been implemented in VHDL using Petrify in
conjunction with some extensions, and synthetisized for an Altera Cyclone FPGA.
An experimental validation of this prototype has been carried out to confirm the
skew and clock frequency bounds predicted by the theoretical analysis, as well
as the very short stabilization times (required for recovering after excessively
many transient failures) achievable in practice.
\end{abstract}

%% file: intro.tex
\section{Introduction}
\label{sec:Intro}

To circumvent the cumbersome clock tree engineering
issue~\cite{BZMLCLD02,Fri01,MDM04,Resetal01}, \emph{systems-on-chip}
(SoC) are nowadays increasingly designed \emph{globally asynchronous locally
synchronous} (GALS) \cite{Cha84}. Using independent and hence unsynchronized
clock domains requires asynchronous cross-domain communication mechanisms or
synchronizers~\cite{DB99,KBY02,PM95}, however, which inevitably create the
potential for metastability \cite{Mar81}. This problem can be circumvented by means of
multi-synchronous clocking \cite{SG03,TGL07}, which guarantees a certain degree
of synchrony between clock domains. Multi-synchronous GALS is particularly
beneficial from a designer's point of view, since it combines the convenient
local synchrony of a GALS system with a global time base across the whole chip,
including the ability for metastability-free high-speed communication across
clock domains~\cite{PHS09:SSS}.

The decreasing feature sizes of deep submicron VLSI technology also resulted in an
increased likelihood of chip components failing during operation: Reduced
voltage swings and smaller critical charges make circuits more susceptible to
ionized particle hits, crosstalk, and electromagnetic
interference~\cite{Con03,KHP04,Bau05,GEBC06,MKB04,MA01}. \emph{Fault-tolerance}
hence becomes an increasingly pressing issue also for chip design.
Unfortunately, faulty components may behave non-benign in many ways: They
may perform signal transitions at arbitrary times and even convey inconsistent
information to their successor components if their outputs
are affected by a failure. Well-known theory on fault-tolerant agreement and
synchronization shows that this behaviour is the key feature of unrestricted,
i.e., \emph{Byzantine} faults \cite{PSL80}. This forces to model faulty
components as Byzantine if a high fault coverage is to be guaranteed.

Unfortunately, lower-bound results \cite{PSL80,DHS86} reveal that, in order to
cope with some maximum number $f$ of Byzantine faulty components (say,
processors) throughout an execution of a system, $n\geq 3f+1$ components are
required. Given the typically transient nature of failures in digital circuits,
these bounds reveal that even a Byzantine fault-tolerant system cannot be
expected to recover from a situation where more than $f$ components became
faulty transiently, since their state may be corrupted. Dealing with this
problem is in the realm of \emph{self-stabilizing algorithms} \cite{Dol00},
which are guaranteed to recover even if each and every component of the system
fails arbitrarily, but later on works according to its specification again:
in that case the system resumes correct operation after some \emph{stabilization
time} following the instant when no more failures occur.
\emph{Byzantine-tolerant self-stabilizing
algorithms}~\cite{BDH08:podc,DD06,DolWelSSBYZCS04,HDD06:SSS,Mal06:SSS,DH07:SSS,DFLS11:sss}
combine the best of both worlds, by guaranteeing both correct operation and
self-stabilization in the presence of up to $f$ Byzantine faulty components in
the system.

This paper presents concept and prototype implementation of a novel approach,
termed FATAL$^+$, for multi-synchronous clocking in GALS systems. It relies on a
self-stabilizing and Byzantine fault-tolerant distributed algorithm, consisting
of $n$ identical instances (called \emn{nodes}), which generate $n$ local clock
signals (one for each clock domain) with the following properties: \emph{Bounded
skew}, i.e., bounded maximum time between the $k$-th clock transitions of any
two clock signals of correct nodes, and \emph{bounded accuracy (i.e.,
frequency)}, i.e., bounded minimum and maximum time between the occurence of any
two successive clock transitions of the clock signal at any correct node. At
most $f< n/3$ nodes may behave Byzantine faulty, in which case their clock
signals may be arbitrary. The whole algorithm can be directly implemented in
hardware, without quartz oscillators, using standard asynchronous logic gates
only.

FATAL$^+$ self-stabilizes within $\BO(kn)$ time with probability $1-2^{-k(n-f)}$
(with constant expectation in typical settings), and is metastability-free by
construction after stabilization in failure-free runs.\footnote{It is easy to see that,
metastable upsets cannot be ruled out in executions involving Byzantine faults.
However, they can be made as unlikely as desired by using synchronizers or elastic
pipelines acting as metastability filters \cite{FFS09:ASYNC09}.} If the number
of faults is not overwhelming, i.e., a majority of at least $n-f$ nodes
continues to execute the protocol in an orderly fashion, recovering nodes and
late joiners (re)synchronize deterministically in constant time.

{\bf Detailed contributions:} (1) In
Sections~\ref{sec:model}--\ref{sec:application}, we present concept and
theoretical analysis of FATAL$^+$, which is based on a variant of the randomized
self-stabilizing Byzantine-tolerant pulse synchronization algorithm
\cite{DFLS11:sss} we recently proposed. It eventually generates synchronized
periodic pulses with moderate skew and low frequency, and improves upon the
results from \cite{DFLS11:sss} in that it tolerates arbitrarily large clock
drifts and allows late joiners or nodes recovering from transient faults to
deterministically resynchronize within constant time. The formal proof of these
properties builds upon and extends the analysis in \cite{DFLS11:TR}.
In~\sectionref{sec:application}, this algorithm is integrated with a
Byzantine-tolerant but non-self-stabilizing tick generation algorithm based on
Srikanth \& Touegs clock synchronization algorithm \cite{ST87}, operating in a
control loop: The latter, referred to as the \emph{quick cycle} algorithm,
generates clock ticks with high frequency and small skew, which also (weakly)
affect pulse generation. On the other hand, quick cycle uses pulses to monitor
its ticks in order to detect the need for stabilization.

(2) In \sectionref{sec:implementation}, we present the major ingredients of an
Altera Cyclone IV FPGA protoype implementation of FATAL$^+$. It primarily
consists of multiple hybrid (asynchronous + synchronous) state machines, which
have been generated semi-automatically from the specification of the algorithms
using Petrify \cite{CKKLY02}. Non-standard extensions were needed for ensuring
deadlock-free communication despite arbitrarily many desynchronized nodes, some
of which could be Byzantine faulty, which e.g.\ forced us to use state-based
communication instead of handshake-based communication. Special care had also to
be exercised for ensuring self-stabilizing elementary building blocks and
metastability-freedom in normal operation (after stabilization).

(3) In \sectionref{sec:experiments}, we provide some results of the experimental
evaluation of our prototype implementation. They demonstrate the feasibility of
FATAL$^+$ and confirm the results of our theoretical analysis, in particular,
a tight skew bound, in the presence of Byzantine faulty nodes.
Special emphasis has been put on experiments validating the predictions related to
stabilization time, which revealed that the system indeed stabilizes in very
short time from any initial/error state.

\sectionref{sec:conclusions} eventually concludes our paper.

{\bf Related work:} 
The work \cite{FM05,MA03,MA04,Fai04} on distributed clock generation in VLSI
circuits is essentially based on (distributed) ring oscillators, formed by
regular structures (rings, meshes) of multiple inverter loops. Since clock
synchronization theory \cite{DHS86} reveals that high connectivity is required
for bounded synchronization tightness in the presence of failures, these
approaches are fundamentally restricted in that they can overcome at most a
small constant number of Byzantine failures.

The only exception we are aware of is the \darts\ fault-tolerant clock
generation approach~\cite{FSFK06:edcc,FDS10:edcc}, which also adresses
multi-synchronous clocking in GALS systems. Like FATAL$^+$, \darts\ is based on
a fault-tolerant distributed algorithm~\cite{WS09:DC} implemented in
asynchronous digital logic. Although it shares many features with FATAL$^+$,
including Byzantine fault-tolerance, it is not self-stabilizing: If more than
$f$ nodes ever become faulty, the system will not recover even if all nodes work
correctly thereafter. Moreover, in \darts, simple transient faults such as
radiation- or crosstalk-induced additional (or omitted) clock ticks accumulate
over time to arbitrarily large skews in an otherwise benign execution. Despite
not suffering from these drawbacks, FATAL$^+$ offers similar
guarantees in terms of area consumption, clock skew, and amortized
frequency as \darts.

Furthermore, a number of Byzantine-tolerant self-stabili\-zing clock
synchronization
protocols~\cite{BDH08:podc,DD06,DolWelSSBYZCS04,HDD06:SSS,Mal06:SSS,DH07:SSS}
have been devised by the distributed systems community. Beyond optimal
resilience, an attractive feature of most of these protocols is a small
stabilization time. However, all of them exhibit deficiencies rendering them
unsuitable in the VLSI context. This motivated to devise the algorithm
from~\cite{DFLS11:sss,DFLS11:TR}, an improved variant of which forms the basis
of FATAL$^+$.

%% file: model.tex
\section{Model}\label{sec:model}

In this section we introduce our system model. Our formal framework will be tied
to the peculiarities of hardware designs, which consist of modules that
\emph{continuously}\footnote{In sharp contrast to classic distributed computing
models, there is no computationally complex discrete zero-time state-transition
here.} compute their output signals based on their input signals.

\subsection*{Signals}

Following \cite{Fue10:diss,FS10:TR}, we define (the trace of) a
\emn{signal} to be a timed event trace over a finite alphabet $\Ss$ of possible
signal states: Formally, signal $\sigma \subseteq \Ss \times \R_0^+$. All times
and time intervals refer to a global \emph{reference time} taken from $\R_0^+$,
that is, signals reflect the system's state from time~0 on. The elements of
$\sigma$ are called \emn{events}, and for each event $(s,t)$ we call $s$ the
\emn{state of event} $(s,t)$ and $t$ the \emn{time of event} $(s,t)$. In
general, a signal $\sigma$ is required to fulfill the following conditions: (i)
for each time interval $[t^-,t^+]\subseteq \R_0^+$ of finite length, the number
of events in $\sigma$ with times within $[t^-,t^+]$ is finite, (ii) from $(s,t)
\in \sigma$ and $(s',t)\in \sigma$ follows that $s=s'$, and (iii) there exists
an event at time~$0$ in $\sigma$.

Note that our definition allows for events $(s,t)$ and $(s,t')\in \sigma$, where
$t < t'$, without having an event $(s',t'')\in \sigma$ with $s'\neq s$ and $t <
t'' < t'$. In this case, we call event $(s,t')$ \emph{idempotent}. Two signals
$\sigma$ and $\sigma'$ are \emph{equivalent}, iff they differ in idempotent
events only. We identify all signals of an equivalence class, as they describe
the same physical signal. Each equivalence class $[\sigma]$ of signals contains
a unique signal $\sigma_0$ having no idempotent events. We say that \emn{signal
$\sigma$ switches to} $s$ at time~$t$ iff event $(s,t)\in \sigma_0$.

The \emn{state of signal}~$\sigma$ at time $t \in \R_0^+$, denoted by
$\sigma(t)$, is given by the state of the event with the maximum time not
greater than $t$.\footnote{To facilitate intuition, we here slightly abuse
notation, as this way $\sigma$ denotes both a function of time and the signal
(trace), which is a subset of $\Ss \times \R_0^+$. Whenever referring to
$\sigma$, we will talk of the signal, not the state function.} Because of (i),
(ii) and (iii), $\sigma(t)$ is well defined for each time~$t\in \R_0^+$. Note
that $\sigma$'s state function in fact depends on $[\sigma]$ only, i.e., we may
add or remove idempotent events at will without changing the state function.

\subsection*{Distributed System}

On the topmost level of abstraction, we see the system as a set of
$V=\{1,\ldots,n\}$ physically remote \emph{nodes} that communicate by means of
\emph{channels}. In the context of a VLSI circuit, ``physically remote''
actually refers to quite small distances (centimeters or even less). However, at
gigahertz frequencies, a local state transition will not be observed remotely
within a time that is negligible compared to clock speeds. We stress this point,
since it is crucial that different clocks (and their attached logic) are not placed too
close to each other, as otherwise they might fail due to the same event such as
a particle hit. This would render it pointless to devise a system that is
resilient to a certain fraction of the nodes failing.

Each node~$i$ comprises a number of \emn{input ports}, namely $S_{i,j}$ for each
node~$j$, an \emn{output port} $S_i$, and a set of \emn{local ports}, introduced
later on. An \emn{execution} of the distributed system assigns to each port of
each node a signal. For convenience of notation, for any port $p$, we refer to
the signal assigned to port $p$ simply by signal~$p$. We say that \emn{node~$i$
is in state~$s$} at time~$t$ iff $S_i(t)=s$. We further say that \emn{node~$i$
switches to state~$s$} at time~$t$ iff signal~$S_i$ switches to~$s$ at time~$t$.

Nodes exchange their states via the channels between them: for each pair of
nodes $i,j$, output port $S_i$ is connected to input port $S_{j,i}$ by a FIFO
channel from $i$ to $j$. Note that this includes a channel from $i$ to $i$
itself. Intuitively, $S_i$ being connected to $S_{j,i}$ by a (non-faulty)
channel means that $S_{j,i}(\cdot)$ should mimic $S_i(\cdot)$, however, with a
slight delay accounting for the time it takes the channel to propagate events.
In contrast to an asynchronous system, this delay is bounded by the \emn{maximum
delay} $d > 0$.\footnote{With respect to $\BO$-notation, we normalize $d\in
\BO(1)$, as all time bounds simply depend linearly on $d$.}

Formally we define: The \emn{channel} from node $i$ to $j$ is said to
     be \emn{correct} during $[t^-,t^+]$ iff there exists a function
     $\tau_{i,j}: \R_0^+ \to \R_0^+$, called the channel's \emn{delay
     function}, such that: (i) $\tau_{i,j}$ is continuous and strictly
     increasing, (ii) $\forall t\in [\max(t^-,\tau_{i,j}(0)),t^+]:
     0 < t-\tau^{-1}_{i,j}(t) < d$, and (iii) for each $t \in
     [\max(t^-,\tau_{i,j}(0)),t^+]$, $(s,t) \in S_{j,i} \LRa
     (s,\tau^{-1}_{i,j}(t)) \in S_{i}$, and for each $t \in
     [t^-,\tau_{i,j}(0))$, $(s,t) \in S_{j,i} \Rightarrow s=S_{i}(0)$.
Note that because of (i), $\tau_{i,j}^{-1}$ exists in the domain
     $[\tau_{i,j}(0),\infty)$, and thus (ii) and (iii) are well
     defined.
We say that node $i$ \emph{observes node~$j$ in state~$s$} at time~$t$
     if $S_{i,j}(t)=s$.

\subsection*{Clocks and Timeouts}

Nodes are never aware of the current reference time and we also do not require
the reference time to resemble Newtonian ``real'' time. Rather we allow for
physical clocks that run arbitrarily fast or slow,\footnote{Note that the
formal definition excludes trivial solutions by requiring clocks' progress
to be in a linear envelope of the reference time, see below.} as long as their
speeds are close to each other in comparison. One may hence think of the
reference time as progressing at the speed of the currently slowest correct
clock. In this framework, nodes essentially make use of bounded clocks with
bounded drift.

Formally, clock rates are within $[1,\vartheta]$ (with respect to reference
time), where $\vartheta>1$ is constant and $\vartheta-1$ is the \emph{(maximum)
clock drift}. A \emph{clock} $C$ is a continuous, strictly increasing function
$C:\R^+_0\to \R^+_0$ mapping reference time to some local time. Clock $C$ is
said to be \emn{correct} during $[t^-,t^+]\subseteq \R^+_0$ iff we have for any
$t,t'\in [t^-,t^+]$, $t<t'$, that $t'-t\leq C(t')-C(t)\leq \vartheta (t'-t)$.
Each node comprises a set of clocks assigned to it, which allow the node to
estimate the progress of reference time.

Instead of directly accessing the value of their clocks, nodes have access to
so-called \emn{timeout ports} of watchdog timers. A \emph{timeout} is a triple
$(T,s,C)$, where $T\in \R^+$ is a duration, $s\in \Ss$ is a state, and $C$ is
some local clock (there may be several), say of node~$i$. Each timeout $(T,s,C)$
has a corresponding timeout port $\Time_{T,s,C}$, being part of node $i$'s local
ports. Signal $\Time_{T,s,C}$ is Boolean, that is, its possible states are from
the set $\{0,1\}$. We say that timeout $(T,s,C)$ is \emn{correct} during
$[t^-,t^+]\subseteq \R^+_0$ iff clock~$C$ is correct during $[t^-,t^+]$ and the
following holds:

\begin{enumerate}
  \item For each time $t_s\in [t^-,t^+]$ when node~$i$ switches to state~$s$,
  there is a time $t\in[t_s,\tau_{i,i}(t_s)]$ such that $(T,s,C)$ is \emph{reset},
  i.e., $(0,t)\in \Time_{T,s,C}$. This is a one-to-one correspondence, i.e.,
  $(T,s,C)$ is not reset at any other times.
  
  \item For a time $t\in [t^-,t^+]$, denote by $t_0$ the supremum of all times
  from $[t^-,t]$ when $(T,s,C)$ is reset. Then it holds that $(1,t)\in
  \Time_{T,s,C}$ iff $C(t)-C(t_0) = T$. Again, this is a one-to-one correspondence.
\end{enumerate}

We say that timeout $(T,s,C)$ \emph{expires} at time~$t$ iff
$\Time_{T,s,C}$ switches to~$1$ at time~$t$, and it \emph{is
expired} at time $t$ iff $\Time_{T,s,C}(t)=1$.
For notational convenience, we will omit the clock $C$ and simply write $(T,s)$
for both the timeout and its signal.

A \emn{randomized timeout} is a triple $({\cal D},s,C)$, where ${\cal D}$ is a
bounded random distribution on $\R^+_0$, $s\in \Ss$ is a state, and $C$ is a
clock. Its corresponding timeout port $\Time_{{\cal D},s,C}$ behaves very similar
to the one of an ordinary timeout, except that whenever it is reset, the local
time that passes until it expires next---provided that it is not reset again
before that happens---follows the distribution $\cal D$. Formally, $({\cal
D},s,C)$ is correct during $[t^-,t^+]\subseteq \R^+_0$, if $C$ is correct during
$[t^-,t^+]$ and the following holds:

\begin{enumerate}
  \item For each time $t_s\in [t^-,t^+]$ when node~$i$ switches to state~$s$,
  there is a time $t\in[t_s,\tau_{i,i}(t_s)]$ such that $({\cal D},s,C)$ is
  \emph{reset}, i.e., $(0,t)\in \Time_{{\cal D},s,C}$. This is a one-to-one
  correspondence, i.e., $({\cal D},s,C)$ is not reset at any other times.

  \item For a time $t\in [t^-,t^+]$, denote by $t_0$ the supremum of all times
  from $[t^-,t]$ when $({\cal D},s,C)$ is reset. Let $\mu:\R_0^+\to \R_0^+$
  denote the density of $\cal D$. Then $(1,t)\in \Time_{{\cal D},s,C}$ ``with
  probability $\mu(C(t)-C(t_0))$'' and we require that the probability of $(1,t)\in
  \Time_{{\cal D},s,C}$---conditional to $t_0$ and $C$ on $[t_0,t]$ being
  given---is independent of the system's state at times smaller than $t$. More
  precisely, if superscript $\cal E$ identifies variables in execution $\cal E$
  and $t_0'$ is the infimum of all times from $(t_0,t^+]$ when node~$i$ switches
  to state~$s$, then we demand for any $[\tau^-,\tau^+]\subseteq [t_0,t_0']$ that
  \begin{eqnarray*}
&&P\left[\exists t'\in [\tau^-,\tau^+]:(1,t')\in \Time_{{\cal D},s,C}
\phantom{\Big|}\right.\\
&&\phantom{P}\left.~\Big|\,
t_0^{\cal E}=t_0 \wedge C\big|_{[t_0,t']}^{\cal
E}=C\big|_{[t_0,t']}\right]\\
&=&\int_{\tau^-}^{\tau^+}\mu(C(\tau)-C(t_0))~d\tau,
\end{eqnarray*}
independently of ${\cal E}\big|_{[0,\tau^-)}$.
\end{enumerate}

We will apply the same notational conventions to randomized timeouts as we do
for regular timeouts.

Note that, strictly speaking, this definition does not induce a random variable
describing the time $t'\in [t_0,t_0')$ satisfying that $(1,t')\in \Time_{{\cal
D},s,C}$. However, for the state of the timeout port, we get the
meaningful statement that for any $t'\in [t_0,t_0')$,
\begin{eqnarray*}
&&P[\Time_{{\cal D},s,C}\mbox{ switches to
$1$ during }[t_0,t']]\\
&=&\int_{t_0}^{t'} \mu(C(\tau)-C(t_0))~d\tau.
\end{eqnarray*}
The reason for phrasing the definition in the above more cumbersome way is that
we want to guarantee that an adversary knowing the full present state of the
system and memorizing its whole history cannot reliably predict when the timeout
will expire.\footnote{This is a non-trivial property. For instance nodes could
just determine, by drawing from the desired random distribution at time $t_0$,
at which local clock value the timeout shall expire next. This would, however,
essentially reveal when the timeout will expire prematurely, greatly reducing
the power of randomization!}

We remark that these definitions allow for different timeouts to be driven by
the same clock, implying that an adversary may derive some information on the
state of a randomized timeout before it expires from the node's behavior, even
if it cannot directly access the values of the clock driving the timeout. This
is crucial for implementability, as it might be very difficult to guarantee
that the behavior of a dedicated clock that drives a randomized timeout
is indeed independent of the execution of the algorithm.

\subsection*{Memory Flags}

Besides timeout and randomized timeout ports, another kind of node~$i$'s local
ports are \emn{memory flags}. For each state $s\in \Ss$ and each node $j \in V$,
$\Mem_{i,j,s}$ is a local port of node~$i$. It is used to memorize whether
node~$i$ has observed node~$j$ in state $s$ since the last reset of the flag. We
say that node $i$ \emph{memorizes node $j$ in state $s$} at time $t$ if
$\Mem_{i,j,s}(t)=1$. Formally, we require that signal $\Mem_{i,j,s}$ switches
to~$1$ at time~$t$ iff node~$i$ observes node~$j$ in state~$s$ at time~$t$ and
$\Mem_{i,j,s}$ is not already in state~$1$. The times $t$ when $\Mem_{i,j,s}$ is
\emn{reset}, i.e., $(0,t)\in\Mem_{i,j,s}$, are specified by node~$i$'s state
machine, which is introduced next.

\subsection*{State Machine}

It remains to specify how nodes switch states and when they reset memory flags.
We do this by means of state machines that may attain states from the finite
alphabet $\Ss$. A node's state machine is specified by (i) the set $\Ss$, (ii) a
function $tr$, called the \emn{transition function}, from ${\cal T}\subseteq
\Ss^2$ to the set of Boolean predicates on the alphabet consisting of
expressions ``$p = s$'' (used for expressing guards), where $p$ is from the
node's input and local ports and $s$ is from the set of possible states of
signal~$p$, and (iii) a function $re$, called the \emn{reset function}, from
$\cal T$ to the power set of the node's memory flags.

Intuitively, the transition function specifies the conditions (guards) under
which a node switches states, and the reset function determines
which memory flags to reset upon the state change.
Formally, let $P$ be a predicate on node~$i$'s input and local ports.
We define $P$ \emn{holds at time $t$} by structural induction: If $P$
is equal to $p = s$, where $p$ is one of node $i$'s input and
local ports and $s$ is one of the states signal $p$ can obtain,
then $P$ \emn{holds at time $t$} iff $p(t)=s$.
Otherwise, if $P$ is of the form $\neg P_1$, $P_1 \wedge P_2$, or $P_1
\vee P_2$, we define $P$ \emn{holds at time $t$} in the
straightforward manner.

We say node~$i$ \emn{follows its state machine during} $[t^-,t^+]$ iff
the following holds: Assume node~$i$ observes itself in
state~$s\in \Ss$ at time~$t\in [t^-,t^+]$, i.e., $S_{i,i}(t)=s$.
Then, for each $(s,s')\in \cal T$, both:
\begin{enumerate}
  \item Node~$i$ switches to state~$s'$ at time~$t$ iff $tr(s,s')$ holds
  at time~$t$ and $i$ is not already in state~$s'$.\footnote{Recall that a
  node may still observe itself in state $s$ albeit already having switched to
  $s'$.} (In case more than one guard $tr(s,s')$ can be true at the same time,
  we assume that an arbitrary tie-breaking ordering exists among the transition
  guards that specifies to which state to switch.)
  
  \item Node $i$ resets memory flag $m$ at some time in the interval
  $[t,\tau_{i,i}(t)]$ iff $m\in re(s,s')$ and $i$ switches from state $s$ to
  state~$s'$ at time~$t$. This correspondence is one-to-one.
\end{enumerate}

A node is defined to be \emn{non-faulty} during $[t^-,t^+]$ iff during
$[t^-,t^+]$ all its timeouts and randomized timeouts are
correct and it follows its state machine. If it employs multiple state machines
(see below), it needs to follow all of them.

In contrast, a faulty node may change states arbitrarily.
Note that while a faulty node may be forced to send consistent output
state signals to all other nodes if its channels remain correct,
there is no way to guarantee that this still holds true if
channels are faulty.\footnote{A single physical fault may cause
this behavior, as at some point a node's output port must be
connected to remote nodes' input ports.
Even if one places bifurcations at different physical locations
striving to mitigate this effect, if the voltage at the output
port drops below specifications, the values of corresponding
input channels may deviate in unpredictable ways.}

\subsection*{Metastability}

While the presented model does not fully capture \emn{propagation} and
\emn{decay} of metastable upsets, i.e., the propagation of intermediate values
through combinational circuit elements, and the probability distributions on the
decay of metastable upsets, it allows to capture its \emn{generation}. An
algorithm is inherently susceptible to metastability by the lacking capability
of state machines to instantaneously take on new states: Node~$i$ decides on
state transitions based on the delayed status of port~$S_{i,i}$ instead of its
``true'' current state~$S_i$. Consider the following example: Node~$i$ is in
state $s$ at some time~$t$, but since it switched to $s$ only very recently, it
still observes itself in state $s' \neq s$ at time $t$. A metastable upset might
occur at time~$t$ (i) if the guard $tr(s',s)$ falls back to false at time~$t$,
or (ii) if there is another transition $(s',s'')$ in $T$ whose guard becomes
true at time~$t$. The treatment of scenario (i) is postponed to
Section~\ref{sec:implementation} where it is discussed together with the
implementation of a node's components. Scenario (ii) is accounted for in the
following definition:

\begin{definition}[Metastability-Freedom]\label{def:metastability}
We denote state machine~$M$ of node $i$ as being \emn{metastability-free during
$[t^-,t^+]$}, iff for each time $t\in [t^-,t^+]$ when $M$ switches from some
state $s$ to another state $s'$, it holds that $\tau_{i,i}(t) < t'$, where $t'$
is the infimum of all times in $(t,t^+]$ when $M$ switches to some state~$s''$.
\end{definition}

\subsection*{Multiple State Machines}
In some situations the previous definitions are too stringent, as there might be
different ``components'' of a node's state machine that act concurrently and
independently, mostly relying on signals from disjoint input ports or orthogonal
components of a signal. We model this by permitting that nodes run several state
machines in parallel. All these state machines share the input and local ports
of the respective node and are required to have disjoint state spaces. If node
$i$ runs state machines $M_1,\ldots,M_k$, node $i$'s output signal is the
product of the output signals of the individual machines. Formally we define:
Each of the state machines $M_j$, $1 \leq j \leq k$, has an additional own
output port~$s_j$. The state of node $i$'s output port $S_i$ at any time $t$ is
given by $S_i(t):=(s_1(t),\ldots,s_k(t))$, where the signals of ports
$s_1,\ldots,s_k$ are defined analogously to the signals of the output ports of
state machines in the single state machine case. Note that by this
definition, the only (local) means for node $i$'s state machines to interact
with each other is by reading the delayed state signal $S_{i,i}$.

We say that \emn{node $i$'s state machine $M_j$ is in state $s$ at time $t$} iff
$s_j(t) = s$, where $S_i(t) = (s_1(t),\ldots,s_k(t))$, and that \emn{node $i$'s
state machine $M_j$ switches to state $s$ at time $t$} iff signal $s_j$ switches
to $s$ at time $t$. Since the state spaces of the machines $M_j$ are disjoint,
we will omit the phrase ``state machine $M_j$'' from the notation, i.e., we
write ``node $i$ is in state $s$'' or ``node $i$ switched to state $s$'',
respectively.

Recall that the various state machines of node $i$ are as loosely coupled as
remote nodes, namely via the delayed status signal on channel $S_{i,i}$ only.
Therefore, it makes sense to consider them independently also when it comes to
metastability.

\begin{definition}[Metastability-Freedom---Multiple SM's]\ \\ We denote state
machine $M$ of node $i\in V$ as \emn{metastability-free during $[t^-,t^+]$}, iff
for each time $t\in [t^-,t^+]$ when $M$ switches from some state $s\in \Ss$ to
another state $s'\in \Ss$, it holds that $\tau_{i,i}(t) < t'$, where $t'$ is the
infimum of all times in $(t,t^+]$ when $M$ switches to some state $s'' \in \Ss$.
\end{definition}

Note that by this definition the different state machines may switch
     states concurrently without suffering from
     metastability.\footnote{However, care has to be taken when
     implementing the inter-node communication of the state components
     in a metastability-free manner,
     cf.~\sectionref{sec:implementation}.} It is even possible that
     some state machine suffers metastability, while another is not
     affected by this at all.\footnote{This is crucial for the
     algorithm we are going to present.
For stabilization purposes, nodes comprise a state machine that is
     prone to metastability.
However, the state machine generating pulses (i.e., having the state
     \acc, cf.~\defref{def:pulse}) does not take its output signal
     into account once stabilization is achieved.
Thus, the algorithm is metastability-free after stabilization in the
     sense that we guarantee a metastability-free signal indicating
     when pulses occur.}

\subsection*{Problem Statement}

The purpose of the pulse synchronization protocol is that nodes generate
synchronized, well-separated pulses by switching to a distinguished state \acc.
Self-stabilization requires that it starts to do so within a bounded time, for any
possible initial state. However, as our protocol makes use of randomization,
there are executions where this does not happen at all; instead, we will show
that the protocol stabilizes with probability one in finite time. To give a
precise meaning to this statement, we need to define appropriate probability
spaces.

\begin{definition}[Adversarial Spaces] Denote for $i\in V$ by ${\cal
C}_i= (C_{i,1}, \dots, C_{i,c_i})$ the tuple of clocks of node $i$.
An \emph{adversarial space} is a probabilistic space that is defined
     by subsets of nodes and channels $W\subseteq V$ and $E\subseteq
     V^2$, a time interval $[t^-,t^+]$, a protocol ${\cal P}$ (nodes'
     ports, state machines, etc.) as previously defined, tuple of all
     clocks $({\cal C}_1,\dots,{\cal C}_n)$, a function $\Theta$
     assigning each $(i,j) \in V^2$ a delay $\tau_{i,j}: \R^+_0\to
     \R^+_0$, an initial state ${\cal E}_0$ of all ports, and an
     \emph{adversarial function ${\cal A}$}.
Here ${\cal A}$ is a function that maps a partial execution ${\cal
     E}|_{[0,t]}$ until time $t$ (i.e., all ports' values until time
     $t$), $W$, $E$, $[t^-,t^+]$, ${\cal P}$, ${\cal C}$, and $\Theta$
     to the states of all faulty ports during the time interval
     $(t,t']$, where $t'$ is the infimum of all times greater than $t$
     when a non-faulty node or channel switches states.

The adversarial space ${\cal AS}(W,E,[t^-,t^+],{\cal P},{\cal
     C},\Theta,{\cal E}_0,{\cal A})$ is now defined on the set of all
     executions ${\cal E}$ satisfying that $(i)$ the initial state of
     all ports is given by ${\cal E}|_{[0,0]}={\cal E}_0$, $(ii)$ for
     all $i\in V$ and $k\in \{1,\ldots,c_i\}:$ $C_{i,k}^{\cal
     E}=C_{i,k}$, $(iii)$ for all $(i,j)\in V^2$, $\tau_{i,j}^{\cal
     E}=\tau_{i,j}$, $(iv)$ nodes in $W$ are non-faulty during
     $[t^-,t^+]$ with respect to the protocol ${\cal P}$, $(v)$ all
     channels in $E$ are correct during $[t^-,t^+]$, and $(vi)$ given
     ${\cal E}|_{[0,t]}$ for any time $t$, ${\cal E}|_{(t,t']}$ is
     given by ${\cal A}$, where $t'$ is the infimum of times greater
     than $t$ when a non-faulty node switches states.
Thus, except for when randomized timeouts expire, ${\cal E}$ is fully
     predetermined by the parameters of ${\cal AS}$.\footnote{This
     follows by induction starting from the initial configuration
     ${\cal E}_0$.
Using ${\cal A}$, we can always extend ${\cal E}$ to the next time
     when a correct node switches states, and when non-faulty nodes
     switch states is fully determined by the parameters of ${\cal
     AS}$ except for when randomized timeouts expire.
Note that the induction reaches any finite time within a finite number
     of steps, as signals switch states finitely often in finite
     time.} The probability measure on ${\cal AS}$ is induced by the
     random distributions of the randomized timeouts specified by
     ${\cal P}$.
\end{definition}

To avoid confusion, observe that if the clock functions and delays do not follow
the model constraints during $[t^-,t^+]$, the respective adversarial space is
empty and thus of no concern. This cumbersome definition provides the means to
formalize a notion of stabilization that accounts for worst-case drifts and
delays and an adversary that knows the full state of the system up to the
current time.

We are now in the position to formally state the pulse synchronization
     problem in our framework.
Intuitively, the goal is that after transient faults cease, nodes
     should with probability one eventually start to issue
     well-separated, synchronized pulses by switching to a dedicated
     state \acc.
Thus, as the initial state of the system is arbitrary, specifying an
     algorithm\footnote{We use the terms ``algorithm'' and
     ``protocol'' interchangeably throughout this work.} is equivalent
     to defining the state machines that run at each node, one of
     which has a state \acc.

\begin{definition}[Self-Stabilizing Pulse Synchronization]\label{def:pulse}\ \\
Given a set of nodes $W \subseteq V$ and a set $E\subseteq V\times V$
     of channels, we say that protocol ${\cal P}$ is a
     \emn{$(W,E)$-stabilizing pulse synchronization protocol with skew
     $\Sigma$ and accuracy bounds $T^-> \Sigma$ and $T^+$ that stabilizes
     within time $T$ with probability $p$} iff the following holds.
Choose any time interval $[t^-,t^+]\supseteq [t^-,t^- +T+\Sigma]$ and
     any adversarial space ${\cal AS}(W,E,[t^-,t^+],{\cal
     P},\cdot,\cdot,\cdot,\cdot)$ (i.e., ${\cal C}$, $\Theta$, ${\cal
     E}_0$, and ${\cal A}$ are arbitrary).
Then executions from ${\cal AS}$ satisfy with probability at least $p$
     that there exists a time $t_s \in [t^-,t^- +T]$ so that, denoting
     by $t_i(k)$ the time when node~$i\in W$ switches to a distinguished
     state \acc\ for the $k^\text{th}$ time after $t_s$
     ($t_i(k)=\infty$ if no such time exists), $(i)$ $t_i(1)\in
     (t_s,t_s+\Sigma)$, $(ii)$ $|t_i(k)-t_j(k)| \le \Sigma$ if
     $\max\{t_i(k),t_j(k)\} \le t^+$, and $(iii)$ $T^-\le
     |t_i(k+1)-t_i(k)| \le T^+$ if $t_i(k)+T^+\le t^+$.
\end{definition}

Note that the fact that ${\cal A}$ is a deterministic function and, more
generally, that we consider each space ${\cal AS}$ individually, is no
restriction: As ${\cal P}$ succeeds for any adversarial space with probability
at least $p$ in achieving stabilization, the same holds true for randomized
adversarial strategies ${\cal A}$ and worst-case drifts and delays.

%% file: algo.tex
\section{The FATAL Pulse Synchronization Protocol}\label{sec:algo}

In this section, we present our self-stabilizing pulse generation
     algorithm.
In order to be suitable for implementation in hardware, it needs to
     utilize very simple rules only.
It is stated in terms of state machines as introduced in the previous
     section.

Since the ultimate goal of the pulse generation algorithm is to
     interact with an application layer, we introduce a possibility
     for a coupling with such a layer in the pulse generation
     algorithm itself: for each node~$i$, we add a further port
     $\Next_i$, which can be driven by node $i$'s application layer.
As for other state signals, its output raises flag $\Mem_{i,\Next}$,
     to which for simplicity we refer to as $\Next_i$ as well.
The purpose of the port is to allow the application layer to influence
     the time between two of node $i$'s successively generated pulses
     within a range that does not prevent the pulse generation
     algorithm to stabilize correctly.

In \sectionref{sec:application} we give an example for an application
     layer: The quick cycle completing the FATAL$^+$ is a
     non-self-stabilizing clock synchonization routine which relies on the pulse
     generation algorithm for self-stabilization.
Since we will show that the pulse algorithm stabilizes independently
     of the behavior of the \Next signal, and the clock synchronization
     routine presented \sectionref{sec:application} is designed such
     that it will stabilize once the pulse generation algorithm did
     so, we can partition the analysis of the compound algorithm into
     two parts.
When proving the correctness of the pulse generation algorithm in
     \sectionref{sec:analysis}, we thus assume that for each node~$i$,
     $\Next_i$ is arbitrary.

\subsection{Basic Cycle}
\label{sec:basic_cycle}


\begin{figure*}[t!]
  \centering
  \begin{tikzpicture}
      \draw (8,0) node[circle,minimum width=1.8cm,draw,align=center] (ac) {$\acc$};
      \draw (8,3) node[circle,minimum width=1.8cm,draw,align=center] (sl) {$\slp$};
      \draw (8,6) node[circle,minimum width=1.8cm,draw,align=center] (srw) {$\slp$\\$\rightarrow$\\$\wake$};
      \draw (4,6) node[circle,minimum width=1.8cm,draw] (wa) {$\wake$};
      \draw (0,6) node[circle,minimum width=1.8cm,draw] (rd) {$\rdy$};
      \draw (0,0) node[circle,minimum width=1.8cm,draw] (pr) {$\prop$};

      \path[->] (ac) edge node[right,align=center] {$T_1$ \textbf{and}\\$\ge n-f$\\$\acc$} (sl);
      \path[->] (sl) edge node[right,align=center] {$(2\vartheta+1)T_1$} (srw);
      \path[->] (srw) edge node[above=0.3cm,align=center] {\emph{true}} (wa);
      \path[->] (wa) edge node[above=0.4cm,align=center] {$(T_2,\acc)$} (rd);
      \path[->] (rd) edge node[right,align=center] {$(T_3 \textbf{ and}$\\$\Next_i=1)$\\\textbf{or} $T_4$ \textbf{or}\\$\ge f+1$\\$\prop$\\or $\acc$} (pr);
      \path[->] (pr) edge node[above=0.3cm,align=center] {$\ge n-f$ $\prop \text{ or } \acc$\\\textbf{or} $\ge f+1$ $\acc$} (ac);

      \draw (srw)+(-2,0) node[fill=white,draw=black] {$\acc$};
      \draw (wa)+(-2,0) node[fill=white,draw=black,align=center] {$\prop$,\\$\Next_i$};
      \draw (pr)+(4,0) node[fill=white,draw=black] {$\acc$};
  \end{tikzpicture}
  \caption{Basic cycle of node~$i$ once the algorithm has stabilized.}\label{fig:main_simple}
\end{figure*}
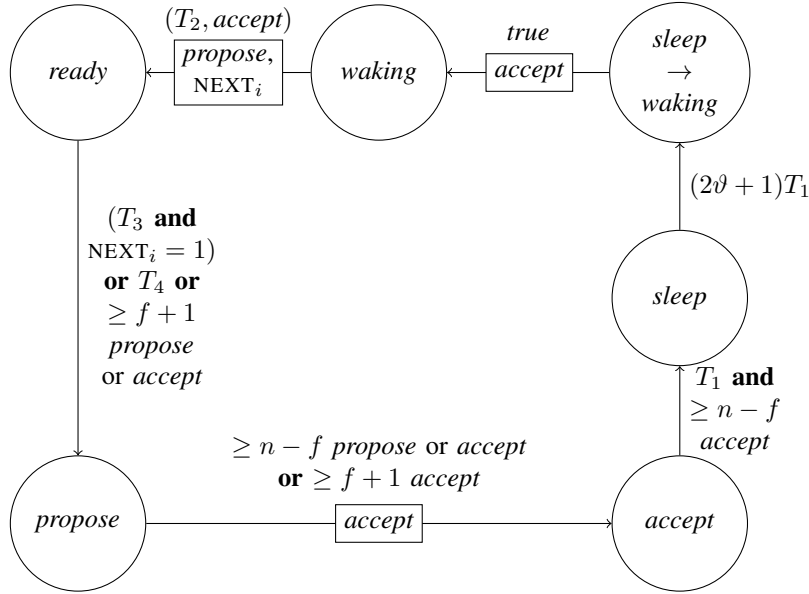

The full pulse generation algorithm makes use of a rather involved interplay between conditions
on timeouts, states, and thresholds to converge to a safe state despite a
limited number of faulty components. As our approach is thus complicated to
present in bulk, we break it down into pieces. Moreover, to facilitate giving
intuition about the key ideas of the algorithm, in this subsection we assume that
there are never more than $f<n/3$ faulty nodes, i.e., the remaining $n-f$ nodes
are non-faulty within $[0,\infty)$. We further assume that channels between non-faulty nodes (including
loopback channels) are correct within $[0,\infty)$. We start by presenting the
basic cycle that is repeated every pulse once a safe configuration is reached
(see \figureref{fig:main_simple}).

We employ graphical representations of the state machine of each node~$i \in V$.
States are represented by circles containing their names, while transition
$(s,s') \in \cal T$ is depicted as an arrow from $s$ to $s'$. The guard
$tr(s,s')$ is written as a label next to the arrow, and the reset function's
value $re(s,s')$ is depicted in a rectangular box on the arrow. To keep labels
more simple we make use of some abbreviations. Recall that in the notation of
timeouts $(T,s,C)$ the driving clock $C$ is omitted. We write $T$ instead of
$(T,s)$ if $s$ is the same state which node~$i$ leaves if the condition
involving $(T,s)$ is satisfied. Threshold conditions like ``\,$\ge f+1$ $s$\,'',
where $s \in \Ss$, abbreviate Boolean predicates that reach over all of
node~$i$'s memory flags $\Mem_{i,j,s}$, where $j \in V$, and are defined in a
straightforward manner. If in such an expression we connect two states by
``or'', e.g., ``\,$\ge n-f$ $s$ or $s'$\,'' for $s,s'\in \Ss$, the summation
considers flags of both types $s$ and $s'$. Thus, such an expression is
equivalent to $\sum_{j\in V}\max\{\Mem_{i,j,s},\Mem_{i,j,s'}\}\geq f+1$. For any
state $s \in \Ss$, the condition $S_{i,i} = s$, (respectively, $\neg(S_{i,i} =
s)$) is written in short as ``in~$s$'' (respectively, ``not in~$s$''). We write
``true'' instead of a condition that is always true (like e.g.\ ``(in $s$) or
(not in $s$)'' for an arbitrary state $s \in \Ss$). Finally, $re(\cdot,\cdot)$
always requires to reset all memory flags of certain types, hence we write e.g.\
\prop\ if all flags $\Mem_{i,j,\prop}$ are to be reset.

We now briefly introduce the basic flow of the algorithm once it
     stabilizes, i.e., once all $n-f$ non-faulty nodes are
     well-synchronized.
Recall that the remaining up to $f< n/3$ faulty nodes may produce
     arbitrary signals on their outgoing channels.
A pulse is locally triggered by switching to state \acc.
Thus, assume that at some time all non-faulty nodes switch to state
     \acc\ within a time window of $2d$, i.e., a pulses are generated
     by non-faulty nodes within a time interval of size $2d$.
Supposing that $T_1\geq 3\vartheta d$, these nodes will observe, and
     thus memorize, each other and themselves in state \acc\ within a
     time interval of size $3d$ and thus before $T_1$ expires at any
     non-faulty node.
This makes timeout $T_1$ the critical condition for switching to state
     \slp.
From state \slp, they will switch to states \srw, \wake, and finally
     \rdy, where the timeout $(T_2,\acc)$ is determining the time this
     takes, as it is considerably larger than
     $\vartheta(2\vartheta+2)T_1$.
The intermediate states serve the purpose of achieving stabilization,
     hence we leave them out for the moment.

Note that upon
switching to state \rdy, nodes reset their \prop\ flags and $\Next_i$. Thus,
they essentially ignore these signals between the most recent time they switched
to \prop\ before switching to \acc\ and the subsequent time when they switch to
\rdy. Since nodes already reset their \acc flags upon switching to \wake, this
ensures that nodes do not take into account outdated information for the
decision when to switch to state \prop. 

Hence, it is guaranteed that the first node switching from state \rdy\
     to state \prop\ again does so because $T_4$ expired or because
     $T_3$ expired and its \Next memory flag is true.
The constraint $\min\{T_3,T_4\}\geq \vartheta (T_2+4d)$ ensures that
     all non-faulty nodes observe themselves in state \rdy\ before the
     first one switches to \prop.
Hence, no node deletes information about nodes that switch to \prop\
     again after the previous pulse.

The first non-faulty node that switches to state \acc\ again cannot do
     so before it memorizes at least $n-f$ nodes in state \prop, as
     the \acc\ flags have been reset upon switching to state \wake.
Therefore, at this time at least $n-2f\geq f+1$ non-faulty nodes are
     in state \prop.
Hence, the rule that nodes switch to \prop\ if they memorize $f+1$
     nodes in states \prop\ will take effect, i.e., the remaining
     non-faulty nodes in state \rdy\ switch to \prop\ after less than
     $d$ time.
Another $d$ time later all non-faulty nodes in state \prop\ will have
     become aware of this and switch to state \acc\ as well, as the
     threshold of $n-f$ nodes in states \prop\ or \acc\ is reached.
Thus the cycle is complete and the reasoning can be repeated
     inductively.

Clearly, for this line of argumentation to be valid, the algorithm
     could be simpler than stated in \figureref{fig:main_simple}.
We already mentioned that the motivation of having three intermediate
     states between \acc\ and \rdy\ is to facilitate stabilization.
Similarly, there is no need to make use of the \acc\ flags in the
     basic cycle at all; in fact, it adversely affects the constraints
     the timeouts need to satisfy for the above reasoning to be valid.
However, the \acc\ flags are much better suited for diagnostic
     purposes than the \prop\ flags, since nodes are expected to
     switch to \acc\ in a small time window and remain in state \acc\
     for a small period of time only (for all our results, it is
     sufficient if $T_1=4\vartheta d$).
Moreover, two different timeout conditions for switching from \rdy\ to
     \prop\ are unnecessary for correct operation of the pulse
     synchronization routine.
As discussed before, they are introduced in order to allow for a
     seamless coupling to the application layer.

\subsection{Main Algorithm}


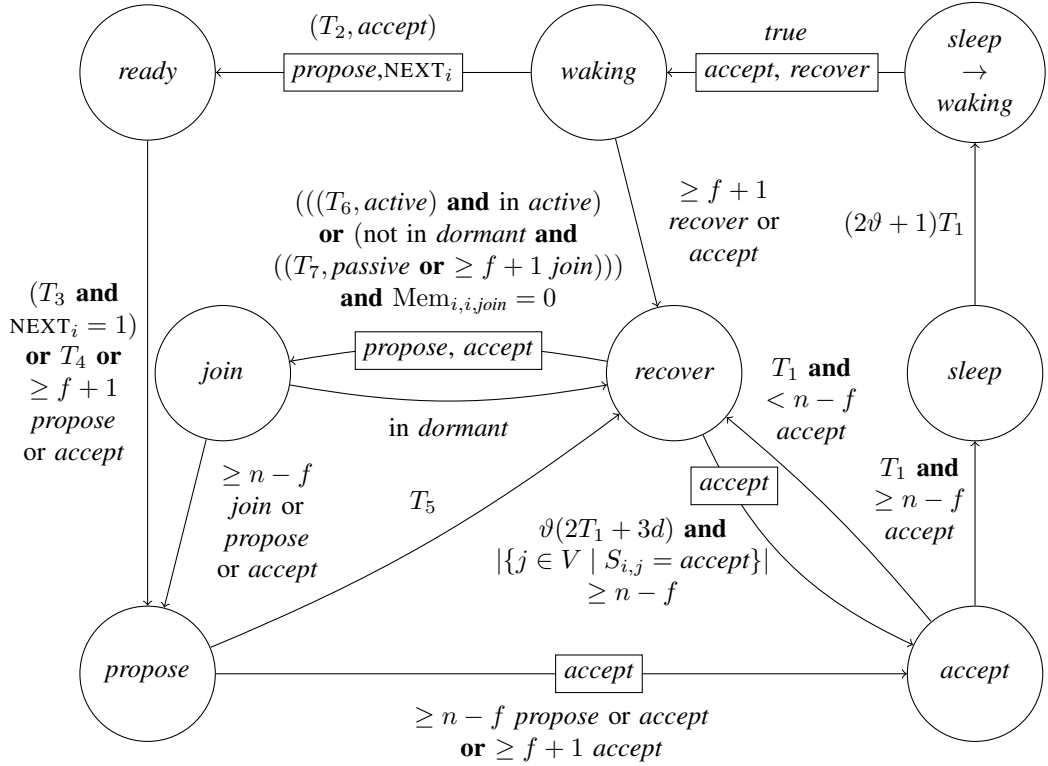
\begin{figure*}[t!]
  \centering
  \begin{tikzpicture}
      \draw (11,0) node[circle,minimum width=1.8cm,draw,align=center] (ac) {$\acc$};
      \draw (11,4) node[circle,minimum width=1.8cm,draw,align=center] (sl) {$\slp$};
      \draw (11,8) node[circle,minimum width=1.8cm,draw,align=center] (srw) {$\slp$\\$\rightarrow$\\$\wake$};
      \draw (6,8) node[circle,minimum width=1.8cm,draw] (wa) {$\wake$};
      \draw (0,8) node[circle,minimum width=1.8cm,draw] (rd) {$\rdy$};
      \draw (0,0) node[circle,minimum width=1.8cm,draw] (pr) {$\prop$};

      \draw (sl)+(-4,0) node[circle,minimum width=1.8cm,draw] (re) {$\rec$};
      \draw (re)+(-6,0) node[circle,minimum width=1.8cm,draw] (jo) {$\join$};

      \path[->] (ac) edge node[left,align=center,yshift=0.3cm] {$T_1$ \textbf{and}\\$\ge n-f$\\$\acc$} (sl);
      \path[->] (sl) edge node[left,align=center] {$(2\vartheta+1)T_1$} (srw);
      \path[->] (srw) edge node[above=0.3cm,align=center] {\emph{true}} (wa);
      \path[->] (wa) edge node[above=0.3cm,align=center] {$(T_2,\acc)$} (rd);
      \path[->] (rd) edge node[left,align=center] {$(T_3 \textbf{ and}$\\$\Next_i=1)$\\\textbf{or} $T_4$ \textbf{or}\\$\ge f+1$\\$\prop$\\or $\acc$} (pr);
      \path[->] (pr) edge node[below=0.3cm,align=center] {$\ge n-f$ $\prop \text{ or } \acc$\\\textbf{or} $\ge f+1$ $\acc$} (ac);

      \path[->] (ac) edge[bend right=5] node[above,xshift=-0.3cm,yshift=0.8cm,align=center] {$T_1$ \textbf{and}\\$< n-f$\\$\acc$} (re);
\path[->] (re) edge[bend right=20] node[left=0.1cm,align=center] {$\vartheta(2T_1+3d)$ \textbf{and}\\$\lvert\{j\in V \mid S_{i,j}=\acc\}\rvert$\\$\ge n-f$} (ac);

      \path[->] (re) edge[bend right=10] node[above=0.3cm,align=center]
      {$(((T_6,\act)$ \textbf{and} in $\act)$\\
      \textbf{or} $($not in $\dorm$ \textbf{and}\\
      $((T_7,\pass$ \textbf{or} $\ge f+1$ $\join)))$\\
      \textbf{and} $\Mem_{i,i,\join}=0$} (jo); \path[->] (jo) 
      edge[bend right=10] node[below=0.1cm,align=center] {in $\dorm$} (re);
\path[->] (pr) edge[bend right=7] node[above=0.3cm,align=center] {$T_5$} (re);

      \path[->] (jo) edge node[right=0.3cm,align=center] {$\ge n-f$\\$\join$ or\\$\prop$\\ or $\acc$} (pr);
      \path[->] (wa) edge node[right=0.3cm,align=center] {$\ge f+1$\\$\rec$ or\\$\acc$} (re);

      \draw (srw)+(-2.5,0) node[fill=white,draw=black] {$\acc$, $\rec$};
      \draw (wa)+(-3,0) node[fill=white,draw=black,align=center] {$\prop$,$\Next_i$};
      \draw (pr)+(6,0) node[fill=white,draw=black] {$\acc$};
      \draw (re)+(0.8,-1.5) node[fill=white,draw=black] {$\acc$};
      \draw (re)+(-3,0.3) node[fill=white,draw=black] {$\prop$, $\acc$};

  \end{tikzpicture}
  \caption{Overview of the core routine of node $i$'s self-stabilizing pulse
algorithm.}\label{fig:main}
\end{figure*}

We proceed by describing the main routine of the pulse algorithm in full.
Alongside the main routine, several other state machines run concurrently and
provide additional information to be used during recovery, as we detail later.

The main routine is graphically presented in \figureref{fig:main}. Except for
the states \rec\ and \join\ and additional resets of memory flags, the main
routine is identical to the basic cycle. The purpose of the two additional
states is the following: Nodes switch to state \rec once they detect that
something is wrong, that is, non-faulty nodes do not execute the basic cycle as
outlined in \sectionref{sec:basic_cycle}. This way, non-faulty nodes will not
continue to confuse others by sending for example state signals \prop or \acc
despite clearly being out-of-sync. There are various consistency checks that
nodes perform during each execution of the basic cycle. The first one is that in
order to switch from state \acc to state \slp, non-faulty nodes need to memorize
at least $n-f$ nodes in state \acc. If this does not happen within $4d\leq
T_1/\vartheta$ time after switching to state \acc, by the arguments given in
\sectionref{sec:basic_cycle}, the nodes could not have entered state \acc\
within $2d$ of each other. Therefore, something must be wrong and it is feasible
to switch to state \rec. Next, whenever a non-faulty node is in state \wake,
there should be no non-faulty nodes in states \acc or \rec. Considering that
the node resets its \acc and \rec flags upon switching to \wake, it should not
memorize $f+1$ or more nodes in states \acc or \rec at a time when it observes
itself in state \wake. If it does, however, it again switches to state \rec.
Last but not least, during a synchronized execution of the basic cycle, no
non-faulty node may be in state \prop\ for more than a certain amount of time
before switching to state \acc. Therefore, nodes will switch from \prop\ to
\rec\ when timeout $T_5$ expires.

There are two different ways for nodes in \rec to switch back to the basic
cycle, corresponding to two different mechanisms for stabilization. The
transition from \rec to \acc requires to (directly) observe $n-f$ nodes in state
\acc. This enables nodes to resynchronize provided that at least $n-f$ nodes are
already executing the basic cycle in synchrony. While this method is easily
implemented, clearly it is insufficient to ensure stabilization from arbitrary
initial configurations. Hence, nodes can also join the basic cycle again via the
second new state, called \join. Since the Byzantine nodes may ``play nice''
towards $n-2f$ or more nodes still executing the basic cycle, making them
believe that system operation continues as usual, it must be possible to do so
without having a majority of nodes in state \rec. On the other hand, it is
crucial that this happens in a sufficiently well-synchronized manner, as
otherwise nodes could drop out of the basic cycle again because the various
checks of consistency detect an erroneous execution of the basic cycle.

In part, this issue is solved by an additional agreement step. In order to enter
the basic cycle again, nodes need to memorize $n-f$ nodes in states \join (the
respective nodes detected an inconsistency), \prop (these nodes continued to
execute the basic cycle), or \acc (there are executions where nodes reset their
\prop flags because of switching to \join\ when other nodes already switched to
\acc). The threshold conditions of $f+1$ nodes memorized in state \join or $f+1$
nodes memorized in state \prop for leaving state \rec, all nodes will follow the
first one switching from \join to \prop quickly, just as with the switch from
\prop to \acc in an ordinary execution of the basic cycle. However, it is
decisive that all nodes are in states that permit to participate in this
agreement step in order to guarantee success of this approach.

As a result, still a certain degree of synchronization needs to be established
beforehand,\footnote{This is the reason for the complicated transition condition
involving additional states and timeouts. The detailed interplay between these
conditions is delicate and beyond the scope of a high-level description of the
algorithm; the interested reader is referred to the analysis section.} both
among nodes that still execute the basic cycle and those that do not. For
instance, if at the point in time when a majority of nodes and channels become
non-faulty, some nodes already memorize nodes in \join that are not, they may
switch to state \join and subsequently \prop prematurely, causing others to
have inconsistent memory flags as well. Byzantine faults may sustain such
amiss configuration of the system indefinitely.

So why did we put so much effort in ``shifting'' the focus to this part of the
algorithm? The key advantage is that nodes outside the basic cycle may take into
account less reliable information for stabilization purposes. They may take the
risk of metastable upsets (as we know it is impossible to avoid these during the
stabilization process, anyway) and make use of randomization.

In fact, to make the above scheme work, it is sufficient that all non-faulty
nodes agree on a so-called \emph{resynchronization point} (cf.~Definitions
\ref{def:res} and~\ref{def:gres}), that is, a point in time at which nodes reset
the memory flags for states \join and \srw as well as certain timeouts, while
guaranteeing that no node is in these states close to the respective reset
times. Except for state \srw, all of these timeouts, memory flags, etc.\ are not
part of the basic cycle at all, thus nodes may enforce consistent values for
them easily when agreeing on such a resynchronization point.

Conveniently, the use of randomization also ensures that it is quite unlikely
that nodes are in state \srw close to a resynchronization point, as the
consistency check of having to memorize $n-f$ nodes in state \acc in order to
switch to state \slp guarantees that the time windows during which non-faulty
nodes may switch to \slp make up a small fraction of all times only.

Consequently, the remaining components of the algorithm deal with agreeing on
resynchronization points and utilizing this information in an appropriate way to
ensure stabilization of the main routine. We describe this connection to the
main routine first. It is done by another, quite simple state machine, which
runs in parallel alongside the core routine. It is depicted in
\figureref{fig:extended}.


\begin{figure}[t!]
  \centering
  \begin{tikzpicture}
      \draw (0,0) node[circle,minimum width=1.8cm,draw,align=center] (do) {$\dorm$};
      \draw (6,0) node[circle,minimum width=1.8cm,draw,align=center] (pa) {$\pass$};
      \draw (3,-3) node[circle,minimum width=1.8cm,draw,align=center] (ac) {$\act$};

      \path[->] (do) edge node[below=0.3cm,align=center] {in $\res$} (pa);
      \path[->] (pa) edge[bend right] node[above,align=center] {not in $\res$} (do);
      \path[->] (ac) edge node[below=0.2cm,left=0.1cm,align=center] {not in\\$\res$} (do);
      \path[->] (pa) edge node[below=0.1cm,right=0.1cm,align=center] {$\ge f+1$\\$\srw$} (ac);

      \draw (do)+(3,0) node[fill=white,draw=black] {$\join$, $\srw$};
  \end{tikzpicture}
  \caption{Extension of node~$i$'s core routine.}\label{fig:extended}
\end{figure}
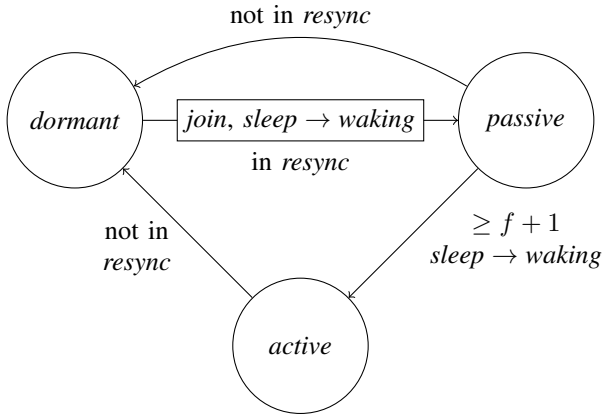

Its purpose is to reset memory flags in a consistent way and to determine when a
node is permitted to switch to \join. In general, a resynchronization point
(locally observed by switching to state \res, which is introduced later)
triggers the reset of the \join and \srw flags. If there are still nodes
executing the basic cycle, a node may become aware of it by observing $f+1$
nodes in state \srw at some time. In this case it switches from the state
\pass, which it entered at the point in time when it locally observed the
resynchronization point, to the state \act. Subsequently, once timeout $T_8$
expires, the node will switch to state, in which it is more susceptive to
switching to state \join. This is expressed by the rather involved transition
rule $tr(\rec,\join)$ (in \figureref{fig:main}). $T_6$ is much smaller than $T_7$, but $T_6$ is of no
concern until the node switches to state \act\ and resets $T_6$.\footnote{The
conditions ``in \act'' and ``not in \dorm'', respectively, here ensure that the
transition is not performed because the node has been in state \res\ a long time
ago, but there was no recent switching to \res.} The condition that
$\Mem_{i,i,\join}=0$ simply means that nodes should not already have attempted
to stabilize by switching to \join since the most recent transition to \pass.
This avoids interfering too much with the second stabilization mechanism
(switching from \rec to \acc), as it might take significantly longer than
the time required for this ``immediate'' recovery to stabilize by means of
agreeing on a resynchronization point.

It remains to explain how resynchronization points are generated.

\subsection{Resynchronization Algorithm}

The resynchronization routine is specified in \figureref{fig:resync}. Similarly
to the extension of the core routine, it is a lower layer that the core routine
uses for stabilization purposes only. It provides some synchronization that is
akin to that of a pulse, except that such ``weak pulses'' occur at random times,
and may be generated inconsistently even after the algorithm as a whole has
stabilized. Since the main routine operates independently of the
resynchronization routine once the system has stabilized, we can afford the
weaker guarantees of the routine: If it succeeds in generating a ``good''
resynchronization point merely once, the main routine will stabilize
deterministically.

\begin{definition}[Resynchronization Points]\label{def:res}
Given $W\subseteq V$, time $t$ is a \emph{$W$-resynchro\-nization point} iff
each node in $W$ switches to state \srr\ in the time interval $(t,t+2d)$.
\end{definition}

\begin{definition}[Good Resynchronization Points]\label{def:gres}\ \\
A $W$-resynchronization point is called \emph{good} iff no node from $W$
switches to state \slp\ during $(t-\Delta_g,t)$, where
$\Delta_g:=(2\vartheta+3)T_1$, and no node is in state \join\ during
$[t-T_1-d,t+4d)$.
\end{definition}

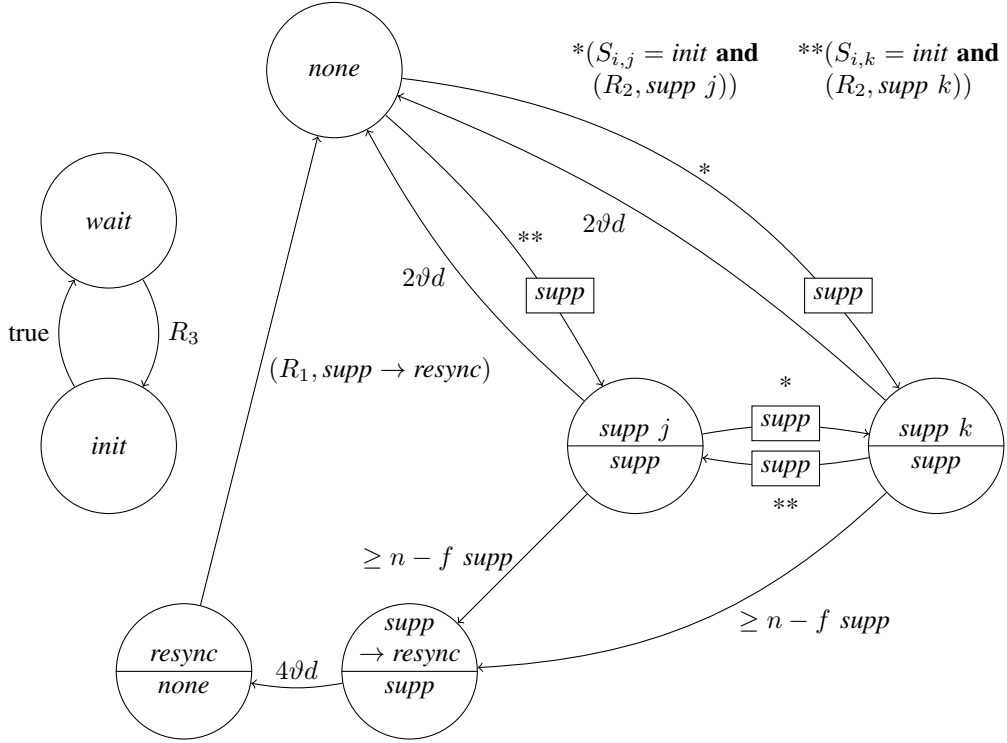
\begin{figure*}[t!]
  \centering
  \begin{tikzpicture}
      \draw (0,0) node[circle,minimum width=1.8cm,draw,align=center] (in) {$\init$};
      \draw (0,3) node[circle,minimum width=1.8cm,draw,align=center] (wa) {$\wait$};

      \path[->] (in) edge[bend left] node[left,align=center] {true} (wa);
      \path[->] (wa) edge[bend left] node[right,align=center] {$R_3$} (in);

      \draw (3,5) node[circle,minimum width=1.8cm,draw,align=center] (no) {$\none$};
      \draw (1,-3) node[circle,minimum width=1.8cm,draw,align=center] (rn) {$\res$\\$\none$};
      \draw (4,-3) node[minimum width=1.8cm,align=center] {$\supp$\\$\rightarrow\res$\\$\supp$\\~};
      \draw (4,-3) node[circle,minimum width=1.8cm,draw,align=center] (sr) {};
      \draw (7,0) node[circle,minimum width=1.8cm,draw,align=center] (sj) {$\supp$ $j$\\$\supp$};
      \draw (11,0) node[circle,minimum width=1.8cm,draw,align=center] (sk) {$\supp$ $k$\\$\supp$};

      \path[->] (rn) edge node[right,align=center] {$(R_1,\srr)$} (no);
      \path[->] (sr) edge[bend left=10] node[above,align=center] {$4\vartheta d$} (rn);
      \path[->] (sj) edge node[left,align=center] {$\ge n-f$ $\supp$} (sr);
      \path[->] (sk) edge[bend left=20] node[right=0.2cm,xshift=0.2cm,align=center] {$\ge n-f$ $\supp$} (sr);

      \path[->] (sj) edge[bend left=10] node[left,align=center,xshift=-0.1cm] {$2\vartheta d$} (no);
      \path[->] (no) edge[bend left=10] node[right,align=center] {**} (sj);
      \draw (no)+(3,-3) node[fill=white,draw=black] {$\supp$};

      \path[->] (sk) edge[bend right=10] node[left,align=center,xshift=-0.3cm] {$2\vartheta d$} (no);
      \path[->] (no) edge[bend left=25] node[right,align=center] {*} (sk);
      \draw (no)+(6.7,-3) node[fill=white,draw=black] {$\supp$};

      \path[->] (sj) edge[bend left=10] node[above=0.3cm,align=center] {*} (sk);
      \path[->] (sk) edge[bend left=10] node[below=0.3cm,align=center] {**} (sj);
      \draw (sj)+(2,0.3) node[fill=white,draw=black] {$\supp$};
      \draw (sj)+(2,-0.3) node[fill=white,draw=black] {$\supp$};

      \draw (sj.west)--(sj.east);
      \draw (sk.west)--(sk.east);
      \draw (sr.west)--(sr.east);
      \draw (rn.west)--(rn.east);

      \draw (no)+(4.4,0) node[align=center] {*$(S_{i,j} = \init$ \textbf{and}\\$(R_2,\supp\ j))$};
      \draw (no)+(7.5,0) node[align=center] {**$(S_{i,k} = \init$ \textbf{and}\\$(R_2,\supp\ k))$};
  \end{tikzpicture}
  \caption{Resynchronization algorithm, comprising two state machines
executed in parallel at node~$i$.}\label{fig:resync}
\end{figure*}

In order to clarify that despite having a linear number of states
($\supp_1,\ldots,\supp_n$), this part of the algorithm can be implemented using
$2$-bit communication channels between state machines only, we generalize our
description of state machines as follows. If a state is depicted as a circle
separated into an upper and a lower part, the upper part denotes the local
state, while the lower part indicates the signal state to which it is mapped. A
node's memory flags then store the respective signal states only, i.e., remote
nodes do not distinguish between states that share the same signal. Clearly,
such a machine can be simulated by a machine as introduced in the model section.
The advantage is that such a mapping can be used to reduce the number of
transmitted state bits; for the resynchronization routine given in
\figureref{fig:resync}, we merely need two bits (\init/\emph{wait} and
\none/\supp) instead of $\lceil \log (n+3)\rceil+1$ bits.

The basic idea behind the resynchronization algorithm is the following: Every
now and then, nodes will try to initiate agreement on a resynchronization point.
This is the purpose of the small state machine on the left in
\figureref{fig:resync}. Recalling that the transition condition ``true'' simply
means that the node switches to state \emph{wait} again as soon as it observes
itself in state \init, it is easy to see that it does nothing else than creating
an \init\ signal as soon as $R_3$ expires and resetting $R_3$ again as quickly
as possible. As the time when a node switches to \init\ is determined by the
randomized timeout $R_3$ distributed over a large interval (cf.\
\equalityref{eq:R_3}) only, therefore it is impossible to predict when it will expire,
even with full knowledge of the execution up to the current point in time. Note
that the complete independence of this part of node $i$'s state from the
remaining protocol implies that faulty nodes are not able to influence the
respective times by any means.

Consider now the state machine displayed on the right of \figureref{fig:resync}.
To illustrate how the routine is intended to work, assume that at the time $t$
when a non-faulty node $i$ switches to state \init, all non-faulty nodes are not
in any of the states \srr, \res, or \supp~$i$, and at all non-faulty nodes the
timeout $(R_2,\supp~i)$ has expired. Then, no matter what the signals from
faulty nodes or on faulty channels are, each non-faulty node will be in one of
the states \supp~$j$, $j\in V$, or \srr\ at time $t+d$. Hence, they will observe
each other (and themselves) in one of these states at some time smaller than
$t+2d$. These statements follow from the various timeout conditions of at least
$2\vartheta d$ and the fact that observing node $i$ in state \init\ will make
nodes switch to state \supp~$i$ if in \none\ or \supp~$j$, $j\neq i$. Hence, all
of them will switch to state \srr\ during $(t,t+2d)$, i.e., $t$ is a
resynchronization point. Since $t$ follows a random distribution that is
independent of the remaining algorithm and, as mentioned earlier, most of the
times nodes do not switch to state \slp\ and it is easy to deal with the
condition on \join\ states, there is a large probability that $t$ is a good
resynchronization point. Note that timeout $R_1$ makes sure that no non-faulty
node will switch to \srr\ again anytime soon, leaving sufficient time for the
main routine to stabilize.

The scenario we just described relies on the fact that at time $t$ no node is in
state \srr\ or state \res. We will choose $R_2\gg R_1$, implying that $R_2+3d$
time after a node switched to state \init\ all nodes have ``forgotten'' about
this, i.e., $(R_2,\supp~i)$ is expired and they switched back to state \none\
(unless other \init\ signals interfered). Thus, in the absence of Byzantine
faults, the above requirement is easily achieved with a large probability by
choosing $R_3$ as a uniform distribution over some interval
$[R_2+3d,R_2+\Theta(n R_1)]$: Other nodes will switch to \init\ $\BO(n)$ times
during this interval, each time ``blocking'' other nodes for at most $\BO(R_1)$
time. If the random choice picks any other point in time during this interval, a
resynchronization point occurs. Even if the clock speed of the clock driving
$R_3$ is manipulated in a worst-case manner (affecting the density of the
probability distribution with respect to real time by a factor of at most
$\vartheta$), we can just increase the size of the interval to account for this.

However, what happens if only \emph{some} of the nodes receive an \init\ signal
due to faulty channels or nodes? If the same holds for some of the subsequent
\supp\ signals, it might happen that only a fraction of the nodes reaches the
threshold for switching to state \srr, resulting in an inconsistent reset of
flags and timeouts across the system. Until the respective nodes switch to
state \none\ again, they will not support a resynchronization point again, i.e.,
about $R_1$ time is ``lost''. This issue is the reason for the agreement step
and the timeouts $(R_2,\supp~j)$. In order for any node to switch to state \srr,
there must be at least $n-2f\geq f+1$ non-faulty nodes supporting this. Hence,
all of these nodes recently switched to a state \supp~$j$ for some $j\in V$,
resetting $(R_2,\supp~j)$. Until these timeouts expire, $f+1\in \Omega(n)$
non-faulty nodes will ignore \init\ signals on the respective channels. Since
there are $\BO(n^2)$ channels, it is possible to choose $R_2\in \BO(n R_1)$
such that this may happen at most $\BO(n)$ times in $\BO(n)$ time. Playing with
constants, we can pick $R_3\in \BO(n)$ maintaining that still a constant
fraction of the times are ``good'' in the sense that $R_3$ expiring at a
non-faulty node will result in a good resynchronization point.

\subsection{Timeout Constraints}

\conditionref{cond:timeout_bounds} summarizes the constraints we
require on the timeouts for the core routine and the
resynchronization algorithm to act and interact as intended.

\begin{condition}[Timeout Constraints]\label{cond:timeout_bounds}

Recall that $\vartheta>1$ and  $\Delta_g:=(2\vartheta+3)T_1$.
Define
\begin{equation}\label{eq:def_lambda}
\lambda:=\sqrt{\frac{25\vartheta-9}{25\vartheta}}\in\left(\frac{4}{5},1\right).
\end{equation}
The timeouts need to satisfy the constraints
\begin{eqnarray}
T_1 &\geq & 4\vartheta d\label{eq:T_1}\\
T_2 &\geq & 3\vartheta\Delta_g+7\vartheta d\label{eq:T_2}\\
&\sr{eq:T_1}{>}&
(2\vartheta^2+4\vartheta)T_1 +9\vartheta d\nonumber\\
T_3 &\geq & (2\vartheta^2+4\vartheta)T_1-T_2+\vartheta T_6+7\vartheta
d\label{eq:T_3}\\
&\stackrel{(\ref{eq:T_1},\ref{eq:T_6})}{>}&
(\vartheta-1)T_2+6\vartheta d\nonumber\\
T_4 &\geq & T_3\label{eq:T_4}\\
T_5 &\geq & \max\{(\vartheta-1)T_2-T_3+\vartheta T_4+7\vartheta
d,\label{eq:T_5}\\
&&(\vartheta-1)T_1+\vartheta(T_2+T_4)-T_6\}\nonumber\\
T_6 &\geq & \vartheta T_2-2\vartheta T_1-2\vartheta d\label{eq:T_6}\\
&\sr{eq:T_2}{>}&
(2\vartheta^2+3\vartheta-1)T_1+7\vartheta d \nonumber\\
T_7 &\geq &(2\vartheta-1)T_1+\vartheta(T_2+T_4+T_5)+T_6\label{eq:T_7}\\
&\sr{eq:T_6}{>}&(2\vartheta^2+3\vartheta-2)T_1+\vartheta(T_2+T_4+T_5+
3d)\nonumber\\
R_1 &\geq & \max\{\vartheta T_7+(4\vartheta^2+8\vartheta)d,\nonumber\\
&&\vartheta(2T_2+2T_4+T_5+7d)-2T_1\}\label{eq:R_1}\\
R_2 &\geq & \frac{2\vartheta \theterm(n-f)}{1-\lambda}~~~~~
\label{eq:R_2}\\
R_3 &= &\mbox{uniformly distributed random variable on}\notag\\
&&\left[\vartheta (R_2+3d), \vartheta(R_2+3d)
+8(1-\lambda)R_2\right]\label{eq:R_3}\\
\lambda&\leq &\frac{T_2-2\vartheta \Delta_g-(\vartheta-1)T_1-4\vartheta d}
{T_2-(\vartheta-1)T_1-\vartheta d}.\label{eq:lambda}
\end{eqnarray}
\end{condition}

We need to show that this system can always be solved.
Furthermore, we would like to allow to couple the pulse generation
     algorithm to an application algorithm with any possible drift.
To this end, we would like to be able to make the ratio
     $(T_2+T_4)/(\vartheta(T_2+T_3+4d))$ arbitrarily large: Thereby,
     $(T_2+T_4)$ is the minimal gap between successive pulses
     generated at each node, provided that the states of all the \Next
     signals are constantly zero, and $\vartheta(T_2+T_3+4d)$ is the
     maximal time it takes nodes to observe themselves in state \rdy\
     with $T_3$ expired after the last generated pulse (as then they
     will respond to $\Next_i$ switching to one).
\begin{lemma}\label{lemma:constraints}
For any $d,\vartheta \in \BO(1)$, \conditionref{cond:timeout_bounds} can be
satisfied with $T_1,\ldots,T_7,R_1\in \BO(1)$ and $R_2\in \BO(n)$, where the
ratio
\begin{equation*}
\alpha:=\frac{(T_2+T_4)/\vartheta}{T_2+T_3+4d}
\end{equation*}
maybe chosen to be an arbitrarily large constant.
\end{lemma}
\begin{proof}
First, observe that if \inequalityref{eq:T_2} holds, the denominator in the
right hand side of \inequalityref{eq:lambda} is positive. Thus, we can
equivalently state \inequalityref{eq:lambda} as
\begin{equation}\label{eq:T_2_merged}
T_2\geq \frac{2\vartheta \Delta_g+(1-\lambda)(\vartheta-1)T_1+(4-\lambda)d}
{1-\lambda}.
\end{equation}
Since $\lambda\in (4/5,1)$, this inequality clearly imposes a stronger
constraint than \inequalityref{eq:T_2}, hence we can replace
Inequalities~\eqref{eq:T_2} and~\eqref{eq:lambda} with this one and obtain an
equivalent system. The requirement of $(T_2+T_4)/(\vartheta(T_2+T_3+4d))=\alpha$
can be rephrased as
\begin{equation}\label{eq:T_4_ratio}
T_4\geq (\alpha \vartheta-1)T_2+\alpha \vartheta(T_3+4d).
\end{equation}
Again, clearly this constraint is stronger than \inequalityref{eq:T_4}, hence
we drop \inequalityref{eq:T_4} in favor of \inequalityref{eq:T_4_ratio}.

We satisfy the inequalities by iteratively defining the values of the left hand
sides in accordance with the respective constraint, in the order \eqref{eq:T_1},
\eqref{eq:T_2_merged}, \eqref{eq:T_6}, \eqref{eq:T_3}, \eqref{eq:T_4_ratio},
\eqref{eq:T_5}, \eqref{eq:T_7}, \eqref{eq:R_1}, and finally \eqref{eq:R_2}. Note
that this is feasible, as in any step the right hand side of the current
inequality is an expression in $d$, $\vartheta$, $\alpha$, and, in case of
\inequalityref{eq:R_2}, $n-f$.\footnote{For simplicity, we refrain from
demanding equality and drop terms in order to get more condensed expressions.
For \ $\vartheta\leq 1.2$, for example, the increase in the bounds is not significant.}
We obtain the solution
\begin{eqnarray*}
T_1 &:=& 4\vartheta d\\
T_2 &:=& \frac{46\vartheta^3 d}{1-\lambda}\\
T_6 &:=& \frac{46\vartheta^4 d}{1-\lambda}\\
T_3 &:=& \frac{(\vartheta^2-1)46\vartheta^3 d}{1-\lambda}+31\vartheta^3d\\
T_4 &:=& \frac{46\vartheta^3(\alpha \vartheta^3-1)d}{1-\lambda}+
35\alpha\vartheta^4d\\
T_5 &:=& \frac{46\vartheta^4(\alpha \vartheta^3-1)d}{1-\lambda}+
39\alpha\vartheta^5 d\\
T_7 &:=& \frac{92\alpha \vartheta^8 d}{1-\lambda}+
78\alpha\vartheta^5 d\\
R_1 &:=& \frac{46\vartheta^6(3\alpha \vartheta^3-1)d}{1-\lambda}+
109\alpha\vartheta^6 d\\
R_2 &:= & \frac{(92\vartheta^7(3\alpha \vartheta^3-1)(n-f)d}{(1-\lambda)^2}\\
&&+\frac{(218\alpha\vartheta^7+108\vartheta^2))(n-f)d}
{1-\lambda}\,.
\end{eqnarray*}
As $\alpha\in \BO(1)$ was arbitrary, $d$ and $\vartheta$ are constants, and
$\lambda\in (4/5,1)$ depends on $\vartheta$ only and is thus a constant as well,
these values satisfy the asymptotic bounds stated in the lemma, concluding the
proof.
\end{proof}

%% file: analysis.tex
\section{Analysis}\label{sec:analysis}

In this section we derive skew bounds $\Sigma$, as well as accuracy bounds
$T^-,T^+$, such that the presented protocol is a $(W,E)$-stabilizing pulse
synchronization protocol, for proper choices of the set of nodes $W$ and the set
of channels $E$, with skew $\Sigma$ and accuracy bounds $T^-,T^+$ that
stabilizes within time $T(k) \in \BO(kn)$ with probability $1-2^{-k(n-f)}$, for
any $k\in \N$. This analysis follows the lines of \cite{DFLS11:TR}, with minor
adjustments due to the changes made to the FATAL protocol. Moreover, we show
that if a set of at least $n-f$ nodes fires pulses regularly, then other
non-faulty nodes synchronize within $\BO(R_1)$ time deterministically. This
stabilization mechanism is much simpler; the main challenge here is to avoid
interference with the other approach.

\subsection{Basic Statements}

To start our analysis, we need to define the basic requirements for
stabilization. Essentially, we need that a majority of nodes is non-faulty and
the channels between them are correct. However, the first part of the
stabilization process is simply that nodes ``forget'' about past events that are
captured by their timeouts. Therefore, we demand that these nodes indeed have
been non-faulty for a time period that is sufficiently large to ensure that all
timeouts have been reset at least once after the considered set of nodes became
non-faulty.
\begin{definition}[Coherent Nodes]
A subset of nodes $W\subseteq V$ is called \emph{coherent} during the time
interval $[t^-,t^+]$, iff during $[t^--(\vartheta(R_2+3d)
+8(1-\lambda)R_2)-d,t^+]$ all nodes
$i\in W$ are non-faulty, and all channels $S_{i,j}$, $i,j\in W$, are correct.
\end{definition}
We will show that if a coherent set of at least $n-f$ nodes fires a pulse, i.e.,
switches to \acc\ in a tight synchrony, this set will generate pulses
deterministically and with controlled frequency, as long the set remains
coherent. This motivates the following definitions.
\begin{definition}[Stabilization Points]
We call time $t$ a \emph{$W$-stabilization point (quasi-stabi\-lization point)}
iff all nodes $i\in W$ switch to \acc\ during $[t,t+2d)$ $([t,t+3d))$.
\end{definition}

\textit{Throughout this section, we assume the set of coherent nodes $W$ with
$|W|\geq n-f$ to be fixed and consider all nodes in and channels originating
from $V\setminus W$ as (potentially) faulty.}  As
all our statements refer to nodes in $W$, we will typically omit the word
``non-faulty'' when referring to the behavior or states of nodes in $W$, and
``all nodes'' is short for ``all nodes in $W$''. Note, however, that we will
still clearly distinguish between channels originating at faulty and non-faulty
nodes, respectively, to nodes in $W$.

As a first step, we observe that at times when $W$ is coherent, indeed all nodes
reset their timeouts, basing the respective state transition on proper
perception of nodes in~$W$.
\begin{lemma}\label{lemma:counters}
If $W$ is coherent during the time interval $[t^-,t^+]$, with $t^- \ge \vartheta(R_2+3d)
+8(1-\lambda)R_2+d$, any (randomized)
timeout $(T,s)$ of any node $i\in W$ expiring at a time $t\in[t^-,t^+]$ has been
reset at least once since time $t^--(\vartheta(R_2+3d)
+8(1-\lambda)R_2)$. If $t'$ denotes the time
when such a reset occurred, for any $j\in W$ it holds that
$S_{i,j}(t')=S_j(\tau_{j,i}^{-1}(t'))$, i.e., at time $t'$, $i$ observes $j$ in a
state $j$ attained when it was non-faulty.
\end{lemma}
\begin{proof}
According to \conditionref{cond:timeout_bounds}, the largest possible value of
any (randomized) timeout is $\vartheta(R_2+3d)
+8(1-\lambda)R_2$. Hence, any timeout that is in state~$1$ at
a time smaller than $t^--(\vartheta(R_2+3d)
+8(1-\lambda)R_2) \ge 0$ expires before time $t^-$ or is reset at least
once. As by the definition of coherency all nodes in $W$ are non-faulty and all
channels between such nodes are correct during $[t^--(\vartheta(R_2+3d)
+8(1-\lambda)R_2),t^+]$, this
implies the statement of the lemma.
\end{proof}

Phrased informally, any corruption of timeout and channel states eventually
ceases, as correct timeouts expire and correct links remember no events that lie
$d$ or more time in the past. Proper cleaning of the memory flags is more
complicated and will be explained further down the road. 

\textit{Throughout this
section, we will assume for the sake of simplicity that the set $W$ is coherent
at all times} and use this lemma implicitly, e.g.\ we will always assume that
nodes from $W$ will observe all other nodes from $W$ in states that they indeed
had less than $d$ time ago, expiring of randomized timeouts at non-faulty nodes
cannot be predicted accurately, etc. We will discuss more general settings in
\sectionref{sec:generalizations}.

We proceed by showing that once all nodes in $W$ switch to \acc\ in a short
period of time, i.e., a $W$-quasi-stabilization point is reached, the algorithm
guarantees that synchronized pulses are generated deterministically with a
frequency that is bounded both from above and below.

\begin{theorem}\label{theorem:stability}
Suppose $t$ is a $W$-quasi-stabilization point. Then
\begin{itemize}
  \item [(i)] all nodes in $W$ switch to \acc\ exactly once within
  $[t,t+3d)$, and do not leave \acc\ until $t+4d$; and
  \item [(ii)] there will be a $W$-stabilization point
  $t'\in(t+(T_2+T_3)/\vartheta,t+T_2+T_4+5d)$ satisfying that no
  node in $W$ switches to \acc\ in the time interval $[t+3d,t')$; and that
  \item [(iii)] each node $i$'s, $i\in W$, core state machine
  (\figureref{fig:main_simple}) is metastability-free during $[t+3d,t'+3d]$.
\end{itemize}
\end{theorem}
\begin{proof}
Proof of (i): Due to \inequalityref{eq:T_1}, a node does not leave the state
\acc\ earlier than $T_1/\vartheta \geq 4d$ time after switching to it.
Thus, no node can switch to \acc\ twice during $[t,t+3d)$. By definition of a
quasi-stabilization point, every node does switch to \acc\ in the interval
$[t,t+3d)\subset [t,t+T_1/\vartheta)$. This proves Statement~(i).

\medskip

Proof of (ii): For each $i\in W$, let $t_i \in [t,t+3d)$ be the time when $i$
switches to \acc. By (i) $t_i$ is well-defined. Further let $t'_i$ be the
infimum of times in $(t_i,\infty)$ when $i$ switches to \rec or
\prop.\footnote{Note that we follow the convention that $\inf\emptyset = \infty$
if the infimum is taken with respect to a (from above) unbounded subset of
$\R^+_0$.} In the following, denote by $i\in W$ a node with minimal $t'_i$.

We will show that all nodes switch to \prop\ via states \slp, \srw, \wake, and
\rdy\ in the presented order. By (i) nodes do not leave \acc\ before $t+4d$.
Thus at time $t+4d$, each node in $W$ is in state \acc\ and observes each other
node in $W$ in \acc. Hence, each node in $W$ memorizes each other node in $W$ in
\acc\ at time $t+4d$. For each node $j\in W$, let $t_{j,s}$ be the time node
$j$'s timeout $T_1$ expires first after $t_j$. Then $t_{j,s} \in
(t_j+T_1/\vartheta,t_j+T_1+d)$.\footnote{The upper bound comprises an additive
term of $d$ since $T_1$ is reset at some time from $(t_j,t_j+d)$.} Since  $|W|
\geq n-f$, each node $j$ switches to state \slp\ at time $t_{j,s}$. Hence, by
time $t+T_1+4d$, no node will be observed in state \acc\ anymore (until the time
when it switches to \acc\ again).

When a node $j\in W$ switches to state \wake\ at the minimal time $t_w$ larger
than $t_j$, it does not do so earlier than at time $t+T_1/\vartheta +
(2+1/\vartheta)T_1 = t+(2+2/\vartheta)T_1 > t+T_1+4d$. This implies that all
nodes in $W$ have already left \acc\ at least $d$ time ago, since they switched
to it at their respective times $t_j<t+T_1+3d$. Moreover, they cannot switch to
\acc\ again until $t_i'$ as it is minimal and nodes need to switch to \prop\
or \rec before switching to \acc. Hence, nodes in $W$ are not
observed in state \acc\ during $(t+T_1+4d,t_i']$, in particular not by node
$j$. Furthermore, nodes in $W$ are not observed in state \rec\ during
$(t_w-d,t_i']$. As it resets its \acc\ and \rec\ flags upon switching to \wake,
$j$ will hence neither switch from \wake\ to \rec\ nor from \rdy to
\prop during $(t_w,t_i')$.

Now consider node $i$. By the previous observation, it will not switch from
\wake\ to \rec, but to \rdy, following the basic cycle. Consequently, it must
wait for timeout $T_2$ to expire, i.e., cannot switch to \rdy\ earlier than at
time $t+T_2/\vartheta$. By definition of $t_i'$, node $i$ thus switches to \prop
at time $t_i'$. As it is the first node that does so, this cannot happen before
timeouts $T_3$ or $T_4$ expire, i.e., before time
\begin{equation}
t+\frac{T_2}{\vartheta} + \frac{\min\{T_3,T_4\}}{\vartheta} \sr{eq:T_4}{=}
t+\frac{T_2+T_3}{\vartheta} \sr{eq:T_3}{>} t+T_2+5d.\label{eq:low}
\end{equation}

All other nodes in $W$ will switch to \wake, and for the first time after $t_j$,
observe themselves in state \wake\ at a time within
$(t+T_1+4d,t+T_1(2+\vartheta)+7d)$. Recall that unless they memorize at least
$f+1$ nodes in \acc\ or \rec\ while being in state \wake, they will all switch
to state \rdy\ by time
\begin{equation}
\max\{t+T_2+4d,t+(2\vartheta+2)T_1+7d\}\sr{eq:T_2}{=}t+T_2+4d.\label{eq:up}
\end{equation}
As we just showed that $t_i'>t+T_2+5d$, this implies that at time $t+T_2+5d$ all
nodes are observed in state \rdy, and none of them leaves before time $t_i'$.

Now choose $t'$ to be the infimum of times from
$(t+(T_2+T_3)/\vartheta,t+T_2+T_4+4d]$ when a node in $W$ switches to state
\acc.\footnote{Note that since we take the infimum on
$(t+(T_2+T_3)/\vartheta,t+T_2+T_4+4d]$, we have that $t'\leq t+T_2+T_4+4d$.}
Because of \inequalityref{eq:low}, node $j$ cannot switch to \prop\ within
$[t_j,t+(T_2+T_3)/\vartheta)$. Thus, (after time $t+3d$) node $j$ does not switch to
\acc again earlier than time $t'$, and timeout $T_5$ cannot expire at $j$ until
time
\begin{equation}
t+\frac{T_2+T_3+T_5}{\vartheta}\sr{eq:T_5}{\geq}t+T_2+T_4+7d
\geq t'+3d,\label{eq:acc_on_time}
\end{equation}
making it impossible for $j$ to switch from \prop\ to \rec\ at a time
within $[t_j,t'+3d]$. What is more, a node from $W$ that switches to \acc\ must
stay there for at least $T_1/\vartheta>3d$ time. Thus, by definition of $t'$, no
node $j\in W$ can switch from \acc\ to \rec\ at a time within $[t_j,t'+3d]$.
Hence, no node $j\in W$ can switch to state \rec\ after $t_j$, but earlier than
time $t'+2d$. It follows that no node in $W$ can switch to other states than
\prop or \acc during $[t+T_2+4d,t'+2d]$. In particular, no node in $W$ resets
its \prop flags during $[t+T_2+5d,t'+2d]\supset [t_i',t'+2d]$.

If at time $t'$ a node in $W$ switches to state \acc, $n-2f\geq f+1$ of its
\prop\ flags corresponding to nodes in $W$ are true, i.e., all correspond to a
flag holding $1$. As the node reset its \prop\ flags at the most recent time
when it switched to \rdy\ and no nodes from $W$ have been observed in \prop\
between this time and $t_i'$, it holds that $f+1$ nodes in $W$ switched to state
\prop\ during $[t_i',t')$. Since we established that no node resets its \prop\
flags during $[t_i',t'+2d]$, it follows that all nodes are in state \prop\ by
time $t'+d$. Consequently, all nodes in $W$ will observe all nodes in $W$ in
state \prop\ before time $t'+2d$ and switch to \acc, i.e., $t'\in
(t+(T_2+T_3)/\vartheta,t+T_2+T_4+4d)$ is a W-stabilization point. Statement (ii)
follows.

On the other hand, if at time $t'$ no node in $W$ switches to state \acc, it
follows that $t'=t+T_2+T_4+4d$. As all nodes observe themselves in state \rdy\
by time $t+T_2+5d$, they switch to \prop\ before time $t+T_2+T_4+5d=t'+d$
because $T_4$ expired. By the same reasoning as in the previous case, they
switch to \acc\ before time $t'+2d$, i.e., Statement (ii) holds as well.

\medskip

Proof of (iii): We have shown that within $[t_j,t'+3d]$, any node $j\in W$
switches to states along the basic cycle only. Note that Condition~(ii) in the
definition of metastability-freedom is satisfied by definition for state
transitions along the basic cycle, as the conditions involve memory flags and
timeouts (that are not associated with the states the nodes switch to) only.
To show the correctness of Statement~(iii), it is thus sufficient to prove that,
whenever $j$ switches from state $s$ of the basic cycle to $s'$ of the basic
cycle during time $[t_j,t'+3d]$, the transition from $s$ to \rec\ is disabled
from the time it switches to $s'$ until it observes itself in this state. We
consider transitions $tr(\acc,\rec)$, $tr(\wake,\rec)$, and $tr(\prop,\rec)$ one
after the other:

\begin{enumerate}
  \item $tr(\acc,\rec)$: We showed that node $j$'s condition $tr(\acc,\slp)$ is
  satisfied before time $t+4d \leq t+T_1/\vartheta$, i.e., before
  $tr(\acc,\rec)$ can hold, and no node resets its \acc\ flags less
  than $d$ time after switching to state \slp.
  When $j$ switches to state \acc\ again at or after time $t'$, $T_1$ will
  not expire earlier than time $t'+4d$.
  
  \item $tr(\wake,\rec)$: As part of the reasoning about Statement~(ii), we
  derived that $tr(\wake,\rec)$ does not hold at nodes from $W$ observing
  themselves in state \wake.
  
  \item $tr(\prop,\rec)$: The additional slack of $d$ in  
  \inequalityref{eq:acc_on_time} ensures that $T_5$ does not expire at any node
  in $W$ switching to state \acc\ during $(t',t'+2d)$ earlier than time $t'+3d$.
\end{enumerate}
Since $[t_j,t'+3d] \supset [t+3d,t'+3d]$, Statement~(iii) follows.
\end{proof}

Inductive application of \theoremref{theorem:stability} shows that by
construction of our algorithm, nodes in $W$ provably do not suffer from
metastability upsets once a $W$-quasi-stabilization point is reached, as long as
all nodes in $W$ remain non-faulty and the channels connecting them correct.
Unfortunately, it can be shown that it is impossible to ensure this property
during the stabilization period, thus rendering a formal treatment infeasible.
This is not a peculiarity of our system model, but a threat to any model that
allows for the possibility of metastable upsets as encountered in physical chip
designs. However, it was shown that, by proper chip design, the probability of
metastable upsets can be made arbitrarily
small~\cite{FFS09:ASYNC09}.\footnote{Note that it is feasible to incorporate
this issue into the model by means of the probability space, as it is beyond
control of ``reasonable'' adversaries to control signals on faulty channels (or
ones that originate at non-faulty nodes) precisely enough to ensure more than a
small probability of a metastable upsets. However, since it is (at best)
impractical to consider metastable states of the system in our theoretical
framework, essentially this approach boils down to counting the number of state
transitions during stabilization where a non-faulty node is in danger of
becoming metastable and control this risk by means of the union bound.}
\textit{In the remainder of this work, we will therefore assume that all
non-faulty nodes are metastability-free in all executions.}

The next lemma reveals a very basic property of the main algorithm that is
satisfied if no nodes may switch to state \join\ in a given period of time. It
states that if a non-faulty node switches to state \slp, other non-faulty
nodes cannot remain too far ahead or behind in their execution of the basic cycle.
\begin{lemma}\label{lemma:sleep_one}
Assume that at time $t_\slp$, some node from $W$ switches to \slp\ and no node from
$W$ is in state \join\ during $[t_\slp-T_1-d,t_\slp+2T_1+3d]$. Then
\begin{itemize}
  \item [(i)] at time $t_\slp+2T_1+3d$, any node is in one of the states
  \slp, \srw, \wake, or \rec;
  \item [(ii)] any node in states \slp, \srw, or \wake reset its timeout $T_2$
  at some time from $(t_\slp-\Delta_g-4d,t_\slp+(2-1/\vartheta)T_1+3d)$; and
  \item [(iii)] no node switches from \rec to \acc during $[t_\slp+T_1+2d,t_a]$,
  where $t_a>t_\slp+2T_1+3d$ denotes the infimum of times larger than
  $t_\slp+T_1+2d$ when a node switches to state \acc.
\end{itemize}
\end{lemma}
\begin{proof}
We claim that there is a subset $A\subseteq W$ of at least $n-2f$ nodes such
that each node from $A$ has been in state \acc\ at some time in the interval
$(t_\slp-T_1-d,t_\slp)$. To see this, observe that if a node switches to state \slp
at time $t_\slp$, it must have observed $n-2f$ non-faulty nodes in state \acc\ at
times from $(t_\slp-T_1,t_\slp]$, since it resets its \acc\ flags at the time
$t_a\geq t_\slp-T_1$ (that is minimal with this property) when it switched to
state \acc. Each of these nodes must have been in state \acc\ at some time
from $(t_\slp-T_1-d,t_\slp)$, showing the existence of a set $A\subseteq W$
as claimed.

During
\begin{eqnarray*}
&&\left[t_\slp+T_1+2d,\phantom{\frac{\vartheta(2T_1+3d)}{\vartheta}}\right.\\
&&\left.~t_\slp-T_1-d+\min\left\{
\frac{\vartheta(2T_1+3d)}{\vartheta},
\frac{T_2}{\vartheta}\right\}\right]\\
&\sr{eq:T_2}{=}&\left[t_\slp+T_1+2d,t_\slp+T_1+2d\right],
\end{eqnarray*}
no node from $A$ is observed in state \acc, as following the basic cycle
requires $T_2$ to expire, no node switches to \join, and in order to switch
directly from \rec to \acc, a timeout of $\vartheta(2T_1+3d)$ needs to expire
first. Since this also applies to the nodes from $A$ and no node is in state
\join until time $t_\slp+2T_1+3d$, the only way to do so is by following the
basic cycle via states \slp, \srw, \wake, \rdy, and \prop. However, this takes
at least until time
\begin{equation*}
t_\slp+\frac{T_2}{\vartheta}\sr{eq:T_2}{>}t_\slp+2T_1+3d
\end{equation*}
as well. This shows Statement~(iii) of the lemma.

Now consider any node that observes itself in one of the states \wake, \rdy, or
\prop at time $t_\slp-T_1-d$. By time $t_\slp+d$, it will memorize all nodes
from $A$ in \acc (provided that it did not switch to \acc in the meantime).
Hence, it satisfies $tr(\wake,\rec)$, $tr(\rdy,\prop)$, and $tr(\prop,\acc)$
until it switches to either \rec or \acc. It follows that any such node must
have switched to \rec or \acc by time $t_s+3d<t_s+T_1+2d$. On the other hand,
nodes that do not observe themselves in state \wake at time $t_\slp-T_1-d$ but
are in one of the states \slp, \srw, or \wake at this time or switch to \slp
during time $(t_\slp-T_1-d,t_\slp+2T_1+3d]$ must have reset their timeout $T_2$
at some time from
\begin{equation*}
\left(t_\slp-(2\vartheta+3)T_1-4d,
t_\slp+\left(2-\frac{1}{\vartheta}\right)T_1+3d\right),
\end{equation*}
i.e., Statement~(ii) holds. To infer Statement~(i), it remains to show that none
of the latter nodes may switch to \rdy until time $t_\slp+2T_1+3d$. As no nodes
from~$W$ are in state \join during $[t_\slp-T_1-d,t_\slp+2T_1+3d]$ by
assumption, the stetement follows immediately from Statement~(ii), as
\begin{equation*}
t_\slp-(2\vartheta+3)T_1+\frac{T_2}{\vartheta}-4d
\sr{eq:T_2}{>}t_\slp+2T_1+3d.
\end{equation*}
The lemma follows.
\end{proof}

Granted that nodes are not in state \join for sufficiently long, this implies
that nodes will switch to \slp in rough synchrony with others or drop out of the
basic cycle and switch to \rec.
\begin{corollary}\label{coro:clean}
Assume that at time $t_\slp$, a node from $W$ switches to \slp, no node is in
state \join during $[t_\slp-T_1-d,t_\slp+2T_1+4d]$, and also that during
$(t_\slp-\Delta_g,t_\slp)=(t_\slp-(2\vartheta+3)T_1,t_\slp)$ no
node in $W$ is in state \slp. Then at time $t_\slp+2T_1+4d$, any node from $W$
is either in one of the states \slp or \srw and observed in \slp, or it is and
is observed in state \rec.
\end{corollary}
\begin{proof}
We apply \lemmaref{lemma:sleep_one} to see that at time $t_\slp+2T_1+3d$, all
nodes are in one of the states \slp, \srw, \wake, or \rec. As nodes remain in
\slp for a timeout of duration $(2\vartheta+1)T_1\geq \vartheta (2T_1+4d)$, the
statement of the corollary follows immediately provided that we can show that
any node that does not switch to state \slp during $[t_\slp,t_\slp+T_1+3d]$ is
not in state \wake at time $t_\slp+T_1+3d$. Consider such a node. If there is a
time from $(t_\slp-\Delta_g,t_\slp+T_1+3d]$ when the node is not in one of the
states \slp, \srw, or \wake, it cannot be in state \wake at time
$t_\slp+2T_1+5d$, since it could not have switched to \slp again in order to get
there. Assume w.l.o.g.\ that the node switches to \slp exactly at time
$t_\slp-\Delta_g$. Thus, it must have previously reset its timeout $T_2$ no
later than
\begin{equation*}
t_\slp-\Delta_g -\frac{T_1}{\vartheta}\sr{eq:T_1}{\leq}t_\slp-\Delta_g-4d.
\end{equation*}
Hence we conclude by \lemmaref{lemma:sleep_one} that the node is in state \rec
at time $t_\slp+2T_1+5d$, finishing the proof.
\end{proof}

\subsection{Resynchronization Points}
In this section, we derive that within linear time, it is very likely that good
resynchronization points occur. As a first step, we infer from
\lemmaref{lemma:sleep_one} that whenever nodes may not enter state \join,
the time windows during which nodes may switch to \slp occur infrequently.

\begin{lemma}\label{lemma:window}
Suppose no node is in state \join during $[t^-,t^+]$. Then the volume of times
$t\in [t^-+T_1+d,t^+]$ satisfying that no node is in state \slp during
$(t-\Delta_g,t)$ is at least
\begin{eqnarray*}
&&\left(\frac{T_2-2\vartheta \Delta_g-(\vartheta-1)T_1-4\vartheta d} 
{T_2-(\vartheta-1)T_1-\vartheta d}\right)(t^+-t^-)\\
&&-(4\Delta_g+T_1+7d).
\end{eqnarray*}
\end{lemma}
\begin{proof}
Denote by $t_0$ the infimum of times from $[t^-+T_1+d,t^+]$ when a node switches
to \slp. Thus, by definition any time $t\in [t^-+\Delta_g+T_1+d,t_0]$ satisfies
that no node is in state \slp during $(t-\Delta_g,t)$. We proceed by induction
over increasing times $t_i\in(t_0,t^+]$, $i\in \{1,\ldots,i_{\max}\}$. The
induction halts at index $i_{\max}\in \N$ iff $t_{i_{\max}}>
t^+-T_2/\vartheta+(1-1/\vartheta)T_1+d$. We claim that, for each $i$, the
volume of times $t\in [t^-+T_1+d,t_i]$ such that no node is in state \slp during
$(t,t-\Delta_g)$ is at least
\begin{equation}\label{eq:volume}
t_i-t^--(T_1+d)-i(2\Delta_g+3d)
\end{equation}
and that 
\begin{eqnarray}
t_i&\geq & t^--(2\vartheta+1+1/\vartheta)T_1
-2d\nonumber\\
&&+i\left(\frac{T_2}{\vartheta}-\left(1-\frac{1}{\vartheta}\right)T_1-d\right).
\label{eq:t_i_growth}
\end{eqnarray}
In fact, we will show these bounds by establishing that no node is in state
\slp during
\begin{equation}\label{eq:no_sleep}
(t_{i-1}+(2\vartheta+3)T_1+3d,t_i)=(t_{i-1}+\Delta_g+3d,t_i)
\end{equation}
and that
\begin{equation}\label{eq:t_i_growth_weak}
t_i\geq t_{i-1}+\frac{T_2}{\vartheta}
-\Delta_g-4d
\sr{eq:T_2}{\geq} t_{i-1}+2\Delta_g+3d
\end{equation}
for all $i\in \{1,\ldots,i_{\max}\}$.

We first establish these bounds for $t_1$. By \lemmaref{lemma:sleep_one}, every
node not switching to state \rec until time $t_0+T_1+3d$ resets $T_2$ at some
time from $(t_0-\Delta_g-4d,t_0+3d)$ and is in one of the states
\slp, \srw, or \wake at time $t_0+T_1+3d$. Hence, such
nodes do not switch to state \rdy and subsequently to \prop, \acc, and \slp
again until $t_0+T_2/\vartheta-\Delta_g-4d\leq t^+$, giving
\begin{equation*}
t_1\geq t_0+\frac{T_2}{\vartheta}
-\Delta_g-4d.
\end{equation*}
Moreover, the lemma implies that no node is in state \slp during
$[t_0+(2\vartheta+3)T_1+3d,t_1)$, as any node in state \slp at time
$t_0+2T_1+3d$ will leave after a timeout of $(2\vartheta+1)T_1$ expires.
Hence, the volume of times $t\in [t^-+T_1+d,t_1]$ such that no node is in state
\slp during $(t,t-\Delta_g)$ is at least
\begin{equation*}
t_0-(t^-+T_1+d+\Delta_g)+t_1-(t_0+\Delta_g+3d),
\end{equation*}
showing the claim for $i=1$.

We now perform the induction step from $i<i_{\max}$ to $i+1$. By
\eqref{eq:no_sleep}, no node is in state \slp during
\begin{equation*}
(t_{i-1}+\Delta_g+3d,t_i)
\sr{eq:t_i_growth_weak}{\supseteq}
(t_i-\Delta_g,t_i).
\end{equation*}
Hence we can apply \corollaryref{coro:clean} to see that nodes not
observing themselves in state \slp at time $t_i+2T_1+4d$ switched to state \rec.
Therefore, nodes that continue to execute the basic cycle must have performed
their most recent reset of timeout $T_2$ at or after time $t_i-T_1-d$. Thus,
such nodes do not switch to state \rdy and subsequently to \prop, \acc, and \slp
again until $t_i+T_2/\vartheta-(1-1/\vartheta)T_1-d\leq t^+$, giving
\begin{equation*}
t_{i+1}\geq t_i+\frac{T_2}{\vartheta}-\left(1-\frac{1}{\vartheta}\right)T_1-d.
\end{equation*}
Moreover, no node is in state \slp during $[t_i+(2\vartheta+3)T_1+3d,t_{i+1})$.
These two statements show \inequalityref{eq:no_sleep} and
\inequalityref{eq:t_i_growth_weak} for $i+1$, and by means of the induction
hypothesis directly imply \inequalityref{eq:volume} and
\inequalityref{eq:t_i_growth} for $i+1$ as well, i.e., the induction succeeds.

From \inequalityref{eq:t_i_growth}, we have that
\begin{eqnarray}
i_{\max}&\leq& \frac{t^+-t^-+(2\vartheta+1+1/\vartheta)T_1+2d}
{T_2/\vartheta-(1-1/\vartheta)T_1-d}\nonumber\\
&\sr{eq:T_2}{<}&\frac{t^+-t^-}{T_2/\vartheta-(1-1/\vartheta)T_1-d}+1.
\label{eq:i_max}
\end{eqnarray}
Observe that the same reasoning as above shows that no node switches to \slp
during $[t_{i_{\max}}+\Delta_g+3d,t^+]$ since $t_{i_{\max}}\geq
t^+-(T_2/\vartheta-(1-1/\vartheta)T_1-d)$. Thus, inserting $i=i_{\max}$ into
\inequalityref{eq:volume}, we infer that the volume of times $t\in
[t^-+T_1+d,t^+]$ such that no node is in state \slp during $(t,t-\Delta_g)$ is
at least
\begin{eqnarray*}
&&t^+-t^--(T_1+d)-(i_{\max}+1)(2\Delta_g+3d)\\
&\sr{eq:i_max}{>}&
\left(\frac{T_2-2\vartheta \Delta_g-(\vartheta-1)T_1-4\vartheta d} 
{T_2-(\vartheta-1)T_1-\vartheta d}\right)(t^+-t^-)\\
&&-(4\Delta_g+T_1+7d),
\end{eqnarray*}
concluding the proof.
\end{proof}

We are now ready to advance to proving that good resynchronization points are
likely to occur within bounded time, no matter what the strategy of the
adversary is. To this end, we first establish that in any execution, at most of
the times a node switching to state \init\ will result in a good
resynchronization point. This is formalized by the following definition.

\begin{definition}[Good Times]
Given an execution $\cal E$ of the system, denote by ${\cal E}'$ any
execution satisfying that ${\cal E}|_{[0,t)}'={\cal E}|_{[0,t)}$,
where at time $t$ a node $i\in W$ switches to state \init\ in
${\cal E}'$.
Time $t$ is \emph{good in $\cal E$ with respect to $W$} provided that
for any such ${\cal E}'$ it holds that $t$ is a
good $W$-resynchronization point in ${\cal E}'$.
\end{definition}
The previous statement thus boils down to showing that in any execution, the
majority of the times is good.
\begin{lemma}\label{lemma:good}
Given any execution $\cal E$ and any time interval $[t^-,t^+]$, the volume of
good times in $\cal E$ during $[t^-,t^+]$ is at least
\begin{equation*}
\lambda^2(t^+-t^-)-\frac{11(1-\lambda)R_2}{10\vartheta}.
\end{equation*}
\end{lemma}
\begin{proof}
Assume w.l.o.g.\ that $|W|=n-f$ (otherwise consider a subset of
size $n-f$) and abbreviate
\begin{eqnarray}
\!\!\!\!\!\!\!&&\!\!\!N\nonumber\\
\!\!\!\!\!\!\!&:=&\!\!\!
\left(\frac{\vartheta(t^+-t^-)}{R_2}+\frac{11}{10}\right)(n-f)\notag\\ 
\!\!\!\!\!\!\!&\geq &\!\!\!
\left\lceil\frac{\vartheta(t^+-t^-)+R_2/10}{R_2}\right\rceil(n-f)\notag\\
\!\!\!\!\!\!\!&\sr{eq:R_2}{\geq}&\!\!\!
\left\lceil\frac{\vartheta(t^+-t^-)+\vartheta(R_1+4\Delta_g+T_1+10d)/(5(1-\lambda))}
{R_2}\right\rceil\nonumber\\
\!\!\!\!\!\!\!&&\!\!\!\cdot(n-f)\notag\\
\!\!\!\!\!\!\!&\sr{eq:def_lambda}{\geq}&\!\!\!
\left\lceil\frac{\vartheta(t^+-t^-+R_1+4\Delta_g+T_1+10d)}
{R_2}\right\rceil(n-f).\label{eq:N_bound}
\end{eqnarray}

The proof is in two steps: First we construct a measurable subset of
$[t^-,t^+]$ that comprises good times only.
In a second step a lower bound on the volume of this set is derived.

\medskip

\bfno{Constructing the set:} Consider an arbitrary time $t \in [t^-,t^+]$, and
assume a node $i\in W$ switches to state \init\ at time $t$. When it does so,
its timeout $R_3$ expires. By \lemmaref{lemma:counters} all timeouts of node $i$
that expire at times within $[t^-,t^+]$, have been reset at least once until
time $t^-$. Let $t_{R3}$ be the maximum time not later than $t$ when $R_3$ was
reset. Due to the distribution of $R_3$ we know that
\begin{equation*}
t_{R3} \sr{eq:R_3}{\leq} t- (R_2+3d).
\end{equation*}
Thus, node $i$ is not in state \init\ during time $[t-(R_2+2d),t)$,
and no node $j\in W$ observes $i$ in state \init\ during time $[t-(R_2+d),t)$.
Thereby any node $j$'s, $j\in W$, timeout $(R_2,\supp~i)$
corresponding to node $i$ is expired at time~$t$.

\medskip

We claim that the condition that no node from $W$ is in or observed in one of
the states \res\ or \srr\ at time $t$ is sufficient for $t$ being a
$W$-resynchronization point. To see this, assume that the condition is
satisfied. Thus all nodes $j\in W$ are in states \none\ or $\supp~k$ for some
$k\in \{1,\ldots,n\}$ at time $t$. By the algorithm, they all will switch to
state $\supp~i$ or state \srr\ during $(t,t+d)$.
It might happen that they subsequently switch to another state \supp~$k'$ for
some $k'\in V$, but all of them will be in one of the states with signal \supp\
during $(t+d,t+2d]$. Consequently, all nodes will observe at least $n-f$ nodes
in state \supp\ during $(t',t+2d)$ for some time $t'<t+2d$. Hence, those nodes
in $W$ that were still in state \supp~$i$ (or \supp~$k'$ for some $k'$) at time
$t+d$ switch to state \srr\ before time $t+2d$, i.e., $t$ is a
$W$-resynchronization point.

\medskip

We proceed by analyzing under which conditions $t$ is a good
$W$-resynchronization point. Recall that in order for $t$ to be good, it has to
hold that no node from $W$ switches to state \slp\ during $(t-\Delta_g,t)$ or is
in state \join\ during $(t-T_1-d,t+4d)$.

\medskip

We begin by characterizing subsets of good times within $(t_r,t_r') \subset
[t^-,t^+]$, where $t_r$ and $t_r'$ are times such that during $(t_r,t_r')$ no
node from $W$ switches to state \srr. Due to timeout
\begin{equation*}
R_1 \sr{eq:R_1}{\geq} (4\vartheta+2)d,
\end{equation*}
we know that during $(t_r+R_1+2d,t_r')$, no node from $W$ will be in, or be
observed in, states \srr\ or \res. Thus, if a node from $W$ switches to \init\
at a time within $(t_r+R_1+2d,t_r')$, it is a $W$-resynchronization point.
Further, all nodes in $W$ will be in state \dorm\ during $(t_r+R_1+2d,t_r'+4d)$.
Thus all nodes in $W$ will be observed to be in state \dorm\ during
$(t_r+R_1+3d,t_r'+4d)$, implying that they are not in state \join\ during
$(t_r+R_1+3d,t_r'+4d)$. In particular, any time $t\in (t_r+R_1+T_1+4d,t_r')$
satisfies that no node in $W$ is in state \join\ during $(t-T_1-d,t+4d)$.
Applying \corollaryref{lemma:window}, we infer that the total volume of times
from $(t_r,t_r')$ that is good is at least
\begin{eqnarray}
&&\left(\frac{T_2-2\vartheta \Delta_g-(\vartheta-1)T_1-4\vartheta d} 
{T_2-(\vartheta-1)T_1-\vartheta
d}\right)(t^+-t^-)\nonumber\\
&&-(4\Delta_g+T_1+10d).\label{eq:good}
\end{eqnarray}

In other words, up to a constant loss in each interval $(t_r,t_r')$, a
constant fraction of the times are good.

\medskip

\bfno{Volume of the set:} In order to infer a lower bound on the volume of good
times during $[t^-,t^+]$, we subtract from $t^+-t^-$ the volume of some
intervals during which we cannot exclude that a node switches to \srr, increased
by the constant term $R_1+4\Delta_g+T_1+10d$ from \inequalityref{eq:good}. The
inequality then yields that at least a fraction of
$(T_2-2\vartheta \Delta_g-(\vartheta-1)T_1-4\vartheta d)/
(T_2-(\vartheta-1)T_1-\vartheta d)$ of the remaining volume of times is good.
Note that we also need to account for the fact that nodes may already be in
state \srr at time $t^-$, which we account for by also covering events prior to
$t^-$ when nodes switch to \srr. Formally, we define
\begin{equation*}
\bar{G} =\!\!\! \bigcup_{\substack{t_r\in [t^--(R_1+4\Delta_g+T_1+10d),t^+]\\
\exists i\in W:\,i\text{ switches to }\srr\ \text{at }t_r}}
\!\!\!\!\![t_r,t_r+R_1+4\Delta_g+T_1+10d]
\end{equation*}
and strive for a lower bound on the volume of $[t^-,t^+] \setminus \bar{G}$. In
order to lower bound the good times in $[t^-,t^+]$, it is thus sufficient to
cover all times when a node switches to \srr during
$[t^--(R_1+4\Delta_g+T_1+10d),t^+]$ by $2N-1$ intervals of volume at most ${\cal V}$
and then infer a lower bound of $t^+-t^--2N({\cal V}+R_1+4\Delta_g+T_1+10d)$ on the
volume of $[t^-,t^+]\setminus \bar{G}$. The remainder of the proof hence is
concerned with deriving such a cover of the times when nodes may switch to \srr
during $[t^--(R_1+4\Delta_g+T_1+10d),t^+]$.

\medskip

Observe that any node in $W$ does not switch
to state \init\ more than
\begin{eqnarray}
&&\left\lceil\frac{t^+-t^-+R_1+4\Delta_g+T_1+10d}{R_3}\right\rceil\notag\\
&\sr{eq:R_3}{\leq}&
\left\lceil\frac{t^+-t^-+R_1+4\Delta_g+T_1+10d}{R_2+d}\right\rceil\notag\\
&\sr{eq:N_bound}{\leq}& \frac{N}{n-f}\label{ieq:frac1}
\end{eqnarray}
times during $[t^--(R_1+4\Delta_g+T_1+10d),t^+]$.

Now consider the case that a node in $W$ switches to state \srr\ at a time $t$
satisfying that no node in $W$ switched to state \init\ during
$(t-(8\vartheta+6)d,t)$. This necessitates that this node observes $n-f$ of its
channels in state \supp\ during $(t-(2\vartheta+1)d,t)$, at least $n-2f\geq f+1$
of which originate from nodes in $W$. As no node from $W$ switched to \init\
during $(t-(8\vartheta+6)d,t)$, every node that has not observed a node $i\in
V\setminus W$ in state \init\ at a time from $(t-(8\vartheta+4)d,t)$ when
$(R_2,\supp~i)$ is expired must be in a state whose signal is \none\ during
$(t-(2\vartheta+2)d,t)$ due to timeouts. Therefore its outgoing channels are not
in state \supp\ during $(t-(2\vartheta+1)d,t)$. By means of contradiction, it
thus follows that for each node $j$ of the at least $f+1$ nodes (which are all
from $W$), there exists a node $i\in V\setminus W$ such that node $j$ resets
timeout $(R_2,\supp~i)$ during the time interval $(t-(8\vartheta+4)d,t)$.

The same reasoning applies to any time $t'\not \in (t-(8\vartheta+6)d,t)$
satisfying that some node in $W$ switches to state \srr\ at time $t'$ and no
node in $W$ switched to state \init\ during $(t'-(8\vartheta+6)d,t')$. Note that
the set of the respective at least $f+1$ events (corresponding to the at least
$f+1$ nodes from $W$) where timeouts $(R_2,\supp~i)$ with $i\in V\setminus W$
are reset and the set of the events corresponding to $t$ are disjoint. However,
the total number of events where such a timeout can be reset during
$[t^--(R_1+4\Delta_g+T_1+10d),t^+]$ is upper bounded~by
\begin{eqnarray}
&&|V\setminus
W||W|\left\lceil\frac{t^+-t^-+R_1+4\Delta_g+T_1+10d}{R_2/\vartheta} 
\right\rceil\notag\\
&\sr{eq:N_bound}{<}&(f+1)N,\label{ieq:frac2}
\end{eqnarray}
i.e., the total number of channels from nodes not in $W$ ($|V\setminus
W|$ many) to nodes in $W$ multiplied by the number of times an
associated timeout can expire at a receiving node in $W$ during
$[t^--(R_1+4\Delta_g+T_1+10d),t^+]$.

\medskip

With the help of inequalities~\eqref{ieq:frac1} and~\eqref{ieq:frac2}, we can
show that $\bar{G}$ can be covered by less than $2N$ intervals of size
$(R_1+4\Delta_g+T_1+10d)+(8\vartheta+6)d$ each. By \inequalityref{ieq:frac1},
there are no more than $N$ times $t\in [t^--(R_1+4\Delta_g+T_1+10d),t^+]$ when one
of the $|W|=n-f$ many non-faulty nodes switches to \init\ and thus may cause
others to switch to state \srr\ at times in $[t,t+(8\vartheta+6)d]$. Similarly,
\inequalityref{ieq:frac2} shows that the channels from $V\setminus W$ to $W$ may
cause at most $N-1$ such times $t\in [t^--(R_1+4\Delta_g+T_1+10d),t^+]$, since any
such time requires the existence of at least $f+1$ events where timeouts
$(R_2,\supp~i)$, $i\in V\setminus W$, are reset at nodes in $W$, and the
respective events are disjoint. Thus, all times $t_r\in
[t^--(R_1+4\Delta_g+T_1+10d),t^+]$ when some node $i\in W$ switches to \srr\ are
covered by at most $2N-1$ intervals of length $(8\vartheta+6)d$.

This results in a cover $\bar{G}' \supseteq \bar{G}$ consisting of at
most $2N-1$ intervals that satisfies
\begin{eqnarray*}
\operatorname{vol}\left( \bar{G} \right) &\leq & \operatorname{vol}\left(
\bar{G}' \right) \\
&<& 2N \theterm.
\end{eqnarray*}

As argued previously, summing over the at most $2N$ intervals that remain in
$[t^-,t^+] \setminus \bar{G}'$ and using \inequalityref{eq:good}, it follows
that the volume of good times during $[t^-,t^+]$ is at least
\begin{eqnarray*}
&&\frac{T_2-2\vartheta \Delta_g-(\vartheta-1)T_1-4\vartheta d}
{T_2-(\vartheta-1)T_1-\vartheta d}\\
&&(t^+-t^--2N\theterm\\
&\sr{eq:lambda}{\geq}&\lambda(t^+-t^--2N\theterm)\\
&=& \lambda \left(t^+-t^--2\left(\frac{\vartheta(t^+-t^-)}{R_2}
+\frac{11}{10}\right)(n-f)\right.\\
&&\left.\cdot\theterm\right)\\
&=& \left(1-\frac{2\vartheta\theterm(n-f)}{R_2}\right)\\
&&\cdot\lambda(t^+-t^-)\\
& &
-\frac{11\lambda\theterm (n-f)}{5}\\
&\sr{eq:R_2}{\geq}&\lambda^2(t^+-t^-)
-\frac{11(1-\lambda)R_2}{10\vartheta},
\end{eqnarray*}
as claimed. The lemma follows.
\end{proof}

We are now in the position to prove our second main theorem, which states that
a good resynchronization point occurs within $\BO(R_2)$ time with overwhelming
probability.

\begin{theorem}\label{theorem:resync}
Denote by $\hat{E}_3:=\vartheta (R_2+3d)+8(1-\lambda)R_2+d$ the maximal value
the distribution $R_3$ can attain plus the at most $d$ time until $R_3$ is reset
whenever it expires.
For any $k\in \N$ and any time $t$, with probability at least
$1-(1/2)^{k(n-f)}$ there will be a good $W$-resynchronization
point during $[t,t+(k+1)\hat{E}_3]$.
\end{theorem}
\begin{proof}
Assume w.l.o.g.\ that $|W|=n-f$ (otherwise consider a subset of size
$n-f$).
Fix some node $i\in W$ and denote by $t_0$ the infimum of times from
$[t,t+(k+1)\hat{E}_3]$ when node $i$ switches to \init.
We have that $t_0< t+\hat{E}_3$.
By induction, it follows that node~$i$ will switch to state \init\ at
least another $k$ times during $[t,t+(k+1)\hat{E}_3]$ at the
times $t_1<t_2<\ldots<t_k$.
We claim that each such time $t_j$, $j\in \{1,..,k\}$, has an
independently by $1/2$ lower bounded probability of being good
and therefore being a good $W$-resynchronization point.

We prove this by induction on $j$: As induction hypothesis, suppose
for some $j\in \{1,\ldots,k-1\}$, we showed the statement for
$j'\in \{1,\ldots,j-1\}$ and the execution of the system is fixed
until time $t_{j-1}$, i.e., ${\cal E}|_{[0,t_{j-1}]}$ is given.
Now consider the set of executions that are extensions of ${\cal
E}|_{[0,t_{j-1}]}$ and have the same clock functions as $\cal E$.
For each such execution ${\cal E}'$ it holds that ${\cal
E}'|_{[0,t_{j-1}]}={\cal E}|_{[0,t_{j-1}]}$, and all nodes'
clocks make progress in ${\cal E}'$ as in ${\cal E}$.
Clearly each such ${\cal E}'$ has its own time $t_j <
t+(j+1)\hat{E}_3$ when $R_3$ expires next after $t_{j-1}$ at node
$i$, and $i$ switches to \init.
We next characterize the distribution of the times $t_{j}$.

As the rate of the clock driving node $i$'s $R_3$ is between $1$ and $\vartheta$,
$t_{j}>t_{j-1}$ is within an interval, call it $[t^-,t^+]$,
of size at most
\begin{equation*}
t^+-t^-\leq 8(1-\lambda)R_2,
\end{equation*}
regardless of the progress that $i$'s clock $C$ makes in any execution~${\cal
E}'$.

Certainly we can apply \lemmaref{lemma:good} also to each of the
${\cal E}'$, showing that the volume of times from $[t^-,t^+]$
that are \emph{not} good in ${\cal E}'$ is at most
\begin{equation*}
(1-\lambda^2)(t^+-t^-)+\frac{11(1-\lambda)R_2}{10\vartheta}.
\end{equation*}

Since clock $C$ can make progress not faster than at rate
$\vartheta$ and the probability density of $R_3$ is
constantly $1/(8(1-\lambda)R_2)$ (with respect to
the clock function $C$), we obtain that the probability of $t_{j}$ not
being a good time is upper bounded~by
\begin{eqnarray*}
&&\frac{(1-\lambda^2)(t^+-t^-)+11(1-\lambda)R_2/(10\vartheta)}
{8(1-\lambda)R_2/\vartheta}\\
&\leq &
\vartheta(1-\lambda^2)+\frac{11}{80}\\
&\sr{eq:def_lambda}{<}&
\vartheta\frac{9}{25\vartheta}+\frac{7}{50}=\frac{1}{2}.
\end{eqnarray*}
Here we use that the time when $R_3$ expires is independent of ${\cal
E}'|_{[0,t_{j-1}]}$.

We complete our reasoning as follows. Given ${\cal E}|_{[0,t_{j-1}]}$, we permit
an adversary to choose ${\cal E}'$, including random bits of all nodes and full
knowledge of the future, with the exception that we deny it control or knowledge
of the time $t_j$ when $R_3$ expires at node $i$, i.e., ${\cal E}'$ is an
imaginary execution in which $R_3$ does not expire at $i$ at any time greater
than $t_{j-1}$. Note that for the good $W$-resynchronization points we
considered, the choice of ${\cal E}'$ does not affect the probability that
$t_1,\ldots,t_{j-1}$ are good $W$-resynchronization points: The conditions
referring to times greater than a $W$-resynchronization point $t$, i.e., that
all nodes in $W$ switch to state \srr\ during $(t,t+2d)$ and no node in $W$
shall be in state \join\ during $(t-T_1-d,t+4d)$, are already fully determined
by the history of the system until time $t$. As we fixed ${\cal E}'$, the
behavior of the clock driving $R_3$ is fixed as well. Next, we determine the
time $t_j$ when $R_3$ expires according to its distribution, given the behavior
of node $i$'s clock. The above reasoning shows that time $t_j$ is good in ${\cal
E}'$ with probability at least $1/2$, independently of ${\cal
E}'|_{[0,t_{j-1}]}={\cal E}|_{[0,t_{j-1}]}$. We define that ${\cal
E}|_{[0,t_j)}={\cal E}'|_{[0,t_j)}$ and in ${\cal E}$ node $i$ switches to state
\init\ (because $R_3$ expired). As---conditional to the clock driving $R_3$
and $t_{j-1}$ being specified---$t_j$ is independent of ${\cal E}|_{[0,t_j)}$,
${\cal E}$ is indistinguishable from ${\cal E}'$ until time $t_j$. Because $t_j$
is good with probability at least $1/2$ independently of ${\cal
E}|_{[0,t_{j-1}]}'={\cal E}|_{[0,t_{j-1}]}$, so it is in ${\cal E}$. Hence, in
${\cal E}$ $t_j$ is a good $W$-resynchronization point with probability $1/2$,
independently of ${\cal E}|_{[0,t_{j-1}]}$. Since ${\cal E}'$ was chosen in an
adversarial manner, this completes the induction step.

In summary, we showed that for \emph{any} node in $W$ and \emph{any} execution
(in which we do not manipulate the times when $R_3$ expires at the respective
node), starting from the second time during $[t,t+(k+1)\hat{E}_3]$ when $R_3$
expires at the respective node, there is a probability of at least $1/2$ that
the respective time is a good $W$-resynchronization point. Since we assumed that
$|W|=n-f$ and there are at least $k$ such times for each node in $W$, this
implies that having no good $W$-resynchronization point during
$[t,t+(k+1)\hat{E}_3]$ is as least as unlikely as $k(n-f)$ unbiased and
independent coin flips all showing tail, i.e., $(1/2)^{k(n-f)}$. This concludes
the proof.
\end{proof}

\subsection{Stabilization via Good Resynchronization Points}

Having established that eventually a good $W$-resyn\-chronization point $t_g$
will occur, we turn to proving the convergence of the main routine. We start with a
few helper statements wrapping up that a good resynchronization point guarantees
proper reset of flags and timeouts involved in the stabilization process of the
main routine.

\begin{lemma}\label{lemma:clean}
Suppose $t_g$ is a good $W$-resynchronization point. Then
\begin{itemize}
  \item [(i)] each node $i\in W$ switches to \pass\ at a time
  $t_i\in (t_g+4d,t_g+(4\vartheta+3)d)$ and observes itself in state \dorm\
  during $[t_g+4d,\tau_{i,i}(t_i))$,
  \item [(ii)] $\Mem_{i,j,\join}|_{[\tau_{i,i}(t_i),t_\join]}\equiv 0$ for all
  $i,j\in W$, where $t_\join\geq t_g+4d$ is the infimum of all times greater than
  $t_g-T_1-d$ when a node from $W$ switches to \join,
  \item [(iii)] $\Mem_{i,j,\srw}|_{[\tau_{i,i}(t_i),t_s]}\equiv 0$ for
  all $i,j\in W$, where $t_s\geq t_g+(1+1/\vartheta)T_1$ is the infimum of all
  times greater or equal to $t_g$ when a node from $W$ switches to \srw,
  \item [(iv)] no node from $W$ resets its \srw\ flags during
  $[t_g+(1+1/\vartheta)T_1,t_g+R_1/\vartheta]$, and
  \item [(v)] no node from $W$ resets its \join\ flags due to switching to \pass\ during
  $[t_g+(1+1/\vartheta)T_1,t_g+R_1/\vartheta]$.
\end{itemize}
\end{lemma}
\begin{proof}
All nodes in $W$ switch to state \srr\ during $(t_g,t_g+2d)$ and switch to state
\res\ when their timeout of $\vartheta 4d$ expires, which does not happen until
time $t_g+4d$. Once this timeout expired, they switch to state \pass\ as soon as
they observe themselves in state \res, i.e., by time $t_g+(4\vartheta+3d)$.
Hence, every node $i\in W$ does not observe itself in state \res\ within
$[t_g+3d,\tau_{i,i}(t_i))$, and therefore is in state \dorm\ during
$[t_g+3d,\tau_{i,i}(t_i)]$. This implies that it observes itself in state \dorm\
during $[t_g+4d,\tau_{i,i}(t_i))$, completing the proof of
Statement~(i).

\medskip

Moreover, from the definition of a good $W$-resynchro\-nization point we have
that no nodes from $W$ are in state \join\ at times in $[t_g-T_1-d,t_\join)$.
Statement~(ii) follows, as every node from $W$ resets its \join\ flags upon
switching to state \pass\ at time $t_i$.

\medskip

Regarding Statement~(iii), observe first that no nodes from $W$ are in state
\srw\ during $(t_g-d,t_g+(1+1/\vartheta)T_1)$ for the following reason: By
definition of a good $W$-resynchronization point no node from $W$ switches to
\slp\ during $(t_g-\Delta_g,t_g) \supset (t_g-(2\vartheta+1)T_1-3d,t_g)$. Any
node in $W$ that is in states \slp\ or \srw\ at time $t_g-(2\vartheta+1)T_1-3d$
switches to state \wake\ before time $t_g-d$ due to timeouts. Finally, any node
in $W$ switching to \slp\ at or after time $t_g$ will not switch to state \srw\
before time $t_g+(1+1/\vartheta)T_1$. The observation follows.

Since nodes in $W$ reset their \srw\ flags at some time from
\begin{eqnarray*}
[t_i,\tau_{i,i}(t_i)]&\subset &(t_g+3d,t_g+(4\vartheta+4)d)\\
&\sr{eq:T_1}{\subseteq}& (t_g+3d,t_g+(1+1/\vartheta)T_1),
\end{eqnarray*}
Statement~(iii) follows.

\medskip

Statements~(iv) and~(v) follow from the fact that all nodes in $W$
switch to state \pass\ until time
\begin{equation*}
t_g+(3+4\vartheta)d\sr{eq:T_1}{\leq}t_g+\left(1+\frac{1}{\vartheta}\right)T_1-d,
\end{equation*}
while timeout $(R_1,\srr)$ must expire first in order to switch to
\dorm\ and subsequently \pass\ again.
\end{proof}

Before we proceed with our third main statement showing eventual stabilization,
we make a few more basic observations. Firstly, if nodes do not make
progress on the basic cycle, they must eventually switch to \rec, i.e., the
timeout conditions ensure detection of deadlocks.
\begin{lemma}\label{lemma:drop_out}
For any time $t^-$ and any node it holds that it must be in state \rec or \join
or switch to \slp at some time from
$[t^-,t^-+(1-1/\vartheta)T_1+T_2+T_4+T_5+4d)$.
\end{lemma}
\begin{proof}
Suppose a node is never in state \rec or \join during
$[t^-,t^-+(1-1/\vartheta)T_1+T_2+T_4+T_5+4d)$. Thus it may follow transitions
along the basic cycle only. Assume w.l.o.g.\ that the node switched to \slp
right before time $t^-$. Thus, it switched to state \acc beforehand, no later
than time $t^--T_1/\vartheta$. Due to timeouts, either switch to \rec at some
point in time or switch to \slp, \srw, \wake, \rdy, \prop, \acc, and finally
\slp again. At each state, it takes less than $d$ time until a respective
timeout is started and it observes itself in the respective state. Hence, the
node switches to \rec or \slp before time
\begin{eqnarray*}
&&t^--\frac{T_1}{\vartheta}+\max\{(2\vartheta+2)T_1+3d,T_2\}\\
&&+T_4+T_5+T_1+4d\\
&\sr{eq:T_2}{=}&t^-+\left(1-\frac{1}{\vartheta}\right)T_1+T_2+T_4+T_5+4d,
\end{eqnarray*}
proving the claim of the lemma.
\end{proof}

Secondly, after a good $W$-resynchronization point $t_g$, no node from $W$ will
switch to state \join\ until either time $t_g+T_7/\vartheta+4d$ or
$T_6/\vartheta$ time after the first non-faulty node switched to \srw\ again
after $t_g$. By proper choice of $T_7>T_6$ and $T_6$, this will guarantee that
nodes from $W$ do not switch to \join\ prematurely during the final steps of the
stabilization process.

\begin{lemma}\label{lemma:switch}
Suppose $t_g$ is a good $W$-resynchronization point.
Denote by $t_s$ the infimum of times greater than $t_g$ when a node in
$W$ switches to state \srw\ and by $t_\join$ the infimum of times
greater than $t_g-T_1-d$ when a node in $W$ switches to state
\join. Then, starting from time $t_g+4d$, $tr(\rec,\join)$ is not satisfied
at any node in $W$ until time
\begin{equation*}
\min\left\{t_g+\frac{T_7}{\vartheta}+4d,t_s+\frac{T_6}{\vartheta}\right\}>t_\join.
\end{equation*}
\end{lemma}
\begin{proof}
By Statements~(ii) and~(iii) of \lemmaref{lemma:clean} and
\inequalityref{eq:T_1}, we have that $t_s\geq t_g+T_1+4d\geq
t_g+(4\vartheta+4)d$ and $t_\join\geq t_g+4d$. Consider a node $i\in W$ not
observing itself in state \dorm\ at some time $t\in [t_g+4d,t_\join]$. According
to Statements~(i) and~(ii) of \lemmaref{lemma:clean}, the threshold condition of
$f+1$ nodes memorized in state \join\ cannot be satisfied at such a node. By
statements~(i) and~(iii) of the lemma, the threshold condition of $f+1$ nodes
memorized in state \srw\ cannot be satisfied unless $t>t_s$. Hence, if at time
$t$ a node from $W$ satisfies that it observes itself in state \act, we have
that $T_6$ expired after being reset after time $t_s$, i.e.,
$t>t_s+T_6/\vartheta$. Moreover, by Statement~(i) of \lemmaref{lemma:clean}, we
have that if $T_7$ is expired at any node in $W$ at time $t$, it holds that
$t>t_g+T_7/\vartheta+4d$. Altogether, we conclude that $tr(\rec,\join)$ is not
satisfied at any node in $W$ during
\begin{equation*}
\left[t_g+4d,\min\left\{t_g+\frac{T_7}{\vartheta}+4d,t_s+\frac{T_6}{\vartheta}
\right\}\right].
\end{equation*}
In particular, $t_\join$ must be larger than the upper boundary of this
interval, concluding the proof.
\end{proof}

Thirdly, after a good $W$-resynchronization point, any node in $W$ switches to
\rec or to \srw within bounded time, and all nodes in $W$ doing the latter
will do so in rough synchrony. Using the previous lemmas, we can show that this
happens before the transition to \join\ is enabled for any node.

\begin{lemma}\label{lemma:rec}
Suppose $t_g$ is a good $W$-resynchronization point and use the
notation of \lemmaref{lemma:switch}. Define
$t^+:=t_g-T_1/\vartheta+T_2+T_4+T_5+3d$ and denote by $t_\slp$ the infimum of
all times greater than $t_g-\Delta_g$ when a node in $W$ switches to \slp. Then
$t_\slp\geq t_g$ and either
\begin{itemize}
  \item [(i)] $t_\slp< t^+$ and at time $t_\slp+2T_1+4d$, any node in $W$ is
  either in one of the states \slp or \srw and observed in \slp or is in \rec and also
  observed in \rec, or
  \item [(ii)] all nodes in $W$ are observed in state \rec at time
  $t^++2T_1+4d$.
\end{itemize}
\end{lemma}
\begin{proof}
By definition of a good resynchronization point, no node switches to \slp during
$(t_g-\Delta_g,t_g)$, giving that $t_\slp\geq t_g$. If $t_\slp<t^+$,
\lemmaref{lemma:switch} yields that
\begin{eqnarray*}
t_\join&>&\min\left\{t_g+\frac{T_7}{\vartheta}+4d,t_\slp+\frac{T_6}{\vartheta}\right\}\\
&\sr{eq:T_7}{\geq}&\min\left\{t^++2T_1+4d,t_\slp+\frac{T_6}{\vartheta}\right\}\\
&\sr{eq:T_6}{\geq}& t_\slp+2T_1+4d.
\end{eqnarray*}
Therefore, by definition of a good resynchronization point, no nodes are in
state \join during $(t_g-T_1-d,t_\join)\supset (t_\slp-T_1-d,t_\slp+T_1+4d)$.
Recalling that during $(t_g-\Delta_g,t_\slp)$, no node is in state
\slp, the preconditions of \corollaryref{lemma:window} hold, implying Statement~(i).

If $t_\slp\geq t^+$, by definition of a good resynchronization point no node
switched to sleep during $(t_g-\Delta_g,t^+)\supset (t_g-T_1-d,t^+)$ and no node
is in state \join during $(t_g-T_1-d,t_\join)$. By \lemmaref{lemma:switch},
\begin{eqnarray*}
t_\join&>&\min\left\{t_g+\frac{T_7}{\vartheta}+4d,t_\slp+\frac{T_6}{\vartheta}\right\}\\
&\sr{eq:T_7}{\geq}&\min\left\{t^++2T_1+4d,t^++\frac{T_6}{\vartheta}\right\}\\
&\sr{eq:T_6}{\geq}& t^++2T_1+4d.
\end{eqnarray*}
Hence, \lemmaref{lemma:drop_out} states that every node must be in state \rec at
some time in $(t_g-T_1-d,t^+)$. Since nodes do not leave state \rec during
$(t_g-T_1-d,t_\join)$, Statement~(ii) follows.
\end{proof}

We have everything in place for proving that a good resynchronization point
leads to stabilization within $R_1/\vartheta-3d$ time.
\begin{theorem}\label{theorem:stabilization}
Suppose $t_g$ is a good $W$-resynchroniza\-tion point. Then there is a
quasi-stabilization point during $(t_g,t_g+R_1/\vartheta-3d]$.
\end{theorem}
\begin{proof}
For simplicity, assume during this proof that $R_1=\infty$, i.e., by
Statement~(i) of \lemmaref{lemma:clean} all nodes in $W$ observe themselves in
states \pass or \act at times greater or equal to $t_g+(4\vartheta+4)d$.
We will establish the existence of a quasi-stabilization point at a time larger
than $t_g$ and show that it is upper bounded by $t_g+R_1/\vartheta-3d$. Hence
this assumption can be made w.l.o.g., as the existence of the
quasi-stabilization point depends on the execution up to time
$t_g+R_1/\vartheta$ only, and $R_1$ cannot expire before this time at any node
in $W$. Moreover, by Statements~(i) and~(ii) of \lemmaref{lemma:clean}, every
node satisfies $\Mem_{i,i,\join}\equiv 0$ on
$[t_g+(4\vartheta+4)d,t_{i,\join})$, where $t_{i,\join}$ denotes the infimum of
all times greater or equal to $t_g+(4\vartheta+4)d$ when node $i$ switches to
\join. During the time span considered in this proof, every node switches at
most once to \join, thus we may w.l.o.g.\ assume that $\Mem_{i,i,\join}=0$ is
always satisfied in the following. We use the notation of Lemmas
\ref{lemma:switch} and \ref{lemma:rec}. By Statements~(ii) of
\lemmaref{lemma:clean} and \inequalityref{eq:T_1}, we have that $t_s\geq
t_g+(1-1/\vartheta)T_1\geq t_g+(4\vartheta+4)d$.

According to \lemmaref{lemma:clean}, all nodes in $W$ switched to state \pass\
during $(t_g+4d,t_g+(3+4\vartheta)d)$, implying that at any node in $W$, $T_7$
will expire at some time from 
\begin{eqnarray*}
&&(t_g+T_7/\vartheta+4d,t_g+T_7+(4\vartheta+4)d)\\
&\sr{eq:T_7}{\subset}&
(t_g+(1+1/\vartheta)T_1,t_g+T_7+(4\vartheta+4)d).
\end{eqnarray*}
By \lemmaref{lemma:switch}, thus $t_\join>t_g+(1+1/\vartheta)T_1$, and by
Statement~(v) of \lemmaref{lemma:clean}, no node resets its \join flags after
$t_g+(1+1/\vartheta)T_1$ again (before $R_1$ expires).

\medskip

{\bf Case 1:} Assume $t_\slp\geq t^+$. Thus, Statement~(ii) of
\lemmaref{lemma:rec} applies, i.e., all nodes are observed in state \rec by time
$t^++2T_1+4d$. Any node from $W$ will switch to state \join\ before time
$t_g+T_7+(4\vartheta+4)d$ because $T_7$ expires no later than that.
Subsequently, it will switch to \prop as soon as it memorizes all non-faulty
nodes in state \join. Denote by $t_\prop\in
(t_g+T_7/\vartheta+4d,t_g+T_7+(4\vartheta+5)d)$ the minimal time when a node
from $W$ switches from \join\ to \prop. Certainly, nodes in $W$ do not switch
from \wake\ to \rdy\ during $(t_\prop,t_\prop+2d)$ and therefore also not reset their
\join\ flags before time $t_\prop+3d$. As nodes in $W$ reset their \prop\ and \acc\
flags upon switching to state \join, some node in $W$ must memorize $n-2f\geq
f+1$ non-faulty nodes in state \join\ at time $t_\prop$. According to Statement~(ii)
of \lemmaref{lemma:clean}, these nodes must have switched to state \join\ at or
after time $t_\join$. Hence, all nodes in $W$ will memorize them in state \join\
by time $t_\prop+d$ and thus have switched to state \join. Hence, all nodes in $W$
will switch to state \prop\ before time $t_\prop+2d$ and subsequently to state \acc\
before time $t_\prop+3d$, i.e., $t_\prop\leq t_g+T_7+(4\vartheta+5)d$ is a
quasi-stabilization point.

\textbf{Case 2:} Assume $t_\slp<t^+$. By Statement~(i) of \lemmaref{lemma:rec}, all
nodes are observed in either \slp or \rec at time $t_\slp+2T_1+4d$. The nodes
observed in state \slp will have been observed in state \srw and switched to
\wake by time $t_\slp+(2\vartheta+3)T_1+5d$.

\textbf{Case 2a:} Suppose $< f+1$ nodes in $W$ are observed in state
\slp at time $t_\slp+2T_1+4d$, i.e., $\geq n-2f\geq f+1$ non-faulty nodes are
observed in state \rec. By \lemmaref{lemma:switch}, we have that
\begin{eqnarray*}
&&t_\slp+(2\vartheta+3)T_1+7d\\
&\leq & t_s+\left(2\vartheta+3-\frac{1}{\vartheta}\right)T_1+7d\\
&\sr{eq:T_6}{\leq}& \min\left\{t_s+\frac{T_6}{\vartheta},
t^++\left(2\vartheta+3-\frac{1}{\vartheta}\right)T_1+7d\right\}\\
&\sr{eq:T_7}{\leq}&\min\left\{t_s+\frac{T_6}{\vartheta},
t_g+\frac{T_7}{\vartheta}+4d\right\}
<t_\join.
\end{eqnarray*}
Hence, any node observing itself in state \wake at some time $t\in
(t_\slp+2T_1+4d,t_\slp+(2\vartheta+3)T_1+6d)$ will also observe at
least $f+1$ nodes in state \rec and switch to \rec. As any node in \slp or \srw
at time $t_\slp+2T_1+4d$ will observe itself in state \wake no later than time
$t_\slp+(2\vartheta+3)T_1+6)d$, by time
$t_\slp+(2\vartheta+3)T_1+7d<t_\join$, all nodes observe themselves
in state \rec. From here we can argue analogously to the first case, i.e., there
exists a quasi-stabilization point $t_\prop\leq t_g+T_7+(4\vartheta+5)d$.

\textbf{Case 2b:} Suppose $\geq f+1$ nodes in $W$ are observed in state
\slp at time $t_\slp+2T_1+4d$. These nodes will switch to \wake and subsequently
\rdy until time 
\begin{eqnarray}
\!\!\!\!\!&&\max\left\{t_\slp+(2\vartheta+3)T_1+6d,
t_\slp-\frac{T_1}{\vartheta}+T_2+d\right\}\nonumber\\
\!\!\!\!\!&\sr{eq:T_2}{=}&t_\slp-\frac{T_1}{\vartheta}+T_2+d\label{eq:in_ready}
\end{eqnarray}
due to $T_2$ being expired while observing themselves in \wake unless they
switch from \wake to \rec. Note that these nodes reset their \acc flags upon
switching to \wake. Denote by $t_\prop$ and $t_\acc$ the infima of times
greater than $t_\slp+2T_1+4d$ when a node switches to \prop or \acc,
respectively. Recall that any node switching from \rec to \join resets its \prop
and \acc flags, and any node switching from \wake to \rdy resets its \prop
flags. Hence, we have for all $i,j\in W$ that
\begin{itemize}
  \item [(i)] $\Mem_{i,j,\prop}(t)=0$ at any time $t\in
  [t_\slp+2T_1+4d,t_\prop]$ when $i$ observes itself in \rdy or \join, and
  \item [(ii)] $\Mem_{i,j,\acc}(t)=0$ at any time $t\in [t_\slp+2T_1+4d,t_\acc]$
  when $i$ observes itself in \rdy, \join, or \prop.
\end{itemize}

By Statements~(ii) and~(iv) of \lemmaref{lemma:clean}, no node from $W$ resets
its \srw\ flags at or after time $t_s\geq t_g+(1+1/\vartheta)T_1$. As $t_s\geq
t_\slp+2T_1+4d$ and all nodes observed in \slp at time $t_\slp+2T_1+4d$ will be
observed in \srw by time $t_\slp+(2\vartheta+3)T_1+5d$,
Statement~(i) of the lemma implies that all nodes in $W$ switch to \act at some
time from $(t_s,t_\slp+(2\vartheta+3)T_1+5d)\subseteq
(t_\slp+2T_1+6d,t_\slp+(2\vartheta+3)T_1+5d)$. As, by the
Statements~(i) and~(ii) from above, the first node switching to state \prop must
do so because of an expiring timeout, \lemmaref{lemma:switch} yields that
\begin{eqnarray*}
t_\prop
&\geq &\min\left\{t_\join,t_\slp-T_1-d+\frac{T_2}{\vartheta}\right\}\\
&\geq &\min\left\{t_g+\frac{T_7}{\vartheta}+4d,
t_s+\frac{T_6}{\vartheta},\right.\\
&&\left.t_\slp-T_1-d+\frac{T_2+\min\{T_3,T_4\}}{\vartheta}\right\}\\
&\sr{eq:T_4}{\geq} &\min\left\{t_\slp-t^++t_g+\frac{T_7}{\vartheta}+4d,\right.\\
&&t_\slp+\left(2-\frac{1}{\vartheta}\right)T_1+\frac{T_6}{\vartheta},\\
&&\left.t_\slp-T_1-d+\frac{T_2+T_3}{\vartheta}\right\}\\
&\stackrel{(\ref{eq:T_3},\ref{eq:T_7})}{=}&
t_\slp+\left(2-\frac{1}{\vartheta}\right)T_1+\frac{T_6}{\vartheta}.
\end{eqnarray*}
Therefore,
\begin{equation}\label{eq:bound_t_prop}
t_\prop\geq t_\slp+\left(2-\frac{1}{\vartheta}\right)+\frac{T_6}{\vartheta}
\sr{eq:T_6}{\geq}t_\slp-\frac{T_1}{\vartheta}+T_2+2d.
\end{equation}
By \inequalityref{eq:in_ready}, we conclude that at time
$t_\slp-T_1/\vartheta+T_2+2d<t_\prop$, any node from $W$ observes itself
in one of the states \rdy, \rec, or \join.

Again, we distinguish two cases.

\textbf{Case 2b-I:} $t_\prop<t_\slp-T_1-d+(T_2+T_3)/\vartheta$. As previously
used, no node can switch from \rdy to \prop during
$(t_\slp+2T_1+4d,t_\slp-T_1-d+(T_2+T_3)/\vartheta))$. Hence, there must be a
node that switches from \join to \prop at time $t_\prop$. By Statements~(i) and~(ii)
from above, the node must memorize at least $n-2f\geq f+1$ nodes from $W$ in
state \join at time $t_\prop$. By Statement~(ii) of \lemmaref{lemma:clean}, these
nodes must have switched to \join at or after time $t_\join$. By
Statements~(iii) and~(v) of the lemma, no node resets its \join flags during
$[t_g+(1+1/\vartheta)T_1,t_g+R_1/\vartheta)\supset [t_\prop,t_\prop+3d)$ unless it
switches to state \join. Hence, all nodes still in state \rec have switched to
\join by time $t_\prop+d$, giving that all nodes are in one of the states \rdy,
\join, or \acc at time $t_\prop+d$ (since they cannot leave \acc earlier than
$t_\prop+T_1/\vartheta\geq t_\prop+4d$ again).

\textbf{Case 2b-II:} $t_\prop\geq t_\slp-T_1-d+(T_2+T_3)/\vartheta$. Recall
that all nodes switched to \act by time
$t_\slp+(2\vartheta+3)T_1+5d$. Hence, any node observing itself in
state \rec at time $t_\slp+2T_1+4d$ will have switched to \join
because $T_6$ expired by time
\begin{eqnarray}
&&t_\slp+(2\vartheta+3)T_1+T_6+6d\nonumber\\
&\sr{eq:T_3}{\leq}&
t_\slp-T_1-d+\frac{T_2+T_3}{\vartheta}\leq t_\prop.\label{eq:T_6_exp}
\end{eqnarray}
Hence, also in this case all nodes are in one of the states \rdy, \join,
or \acc at time $t_\prop+d$.

\textbf{Continuing Case 2b:} Next, we claim that any node is in states \prop or
\join by time $t_\slp-T_1/\vartheta+T_2+T_4+2d$. To see this, observe that any
node following the basic cycle must switch from \rdy to \prop by this time due
to timeouts. On the other hand, according to \inequalityref{eq:T_6_exp}, all
nodes in state \rec switch to \join by time
\begin{equation*}
t_\slp-T_1-d+\frac{T_2+T_3}{\vartheta}\stackrel{(\ref{eq:T_2},\ref{eq:T_4})}{<}
t_\slp-\frac{T_1}{\vartheta}+T_2+T_4+2d,
\end{equation*}
showing the claim.

In summary, we showed the following points:
\begin{itemize}
  \item [(i)] At time $t_\prop+d$, all nodes are observed in states \rdy,
  \join, \prop, or \acc.
  \item [(ii)] All nodes switch to states \prop or \join during
  $[\min\{t_\join,t_\prop\},t_\slp-T_1/\vartheta+T_2+T_4+2d)$.
  \item [(iii)] No node resets its \prop or \acc flags at or after time
  $t_\prop+d$ unless switching to \acc first.
  \item [(iv)] No node memorizes nodes in state \prop or \acc that have not
  been in that state at or after time $t_\prop$.
\end{itemize}

We claim that the infimum $t_q$ of all times from
\begin{equation*}
\left[t_\prop,t_\slp-\frac{T_1}{\vartheta}+T_2+T_4+d\right]
\end{equation*}
when a node switches to \acc is a quasi-stabilization point. Note that because
\begin{eqnarray*}
&&t_\slp-\frac{T_1}{\vartheta}+T_2+T_4+4d\\
&\sr{eq:T_5}{\leq}& t_\slp+T_1+4d+\frac{T_5+T_6}{\vartheta}\\
&\sr{eq:bound_t_prop}{\leq}& t_\prop+\frac{T_5}{\vartheta}
\end{eqnarray*}
no node will switch from \prop to \rec before time $t_q+3d$.

Again, we distinguish two cases. First assume that
$t_q<t_\slp-T_1/\vartheta+T_2+T_4+2d$, i.e., at time $t_q$ indeed a node
switches to state \acc. Due to Statement~(iv) from the above list and the
minimality of $t_q$, it follows that the respective node memorizes $n-2f\geq
f+1$ nodes from $W$ in state \prop that switched to \prop at or after time
$t_p$. These nodes must be in one of the states \prop or \acc during
$[t_q,t_q+3d]$. According to Statement~(i) from above, thus all nodes still in
\rdy will switch to \prop by time $t_q+d$. By time $t_q+2d$, all nodes in \join
will observe the at least $n-f$ nodes from $W$ in one of the states \join,
\prop, or \acc, and hence switch to \prop. Another $d$ time later, all nodes
will have switched to \acc, i.e., $t_q$ is indeed a quasi-stabilization point.

On the other hand, if $t_q=t_\slp-T_1/\vartheta+T_2+T_4+2d$, Statement~(ii) from
the above list gives that all nodes from $W$ are in one of the states \join,
\prop, or \acc during $[t_q+d,t_q+3d]$. Therefore, nodes will switch from \join
to \prop and subsequently from \prop to \acc until time $t_q+3d$ as well.

It remains to check that in all cases, the obtained quasi-synchronization point
$t_q$ occurs no later than time $t_g+R_1/\vartheta-3d$. In Cases~1 and~2a, we
have that
\begin{equation*}
t_q=t_\prop\leq
t_g+T_7+(4\vartheta+5)d\sr{eq:R_1}{\leq}t_g+\frac{R_1}{\vartheta}-3d.
\end{equation*}
In Case~2b, it holds that
\begin{eqnarray*}
t_q&\leq &t_\slp-\frac{T_1}{\vartheta}+T_2+T_4+d\\
&\leq &t^+-\frac{T_1}{\vartheta}+T_2+T_4+d\\
&=& t_g-\frac{2T_1}{\vartheta}+2T_2+2T_4+T_5+4d\\
&\sr{eq:R_1}{\leq} & t_g+\frac{R_1}{\vartheta}-3d.
\end{eqnarray*}
We conclude that indeed all nodes in $W$ switch to \acc\ within a window of less
than $3d$ time before, at any node in $W$, $R_1$ expires and it leaves state
\res, concluding the proof.
\end{proof}

Finally, putting together the established main theorems and
\lemmaref{lemma:constraints}, we deduce that the system will
stabilize from an arbitrary initial state provided that a subset of $n-f$
nodes remains coherent for a sufficiently large period of time.

\begin{corollary}\label{coro:stabilization}
Let $W\subseteq V$, where $|W|\geq n-f$, and define for any $k\in \N$
\begin{equation*}
T(k):=(k+2)(\vartheta(R_2+3d)+8(1-\lambda)R_2+d)+R_1/\vartheta.
\end{equation*}
Then, for any $k\in \N$, the proposed algorithm is a $(W,W^2)$-stabilizing
pulse synchronization protocol with skew $2d$ and accuracy bounds
$(T_2+T_3)/\vartheta-2d$ and $T_2+T_4+7d$ stabilizing within time $T(k)$ with
probability at least $1-2^{-k(n-f)}$. It is feasible to pick timeouts such that
$T(k)\in \BO(kn)$ and $T_2+T_4+7d\in \BO(1)$.
\end{corollary}
\begin{proof}
The satisfiability of \conditionref{cond:timeout_bounds} with $T(k)\in \BO(kn)$
and $T_2+T_4+7d\in \BO(1)$ follows from \lemmaref{lemma:constraints}. Assume
that $t^+$ is sufficiently large for $[t^- +T(k)+2d,t^+]$ to be non-empty, as
otherwise nothing is to show. By definition, $W$ will be coherent during
$[t_c^-,t^+]$, with $t_c^- = t^-+\vartheta(R_2+3d)+8(1-\lambda)R_2+d$. According
to \theoremref{theorem:resync}, there will be some good $W$-resynchronization
point $t_g\in [t_c^-,t_c^-+(k+1)(\vartheta(R_2+3d)+8(1-\lambda)R_2+d)]$ with
probability at least $1-1/2^{k(n-f)}$. If this is the case,
\theoremref{theorem:stabilization} shows that there is a $W$-stabilization point
$t\in [t_g,t^- +T(k)]$. Applying \theoremref{theorem:stability} inductively, we
derive that the algorithm is a $(W,W^2)$-stabilizing pulse synchronization
protocol with the bounds as stated in the corollary that stabilizes within time
$T(k)$ with probability at least $1-1/2^{k(n-f)}$.
\end{proof}

\subsection{Late Joining and Fast Recovery}
An important aspect of combining self-stabilization with Byzantine
fault-tolerance is that the system can remain operational when facing a limited
number of transient faults. If the affected components stabilize quickly enough,
this can prevent future faults from causing system failure. In an environment
where transient faults occur according to a random distribution that is not too
far from being uniform (i.e., one deals not primarily with bursts), the mean
time until failure is therefore determined by the time it takes to recover from
transient faults. Thus, it is of significant interest that a node that starts
functioning according to the specifications again synchronizes as fast as
possible to an existing subset of correct nodes. Moreover, it is of interest
that a node that has been shut down temporarily, e.g.\ for maintenance, can join
the operational system again quickly.

\begin{theorem}\label{theorem:constant}
Suppose there exists a node $i$ in $V$ and a set $W\subseteq V$, $|W|\geq n-f$,
such that there is a $W$-stabilization point at some time $t^-$ and $W\cup \{i\}$ is
coherent during $[t^-,t^-+(1+5/(2\vartheta))R_1]$. Then there is a
$(W\cup\{i\})$-stabilization point at some time $t \in
[t^-,t^-+(1+5/(2\vartheta))R_1]$.
\end{theorem}
\begin{proof}
Again, the proof is executed by distinguishing cases. W.l.o.g., we assume for
the moment that $W\cup \{i\}$ is coherent during $[t^-,\infty)$ and later show
that indeed $t\leq t^-+(1+5/(2\vartheta))R_1$.

\textbf{Case 1:} 
Node $i$ does not switch to \srr during
$[t^-,t^-+R_1+(\vartheta-1)T_1+2T_2+2T_4+18d]$. Thus, after $R_1$ expires at the
latest by time $t^-+R_1+d$, it will observe itself in \dorm during
$[t^-+R_1+2d,t^-+R_1+(\vartheta-1)T_1+2T_2+2T_4+18d]$ and therefore not be (or
observe itself) in state \join during
$[t^-+R_1+3d,t^-+R_1+(\vartheta-1)T_1+2T_2+2T_4+18d]$. By
\theoremref{theorem:stability}, there is a $W$-stabilization point $t_W\in
[t^-+R_1+(\vartheta-1)T_1+4d,t^-+R_1+(\vartheta-1)T_1+T_2+T_4+9d)$.
Subsequently, the nodes in $W$ will switch to \slp during
$[t_W+T_1/\vartheta,t_W+T_1+5d]$. Denote by $t_\slp$ the minimum of the
respective times. We apply \lemmaref{lemma:sleep_one} to $W\cup \{i\}$. Thus, at
time $t_\slp+2T_1+3d$, node $i$ is either in state \rec and will not leave until
the next $W$-stabilization point (or it switches to \join), or it is in state
\slp and reset its timeout $T_2$ at some time from
$[t_W^--\Delta_g+T_1/\vartheta-4d,t_W^-+(3-1/\vartheta)T_1+8d]$.

\textbf{Case 1a:}
Node $i$ is in \rec at time $t_\slp+2T_1+3d$. As it cannot switch to \join until
time $t^-+R_1+(\vartheta-1)T_1+2T_2+2T_4+18d$, it will stay in \rec until the
subsequent $W$-stabilization point $t_W'\in (t_W+(T_2+T_3)/\vartheta,
t^-+R_1+(\vartheta-1)T_1+2T_2+2T_4+14d)$ (existing according to
\theoremref{theorem:stability}). By time $t_W'$, clearly timeout
$(\rec,\vartheta(2T_1+3d))$ has expired at the node, as
\begin{equation*}
t_W'\geq
t_W+\frac{T_2+T_3}{\vartheta}\sr{eq:T_2}{>}t_\slp+(\vartheta+1)(2T_1+3d).
\end{equation*}
Because $T_1/\vartheta\geq 4d$, $i$ will observe all nodes from $W$ in \acc
during $[t_W'+3d,t_W'+4d]$. Hence it will switch to \acc by time $t_W'+3d$,
i.e., $t_W'$ is a $W\cup \{i\}$ quasi-stabilization point.

\textbf{Case 1b:} Node $i$ is in \slp at time $t_\slp+2T_1+3d$. Denote by $t_W'$
the $W$-stabilization point subsequent to $t_W$ as in the previous case. As no
node from $W$ is observed in state \acc or \rec during $[t_s+2T_1+3d,t_W')$ and
$i$ reset its timeout $T_2$ no earlier than time
$t_W-\Delta_g+T_1/\vartheta-4d$, it will not switch to \rec before time
$\min\{t_W',t_W-\Delta_g+(T_1+T_2+T_3+T_5)/\vartheta-4d\}$ unless it switches to
\acc first. However, as it resets its \prop and \acc flags before switching to
\rdy, it cannot switch to \acc before at least $f$ nodes from $W$ switched to
\prop (unless switching to \rec first). Moreover, by time $t_W'$, it will
already have switched to \rdy since
\begin{equation*}
t_W'\geq t_W+\frac{T_2+T_3}{\vartheta}\geq
t_W+\left(3-\frac{1}{\vartheta}\right)T_1+\frac{T_2}{\vartheta}+8d.
\end{equation*}
Hence, reasoning analogously to the proof of \theoremref{theorem:stability},
$t_W'$ is in fact a $W\cup\{i\}$-stabilization point provided that $i$ switches
to \acc instead of \rec first. This in turn follows from the bound
\begin{eqnarray*}
&&t_W-\Delta_g+\frac{T_1+T_2+T_3+T_5}{\vartheta}-4d\\
&\sr{eq:T_5}{\geq}&t_W-\Delta_g+T_1+T_2+T_4+\frac{T_2+T_3-T_6}{\vartheta}-4d\\
&\sr{eq:T_3}{>}&t_W-\Delta_g+(2\vartheta+5)T_1+T_2+T_4+3d\\
&\sr{eq:T_1}{>}& t_W+T_2+T_4+7d\\
&>& t_W'+2d,
\end{eqnarray*}
where in the last step we used that $t_W'< t_W+T_2+T_4+5d$ according to
\theoremref{theorem:stability}. This shows that $T_5$ does not expire at $i$
while it is in \prop before time $t_W'+2d$. Hence, $t_W'$ is a
$W\cup\{i\}$-stabilization point.

\textbf{Case 2:} Node $i$ switches to \srr at a time $t'\in
[t^-,t+R_1+(\vartheta-1)T_1+2T_2+2T_4+18d]$. Denote by $t_W$ and $t_W'$ the
maximal $W$-stabilization point smaller than $t'$ and the minimal
$W$-stabilization point larger than $\max\{t',t_W+2d\}$, which exist by
\theoremref{theorem:stability}. Denote by $t_\slp$ the minimal time larger than
$t_W$ when a node from $W$ switches to \slp. Analogously\footnote{Note that we
can apply \lemmaref{lemma:sleep_one} even if $i$ switches to \join, as we can
simply replace the set $A$ by $W$.} to Case 1b, $t_W'$ is a
$W\cup\{i\}$-stabilization point if $i$ is in state \slp at time
$t_\slp+2T_1+3d$. Hence, assume w.l.o.g.\ that $i$ is in state \rec or already
switched to \join by this time. Analogously to Case 1a, $t_W'$ will be a $W\cup
\{i\}$-quasi-stabilization point if it stays in \rec until time $t_W'+3d$.
Therefore, w.l.o.g., $i$ switches to \join at some time during $(t',t_W'+3d)$,
implying that it will leave the state no later than time $t_W'+4d$ and switch to
state \acc by time $t_W'+5d$.

Now either $i$ continues to execute the basic cycle and thus will, analogously
to Case 1b, participate in the minimal $W$-stabilization point $t_W''> t_W'+2d$,
or it will switch to \rec again. In the latter case, it cannot switch back to
\join until at least time $t'+R_1/\vartheta$ because it needs to reset its \join
flags first, which happens upon switching to \pass only. As we have that
\begin{eqnarray*}
t'+\frac{R_1}{\vartheta}&\sr{eq:R_1}{\geq}&
t'-\frac{2T_1}{\vartheta}+2T_2+2T_4+T_5+7d\\
&\sr{eq:T_5}{>}&t'-\frac{2T_1}{\vartheta}+3T_2+3T_4-T_6+7d\\
&\stackrel{(\ref{eq:T_3},\ref{eq:T_4})}{>}& t'+2T_2+2T_4+14d\\
&\geq & t_W''+4d,
\end{eqnarray*}
$i$ cannot leave state \rec through \join again before time $t_W''+4d$.
Therefore, $t_W''$ is a $W\cup\{i\}$-quasi-stabilization point, analogously to
Case 1a.

We have shown that there is some $W\cup\{i\}$-quasi-stabilization at the latest
by time
\begin{eqnarray*}
t_W''&\leq & t'+2T_2+2T_4+10d\\
&\leq & t^-+R_1+(\vartheta-1)T_1+4T_2+4T_4+28d
\end{eqnarray*}
in Case 2, while in Case 1 there is a quasi-stabilization point no later than
time $t^-+R_1+(\vartheta-1)T_1+2T_2+2T_4+18d$. By
\theoremref{theorem:stability}, this implies a $W\cup\{i\}$-stabilization point
by time
\begin{equation*}
t^-+R_1+(\vartheta-1)T_1+5T_2+5T_4+23d<t^-+\left(1+\frac{5}{2\vartheta}\right)R_1,
\end{equation*}
where the estimate is obtained analogously to the bound
$t'+R_1/\vartheta>t_W''+4d$ shown above. This concludes the proof, as indeed
there is a $W\cup\{i\}$-stabilization point no later than time
$t^-+(1+5/(2\vartheta))R_1$.
\end{proof}

%% file: generalizations.tex
\section{Generalizations}\label{sec:generalizations}

This section provides a few extensions of the core results derived in the
previous section. In particular, we show that it is not necessary to map faulty
channels to, for example, faulty nodes (thus rendering a non-faulty node
effectively faulty in terms of results), that the algorithm can tolerate an even
stronger adversary than defined in \sectionref{sec:model} without significant
change of stabilization time, and that in many reasonable setting stabilization
takes $\BO(R_1)$ time only, even if there is no majority of non-faulty nodes
that is already synchronized. With the exception of
\corollaryref{coro:quick_simple}, we again follow \cite{DFLS11:TR} during this
section.

\subsection{Synchronization Despite Faulty Channels}
\theoremref{theorem:stabilization} and our notion of coherency require that all
involved nodes are connected by correct channels only. However, it is desirable
that non-faulty nodes synchronize even if they are not connected by correct
channels. To capture this, the notions of coherency and stability can be
generalized as follows.
\begin{definition}[Weak Coherency]
We call the set $C\subseteq V$ \emph{weakly coherent} during $[t^-,t^+]$, iff for
any node $i\in C$ there is a subset $C'\subseteq C$ that contains $i$, has size
$n-f$, and is coherent during $[t^-,t^+]$.
\end{definition} 
In particular, if there are in total at most $f$ nodes that are faulty
or have faulty outgoing channels, then the set of non-faulty nodes is
(after some amount of time) weakly coherent.
\begin{corollary}\label{coro:weak_stabilization}
For each $k\in\N$ let $T'(k) :=
T(k)-((\vartheta(R_2+3d)+8(1-\lambda)R_2+d))$, where $T(k)$ is
defined as in \corollaryref{coro:stabilization}.
Suppose the subset of nodes $C\subseteq V$ is weakly coherent during
the time interval $[t^-,t^+] \supseteq [t^-
+T'(k)+T_2+T_4+8d,t^+] \neq \emptyset$.
Then, with probability at least $1-(f+1)/2^{k(n-f)}$, there is a
$C$-quasi-stabilization point $t\leq t^- + T'(k) + T_2+T_4+5d$
such that the system is weakly $C$-coherent during $[t,t^+]$.
\end{corollary}
\begin{proof}
By the definition of weak coherency, every node in $C$ is in some
coherent set $C'\subseteq C$ of size $n-f$.
Hence, for any such $C'$ it holds that we can cover all nodes in $C$
by at most $1+|V \setminus C'| \leq f+1$ coherent sets
$C_1,\ldots,C_{f+1}\subseteq C$.
By \corollaryref{coro:stabilization} and the union bound, with
probability at least $1-(f+1)/2^{k(n-f)}$, for each of these sets
there will be at least one stabilization point during
$[t^-,t^-+T'(k)-(T_2+T_4+5d)]$.
Assuming that this is indeed true, denote by $t_{i_0}\in
[t^-,t^-+T'(k)-(T_2+T_4+5d)]$ the time
\begin{eqnarray*}
&&\max_{i\in \{1,\ldots,f+1\}}\{\max \{t\leq t^-+T'(k)-(T_2+T_4+5d)\\
&&~~~~~~~~~~~~~~~|\,t\mbox{ is a $C_i$-stabilization point}\}\},
\end{eqnarray*}
where $i_0\in \{1,\ldots,f+1\}$ is an index for which the first maximum is
attained and $t_{i_0}$ is the respective maximal time, i.e., $t_{i_0}$ is a
$C_{i_0}$-stabilization point.

Define $t_{i_0}'\in (t_{i_0}+2d,t^-+T'(k)]$ to be minimal such that it is
another $C_{i_0}$-stabilization point. Such a time must exist by
\theoremref{theorem:stability}. Since the theorem also states that no node from
$C_{i_0}$ switches to state \acc\ during $[t_{i_0}+2d,t_{i_0}')$ and $C_i\cap
C_{i_0}\neq \emptyset$, there can be no $C_i$-stabilization point during
$(t_{i_0}+2d,t_{i_0}'-2d)$ for any $i\in \{1,\ldots,f+1\}$. Applying the theorem
once more, we see that there are also no $C_i$-stabilization points during
$(t_{i_0}'+2d,t_{i_0}'+(T_2+T_3)/\vartheta)-2d$ for any $i\in
\{1,\ldots,f+1\}$. On the other hand, the maximality of $t_{i_0}$ implies that
every $C_i$ had a stabilization point by time $t_{i_0}$. Applying
\theoremref{theorem:stability} to the latest stabilization point until time
$t_{i_0}$ for each $C_i$, we see that it must have another stabilization
point before time $t_{i_0}+T_2+T_4+5d$. We have that
\begin{equation*}
\frac{2(T_2+T_3)}{\vartheta}-2d
\sr{eq:T_2}{>}\frac{T_2+T_3+T_5}{\vartheta}\sr{eq:T_5}{>}T_2+T_4+5d,
\end{equation*}
i.e., all $C_i$ have stabilization points within a short time interval of
$(t_{i_0}'-2d,t_{i_0}'+2d)$. Arguing analogously about the previous
stabilization points of the sets $C_i$ (which exist because $t_{i_0}$ is
maximal), we infer that all $C_i$ had their previous stabilization point
during $(t_{i_0}-2d,t_{i_0}+2d)$.

Now suppose $t_a$ is the minimal time in $(t_{i_0}'-2d,t_{i_0}'+2d)$ when a node
from $C$ switches to \acc\ and this node is in set $C_i$ for some $i\in
\{1,\ldots,f+1\}$. As usual, there must be at least $f+1$ non-faulty nodes from
$C_i$ in state \prop\ at time $t_a$ and by time $t_a+d$, all nodes from $C_i$
will be observed in either of the states \prop\ or \acc. As $|C_i\cap C_j|\geq
f+1$ for any $j\in \{1,\ldots,f+1\}$, all nodes in $C_j$ will observe at
least $f+1$ nodes in states \prop or \acc at times in $(t_a,t_a+2d)$. We have
that $t_a\geq t_{i_0}+(T_2+T_3)/\vartheta-2d$ according to
\theoremref{theorem:stability}. As no nodes switched to state \acc\ during
$(t_{i_0}+2d,t_a)$ and none of them switch to state \rec\
(cf.~\theoremref{theorem:stability}), for any $j$ we can bound
\begin{eqnarray*}
t_a+d&\geq & t_i+d+\frac{T_2+T_3}{\vartheta}\\
&>& t_j-3d+\frac{T_2+T_3}{\vartheta} \\
&\sr{eq:T_3}{>}& t_j+T_2+3d
\end{eqnarray*}
that all nodes from $C_j$ are in one of the states \rdy, \prop, or \acc at time
$t_a+d$. Hence, they will switch from \rdy\ to \prop\ if they still are in \rdy\
before time $t_a+2d$. Less than $d$ time later, all nodes in $C_j$ will memorize
all nodes in $C_j$ in state \prop\ and therefore switch to \acc\ if not done so
yet. Since $j$ was arbitrary, it follows that $t_a$ is a $C$-quasi-stabilization
point.
\end{proof}

\begin{corollary}\label{coro:weak_stability}
Suppose $C\subseteq V$ is weakly coherent during $[t^-,t^+]$ and $t\in
[t^-,t^+-(T_2+T_4+8d)]$ is a $C$-quasi-stabilization point. Then
\begin{itemize}
  \item [(i)] all nodes from $C$ switch to \acc\ exactly once within $[t,t+3d)$;
  \item [(ii)] there will be a $C$-quasi-stabilization point
  $t'\in[t+(T_2+T_3)/\vartheta,t+T_2+T_4+5d)$ satisfying that no
  nodes switch to \acc\ in the time interval $[t+3d,t')$;
  \item [(iii)] and each node $i$'s, $i\in C$, main state machine
  (\figureref{fig:main_simple}) is metastability-free during $[t+4d,t'+4d)$.
\end{itemize}
\end{corollary}
\begin{proof}
Analogously to the proofs of \theoremref{theorem:stability} and
\corollaryref{coro:weak_stabilization}.
\end{proof}
We point out that one cannot get stronger results by the proposed technique.
Even if there are merely $f+1$ failing channels, this can e.g.\ effectively
render a node faulty (as it may never see $n-f$ nodes in states \prop\ or \acc)
or exclude the existence of a coherent set of size $n-f$ (if the channels
connect $f+1$ disjoint pairs of nodes, there can be no subset of $n-f$ nodes
whose induced subgraph contains correct channels only). Stronger resilience to
channel faults would necessitate to propagate information over several hops in a
fault-tolerant manner, imposing larger bounds on timeouts and weaker
synchronization guarantees.

Combination of \corollaryref{coro:weak_stabilization} and
\corollaryref{coro:weak_stability} finally yields:

\begin{corollary}\label{coro:final_weak}
Let $C \subseteq V$ be such that, for each $i\in C$, there is a set
$C_i \subseteq C$ with $|C_i| = n-f$, and let $E =\bigcup_{i\in
C} C_i^2$. Then, for any $k\in\N$, the proposed algorithm is a
$(C,E)$-stabilizing pulse synchronization protocol with skew $3d$ and accuracy
bounds $(T_2+T_3)/\vartheta-3d$ and $T_2+T_4+8d$ stabilizing within time
$T(k)+T_2+T_4+5d$ with probability at least $1-(f+1)/2^{k(n-f)}$.
\end{corollary}
\begin{proof}
Analogously to the proof of \corollaryref{coro:stabilization}
\end{proof}

\subsection{Stronger Adversary}
So far, our analysis considered a fixed set $C$ of coherent (or weakly coherent)
nodes. But what happens if whether a node becomes faulty or not is not
determined upfront, but depends on the execution? Phrased differently, does the
algorithm still stabilize quickly with a large probability if an adversary may
``corrupt'' up to $f$ nodes, but may decide on its choices as time progresses,
fully aware of what happened so far? Since we operate in a system where all
operations take positive time, it might even be the case that a node might fail
just when it is about to perform a certain state transition, and would not have
done so if the execution had proceeded differently. Due to the way we use
randomization, this however makes little difference for the stabilization
properties of the algorithm.
\begin{corollary}
Suppose at every time $t$, an adversary has full knowledge of the state of the
system up to and including time $t$, and it might decide on in total up to $f$
nodes (or all channels originating from a node) becoming faulty at arbitrary
times. If it picks a node at time $t$, it fully controls its actions after and
including time $t$. Furthermore, it controls delays and clock drifts of
non-faulty components within the system specifications, and it initializes the
system in an arbitrary state at time $0$. For any $k\in \N$, define $t_k$ as
\begin{equation*}
(k+3)(\vartheta(R_2+3d)
+8(1-\lambda)R_2+d)+R_1/\vartheta+T_2+T_4+5d.
\end{equation*}
Then the set of all nodes that remain non-faulty until time $t_k$ reaches a
quasi-stabilization point during $[\hat{E}_3,t_k]$ with probability at least
\begin{equation*}
1-(f+1)e^{-k(n-f)/2}.
\end{equation*}
Moreover, at any time $t\geq \hat{E}_3$, the set of nodes that are non-faulty at
time $t$ is coherent. 
\end{corollary}
\begin{proof}
The last statement of the corollary holds by definition.

We need to show that \theoremref{theorem:resync} holds for the modified time
interval $[\hat{E}_3,(k+3)\hat{E}_3]$ with the modified probability of at least
$1-e^{-k(n-f)/2}$. If this is the case, we can proceed as in
Corollaries~\ref{coro:weak_stabilization} and~\ref{coro:weak_stability}.

We start to track the execution from time $\hat{E}_3$. Whenever a non-faulty
node switches to state \init\ at a good time, the adversary must corrupt it in
order to prevent subsequent deterministic stabilization. In the proof of
\theoremref{theorem:resync}, we showed that for any non-faulty node, there are
at least $k+1$ different times during $[\hat{E}_3,(k+3)\hat{E}_3]$ when it
switches to \init\ that have an independently by $1/2$ lower bounded probability
to be good. Since \lemmaref{lemma:good} holds for \emph{any} execution where we
have at most $f$ faults, the adversary corrupting some node at time $t$ affects
the current and future trials of that node only, while the statement still holds
true for the non-corrupted nodes. Thus, the probability that the adversary may
prevent the system from stabilizing until time $t_k$ is upper bounded by the
probability that $(k+1)(n-f)$ independent and unbiased coin flips show $f$ or
less times tail. Chernoff's bound states for the random variable $X$ counting
the number of tails in this random experiment that for any $\delta\in (0,1)$,
\begin{equation*}
P[X<(1-\delta)\E[X]]
<\left(\frac{e^{-\delta}}{(1-\delta)^{1-\delta}}\right)^{\E[X]}
<e^{-\delta \E[X]}.
\end{equation*}
Inserting $\delta=k/(k+1)$ and $\E[X]=(k+1)(n-f)/2$, we see that the probability
that
\begin{equation*}
P[X\leq f]\leq P[X<(n-f)/2]<e^{-k(n-f)/2},
\end{equation*}
as claimed.
\end{proof}

\subsection{Constant-Time Stabilization}

Up to now, we considered worst-case scenarios only. In practice, it is likely
that faulty nodes show not entirely arbitrary behavior. In particular, they might
still be partially following the protocol, not exhibit a level of coordination
that could only be achieved by a powerful central instance, or not be fully
aware of non-faulty nodes states. Moreover, it is unlikely that at the time when
a majority of the nodes becomes non-faulty, all their timeouts $R_2$ and $R_3$
have been reset recently. In such settings, stabilization will be much easier and
therefore be achieved in constant time with a large probability. It is
difficult, however, to name simple conditions that cover most reasonable cases.
Generally speaking, once (randomized) timeouts of duration $R_2$ or $R_3$ are
not ``messed up'' at non-faulty nodes anymore, faulty channels and nodes need to
collaborate in an organized manner in order to prevent stabilization for a large
time period. We give a few examples in the following corollary.
\begin{corollary}\label{coro:quick_simple}
Suppose $W\subseteq V$, where $|W|\geq n-f$, satisfies that for each $i\in W$,
all (randomized) timeouts of duration $R_2$ or $R_3$ are correct during
$[t^-,t^+]$, and the node is non-faulty during
$[t^-+\vartheta(3(R_2+3d)+2(8(1-\lambda)R_2+d)),t^+]$. Moreover, channels
between nodes in $W$ are correct during
$[t^-+\vartheta(3(R_2+3d)+2(8(1-\lambda)R_2+d)),t^+]$ and did not insert
\init signals that have not been sent during
$[t^-,t^-+\vartheta(3(R_2+3d)+2(8(1-\lambda)R_2+d))]$ or delay them by more
than $R_1$ time. Define
$\tilde{t}^-:=t^-+\vartheta(3(R_2+3d)+2(8(1-\lambda)R_2+d))+R_1+d$. Moreover,
assume that one of the following statements holds during $[\tilde{t}^-,t^+]$.
\begin{itemize}
  \item [(i)] Nodes in $V\setminus W$ switch to \init at times that are
  independently distributed with probability density at most $\BO(1/(R_1 n))$,
  and channels from $V\setminus W$ to $W$ do not generate \init signals on
  their own (or delay \init signals from before $\tilde{t}^-$ more than $R_1$
  time).
  \item [(ii)] Channels from $V\setminus W$ to $W$ switch to \init at times
  that are independently distributed with probability density at most
  $\BO(1/(R_1 n^2))$.
  \item [(iii)] Channels from $V\setminus W$ to $W$ switch to \init obliviously
  of the history of signals originating at nodes in $W$ and do not know the time
  $\tilde{t}^-$.
\end{itemize}
If $t^+\in \tilde{t}^-+\Omega(k R_1)$, $k\in \N$, then there is a
$W$-stabilization point during $[\tilde{t}^-,\tilde{t}^-+\BO(k R_1)]$ with
probability at least $1-2^{-\Omega(k)}$.
\end{corollary}
\begin{proof}
In Theorems~\ref{theorem:stability} and~\ref{theorem:stabilization}, we showed
that stabilization is deterministic once a good resynchronization point occurs.
The notion of coherency essentially states that at non-faulty nodes, each
timeout expired at least once and has not been reset again because of incorrect
observations on other non-faulty nodes' states until the set is considered
coherent (cf.~\lemmaref{lemma:counters}). Subsequently, the respective nodes are
non-faulty and the channels connecting them correct. This is true by the
prerequisites of the corollary, which essentially state respective conditions on
timeouts $R_2$ and $R_3$ explicitly, while rephrasing the conditions for
coherency for the remaining timeouts (note that $R_1$ is the largest timeout
except for $R_2$ and $R_3$).

Moreover, the time span during which $R_2$ and $R_3$ behave and are observed
regularly is large enough for $R_3$ to expire twice and additional $R_2+3d$ time
to pass. This accounts for the fact that in the proof of
\theoremref{theorem:resync}, we essentially first wait until $R_3$ expires once
(so the adversary has no useful information on the timeout at the respective
node anymore) and then consider the subsequent time(s) when it expire(s). The
proof then exploits that non-faulty nodes timeout $R_3$ will expire at
roughly independently uniformly distributed points in time. Therefore, unless
faulty nodes or channels interfere, the statement of the corollary holds.

Hence, we need to show that for any of the three conditions, there is not too
much meddling from outside $W$. For Conditions~(i) or~(ii), we see that the
probability that there are no \init signals on channels from $V\setminus W$ to
$W$ at all for any time span of length $\BO(R_1)$ is at least constant,
regardless of the time interval considered. Regarding Condition~(iii), recall
that \theoremref{theorem:resync} essentially shows that whatever the strategy of
the adversary, the expected number of good $W$-resynchronization points during a
time interval (where $W$ is coherent) is linear in the size of the interval
divided by $R_1$ if the interval is sufficiently large. Since the adversary is
oblivious of the current time in relation to $\tilde{t}^+$ and the state of $W$,
the statement that for any strategy of the adversary the amortized number of
good stabilization points per $R_1$ time is constant yields the claim of the
lemma.
\end{proof}
We remark that this observation is particularly interesting as the core routine
of the algorithm is independent of the resynchronization routine after
stabilization. If at some time $W$ becomes subject to a large number of faults
resulting in loss of synchronization, however the resynchronization routine
still works properly, it is very likely that $W$ will recover within $\BO(1)$
time (provided $R_1\in \BO(1)$). On the other hand, if the resynchronization
routine fails in the sense that a majority of the nodes suffers from faulty
timeouts $R_2$ or $R_3$, or communication is faulty between too many nodes, this
will not affect the core routine unless too many components related to it fail
as well.

%% file: application.tex
\section{The FATAL\texorpdfstring{$^+$}{+} Protocol}\label{sec:application}

The synchronized pulses established by the FATAL pulse synchronization algorithm
could in principle serve  as the local clock signals provided to the application
layer of the SoC.\footnote{In order to establish a consistent global tick
numbering (needed for establishing a global notion of time across different
clock domains) of arbitrarily large bounded clocks, a self-stabilizing digital
clock synchronization algorithm like the one from~\cite{HDD06:SSS} can be
employed. Implementing such algorithms in SoCs is part of our future work and
thus outside the scope of this paper, however.} However, just using the FATAL
protocol in this way would result in a very low clock frequency: Despite the
fact that the time between pulses is $\Theta(d)$ (if the timeouts are chosen
accordingly) and thus asymptotically optimal, the actual clock speed would be
several orders of magnitude below the upper bound resulting from
\cite{lundelius84}, due to the large implied constants. Moreover, the system
model introduced in \sectionref{sec:model} assumes that delays may vary
arbitrarily between $0$ and $d$, with $d$ also covering the fairly complex
implementation of communicating the main algorithm's states (see
\sectionref{sec:communication}). By contrast, pure \emph{wire} delays of the
communication channels between different nodes are much smaller and also vary
within a smaller range.

This section contains an extension of FATAL, termed FATAL$^+$,
which overcomes these limitations. In a nutshell, it consists of
adding a fast non-self-stabilizing, Byzantine-tolerant algorithm
termed \emph{quick cycle} to FATAL, which generates exactly
$M>1$ fast clock ticks between any two pulses at a correct node after
stabilization.

\subsection*{The Quick Cycle Algorithm}

Consider a system of $n$~nodes, each of which runs the FATAL pulse
synchronization protocol. Additionally, each node is equipped with an instance
of the \emph{quick cycle} state machine depicted in \figureref{fig:top_alg}. The
interface between the quick cycle algorithm and the underlying FATAL pulse
synchronization protocol is by means of two signals only, one for each direction
of the communication: (i) The quick cycle state machine generates the $\Next$
signal by which it (weakly) influences the time between two successive pulses
generated by FATAL, and (ii) it observes the state of the $(T_2^+,\acc)$ signal,
which signals the expiration of an additional timer added to the FATAL protocol.
The timer is coupled to the state $\acc$ of FATAL, in which the pulse
synchronization algorithm generates a new pulse. The signal's purpose is to
enforce a consistent reset of the quick cycle state machine once FATAL has
stabilized.

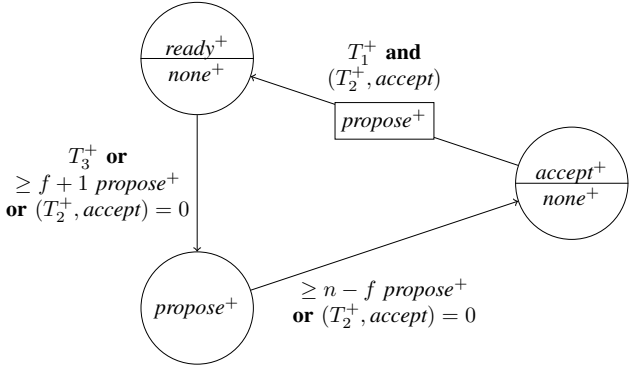
\begin{figure}[tb]
\centering
\scalebox{.83}{
\begin{tikzpicture}
\draw (6,2) node[circle,minimum width=1.8cm,draw,align=center] (ac) {$\acc^+$\\$\none^+$};
\draw (0,4) node[circle,minimum width=1.8cm,draw,align=center] (rd) {$\rdy^+$\\$\none^+$};
\draw (0,0) node[circle,minimum width=1.8cm,draw,align=center] (pr) {$\prop^+$};

\draw (ac.west) -- (ac.east);
\draw (rd.west) -- (rd.east);

\path[->] (ac) edge node[above=0.35cm,align=center] {$T^+_1$ \textbf{and}\\$(T_2^+,\acc)$} (rd);
\path[->] (rd) edge node[left,align=center] {$T^+_3$ \textbf{or}\\$\ge f+1$
$\prop^+$\\\textbf{or} $(T_2^+,\acc)=0$} (pr);
\path[->] (pr) edge node[below=0.4cm,align=center] {$\ge n-f$ $\prop^+$\\\textbf{or} $(T_2^+,\acc)=0$} (ac);

\draw (ac)+(-3,1) node[fill=white,draw=black] {$\prop^+$};
\end{tikzpicture}
}
\caption{The quick cycle of the FATAL$^+$ protocol.}\label{fig:top_alg}
\end{figure}

Essentially, the quick cycle state machine is a copy of the outer cycle of
\figureref{fig:main} that is stripped down to the minimum. However, an
additional mechanism is introduced in order to ensure stabilization, namely,
some coupling to the \acc\ state of the main algorithm: Whenever a pulse is
generated by FATAL, we require that all nodes switch to the \accp state
unless they already occupy that state. This is easily achieved by incorporating
the state of the expiration signal of the additional FATAL timer $(T_2^+,\acc)$
in the guards of \figureref{fig:top_alg}. Since pulses are synchronized up to
the skew $\Sigma$ of the pulse synchronization routine, it follows that all
nodes switch to \accp within a time window of $\Sigma+2d$. Subsequently, all
nodes will switch to state \rdyp before the first one switches to \propp
provided that $T_3^+$ is sufficiently large, and the condition that $f+1$ \propp
signals trigger switching to \propp guarantees that all nodes switch to \accp in
a tightly synchronized fashion.

One element that is not depicted explicitly in \figureref{fig:top_alg} is that
nodes increase an integer cycle counter by one whenever they switch to \accp.
The counter is reset to zero whenever $(T_2^+,\acc)$ expires, i.e., shortly
after a pulse generated by the underlying pulse synchronization algorithm. The
algorithm makes sure that, once the compound algorithm stabilized, these resets
never happen when the counter holds a non-zero value. The counter operates mod
$M\in \N$, where $M$ is large enough so that at least roughly $T_2+T_3$ and at
most $(T_2+T_4)/\vartheta$ time passed since the most recent pulse when it
reaches $M\equiv 0$ again. Whenever the counter is set to $0$, node $i\in V$
will set its \Next[_i] signal to $1$ and switch it back to $0$ at once (thus
raising the respective \Next[_i] memory flag of the main algorithm). Thus, by
actively triggering the next pulse, we ensure that a pulse does not occur at an
inconvenient point in time: When the system has stabilized, exactly $M$ switches
to \accp\ of the quick cycle algorithm occur between any two consecutive pulses
at a correct node. As these switches occur also synchronously at different
nodes, it is apparent that the quick cycle state machine in fact implements a
bounded-size synchronized clock.

To derive accurate bounds on the skew of the protocol, we need to state the
involved delays more carefully.
\begin{definition}[Refined Delay Bounds]
The state of the quick cycle algorithm is communicated via separate channels
$S_{i,j}^+$, with $i,j\in V$, whose delays vary within $d_{\min}^+$ and
$d_{\max}^+$ in order to be considered correct during $[t^-,t^+]$. State
transitions of the quick cycle state machine, resets of its timeouts, and
clearance of its memory flags take at most $d_{\max}^+$ time.
\end{definition}

Setting $\Sigma^+ := 2d_{\max}^+-d_{\min}^+$, we assume that the following
timing constraints hold:
\begin{eqnarray}
T_1^+ &\ge & \vartheta (T_2^++\Sigma^++3d+d_{\max}^+)\label{eq:T1+}\\
T_2^+ &\ge & \vartheta (3d+3d_{\max}^+)\label{eq:T2+}\\
T_3^+ &\ge &\vartheta (T_1^++d^+_{\max})\label{eq:T3+}\\
M &\in &
\scriptstyle\left[\frac{\vartheta(T_2+T_3+3d)+T_1^+-T_2^+}{T_1^++T_3^+},
\frac{T_2+T_4-3\vartheta d}
{T_1^++T_3^++\Sigma^++3d_{\max}^+}\right]\label{eq:M}
\end{eqnarray}
It follows from \lemmaref{lemma:constraints} that it is always possible to pick
appropriate values for the timeouts and~$M$. Note, however, that choosing $M\in
\omega(1)$ requires that $T_2+T_3\in \omega(1)$, resulting in a superlinear
stabilization time. More precisely, the stabilization time of FATAL$^+$ is,
given $M$ and minimizing the timeouts under this constraint, in $\Theta(M n)$.
As mentioned previously, this limitation can be overcome by employing a digital
clock synchronization such as~\cite{BDH08:podc}.

We now prove the correctness of the FATAL$^+$ protocol.

\begin{theorem}\label{theorem:quick}
Let $W\subseteq V$, where $\lvert W\rvert \ge n-f$, and define $T(k)$, for $k\in
\N$, as in \corollaryref{coro:stabilization}. Then, for any $k\in
\N$, the FATAL$^+$ protocol is a $(W,W^2)$-stabilizing pulse synchronization
protocol (where \accp is the ``pulse'' state) with skew
$\Sigma^+$ and accuracy bounds
$(T_1^++T_3^+)/\vartheta-\Sigma^+$ and
$T_1^++T_3^++\Sigma^++3d_{\max}^+$. It stabilizes
within time $T(k)+T_1^++T_3^++\Sigma^++3d+3d_{\max}^+$ with
probability at least $1-2^{-k(n-f)}$. Moreover, the cycle
counters increase by exactly one mod $M$ at each pulse, within a time window of
$\Sigma^+$, and both the quick cycle state machine and the cycle counters are
metastability-free once the protocol stabilized and remains fault-free in $W$.
\end{theorem}
\begin{proof}
Assume that nodes in $W\subseteq V$, where $|W|\geq n-f$, are non-faulty and
channels between them are correct during $[t^-,t^+]$, where $t^+\geq
t^-+T(k)+T_1^++T_3^++\Sigma^++3d+3d_{\max}^+$. According to
\corollaryref{coro:stabilization}, with probability at least $1-2^{-k(n-f)}$, there
exists a time $t_0\in [t^-,t^-+T(k))]$ such that all nodes in $W$ switch to
$\acc$ within $[t_0,t_0+2d)$, and they will continue to switch to \acc regularly
in a synchronized fashion until at least $t^+$. For the remainder of the proof,
we assume that such a time $t_0$ is given; from here we reason
deterministically.

The skew bound is shown by induction on the $k$-th consecutive quick cycle
pulse, where $k\in \N$, generated after the stabilization time $t_0$ of the
FATAL algorithm. Note that the time for which we are going to establish that the
compound algorithm stabilizes is $t_1>t_0$; here we denote for $k\ge 1$ by
$t_k$ the time when the first node from $W$ switches to \accp for the $k^{th}$
time after $t_0+3d$, i.e., the beginning of the $k^{th}$ pulse of FATAL$^+$ that
we prove correct. W.l.o.g.\ we assume that $t^+=\infty$; otherwise, all
statements will be satisfied until $t^+$ only (which is sufficient).

To prove the theorem, we are going to show by induction on $k\in \N$ that
\begin{itemize}
  \item [(i)] $t_1\in [t_0+(T_2^++T_3^+)/\vartheta,
  t_0+T_1^++T_3^++\Sigma^++3d+3d_{\max}^+]$,
  \item [(ii)] if $k\geq 2$, $t_k\leq
  t_{k-1}+T_1^++T_3^++\Sigma^++3d_{\max}^+$,
  
  \item [(iii)] if $ k\geq 2$, $t_k\geq t_{k-1}+(T_1^++T_3^+)/\vartheta$,
  
  \item [(iv)] $\forall i\in W: i$ switches to \accp for the
  $k^{th}$ time after $t_0+3d$ during $[t_k,t_k+\Sigma^+)$,
  
  \item [(v)] if $ k\geq M$, for $l:=\lfloor k/M\rfloor$, $\forall i\in W: i$ switches
  to \acc for the $l^{th}$ time after $t_0+3d$ during
  $[t_{Ml},t_{Ml}+\Sigma^++2d)$,
  
  \item [(vi)] $\forall i\in W:$ node $i$'s cycle counter switches from
  $k-1\operatorname{mod} M$ to $k\operatorname{mod} M$ at some time from
  $[t_k,t_k+\Sigma^+)$, and
  
  \item [(vii)] if $k\geq 2$, $\forall i\in W:$ node $i$'s cycle counter
  changes its state exactly once during $[t_{k-1},t_k)$.
\end{itemize}
In particular, the protocol is a pulse synchronization protocol with the claimed
bounds on skew, accuracy, and stabilization time. Proving these properties will
also reveal that quick cycle is metastability-free after time $t_1$.

To anchor the induction at $k=1$, we need to establish Statement (i) as well as
Statements (iv) and (vi) for $k=1$; the remaining
statements are empty for $k=1$.

Recall that any node $i\in W$ switches to \acc during $[t_0,t_0+2d)$. Hence,
during
\begin{equation*}
\left[t_0+3d,t_0+\frac{T_2^+}{\vartheta}\right)\sr{eq:T2+}{\subseteq}
[t_0+3d,t_0+3d+3d_{\max}^+),
\end{equation*}
at no node in $W$, $(T_2^+,\acc)$ is expired, implying that all nodes in $W$ are
in state \accp\ during $[t_0+3d+2d_{\max}^+,t_0+3d+3d_{\max}^+)$. Note that each
node will reset its cycle counter to $0$ when $(T_2^+,\acc)$ expires, i.e.,
after having completed its transition to \accp.

The above bound shows that at the minimal time after $t_0+3d$ when a node in $W$
switches to \rdyp, it is guaranteed that no node is observed in \propp until the
minimal time $t_p\geq t_0+3d$ when a node in $W$ switches to \propp. Moreover,
at any node switching to state \rdyp\ timeout $(T_2^+,\acc)$ must be expired,
implying that the node may not switch from \rdyp to \propp due to this signal
until it switches to \acc again. Recall that nodes set their \Next signals to
$1$ only briefly when their cycle counters are set to $0$. Hence, for each such
node in $W$, this signal is observed in state $0$ from the time when
$(T_2^+,\acc)$ expires until (a) at least time $t_M$ or (b) the time the node is
forced by a switch to \acc to set its counter to $0$, whatever is earlier.
Examining the main state machine, it thus can be easily verified that no node in
$W$ may switch from \rdyp to \propp because $(T_2^+,\acc)=0$ before (a) time
$t_M$ or (b) time
\begin{equation}
t_0+\frac{T_2+T_4}{\vartheta}\sr{eq:M}{>} t_0+M(T_1^++T_3^++3d_{\max}^+)+3d
\label{eq:no_forced}
\end{equation}
is reached. We obtain:

\smallskip

(P1) No node in $W$ observes $(T_2^+,\acc)(t)=0$ at some time
$t\in [t_0+3d,\min\{t_M,t_0+M(T_1^++T_3^++3d_{\max}^+)+3d\}]$ when it is not
in state \accp.

\smallskip

Considering that any node $i\in W$ will switch to \rdyp once both $T_1^+$ and
$T_2^+$ expired and subsequently to \propp at the latest when $T_3^+$ expires
(provided that it does not switch back to \accp first), it follows that by time
\begin{eqnarray}
&&t_0+3d+\max\{T_1^++d_{\max}^+,T_2^+\}+T_3^++2d_{\max}^+\nonumber\\
&\sr{eq:T1+}{=}&t_0+T_1^++T_3^++3d_{\max}^++3d\label{eq:t_1_eq}\\
&\sr{eq:no_forced}{<}&t_0+\frac{T_2+T_4}{\vartheta},\label{eq:t_1_upper}
\end{eqnarray}
each node in $W$ must have been observed in \propp at least once. On the other
hand, as we established that nodes do not observe nodes in $W$ in state \propp
when switching to \rdyp at or after time $t+3d$ before the first node in $W$
switches to \propp, it follows that until time
\begin{equation}
t_0+\frac{T_2^++T_3^+}{\vartheta}\sr{eq:T3+}{\geq}
t_0+3d+T_1^++2d_{\max}^+,
\end{equation}
nodes in $W$ will have at most $|V\setminus W|\leq f$ of their \propp
flags in state $1$, and their timeout $T_3^+$ did not expire yet.
Thus, by (P1), the first node in $W$ that switches to \propp after $t_0+3d$
must do so at time $t_p \ge t_0+3d+T_1^++2d_{\max}^+$.

Recall that $t_1$ is the minimal time larger than $t_0+3d$ when a node
in $W$ switches to state \accp.
By \eqref{eq:t_1_eq} and since $|W|\geq n-f$, we have that each node
in $W$ observes at least $n-f$ nodes in $\propp$ by time
$t_0+T_1^++T_3^++3d+3d_{\max}^+$, and thus
\begin{equation}
t_1\leq t_0+T_1^++T_3^++3d+3d_{\max}^+\enspace.\label{eq:t1_lower}
\end{equation}
Moreover, we can trivially bound
\begin{equation}
t_1\geq t_p\geq t_0+(T_2^++T_3^+)/\vartheta\enspace.\label{eq:t1_upper}
\end{equation}
From \eqref{eq:t1_lower} and \eqref{eq:t1_upper} it follows that $t_1$
satisfies Statement~(i) of the claim.

Since at time $t_1$ there is a node $i\in W$ switching from \propp to \accp,
(P1) implies that it must memorize at least $n-2f\geq f+1$ nodes in $W$ in state
\propp, which must have switched to this state during $[t_p,t_1-d_{\min}^+]$. By
the above considerations regarding the reset of the \propp flags, this yields
that all nodes in $W$ will memorize at least $f+1$ nodes in state \propp by time
$t_1+d_{\max}^+-d_{\min}^+$ and thus switch to \propp (if they have not done so
yet). It follows that by time $t_1+2d_{\max}^+-d_{\min}^+=t_1+\Sigma^+$, all
nodes in $W$ memorize at least $|W|\geq n-f$ nodes in \propp and therefore
switched to \accp. Hence, we successfully established Statement (iv) of the
claim for $k=1$. Statement (vi) follows for $k=1$, as the cycle counters have
been reset to zero at the expiration of $(T_2^+,\acc)$ and are increased upon
the subsequent state transition to \accp. Note that Statements (ii), (iii), (v),
and (vii) trivially hold.

\medskip

We now perform the induction step from $k\in \N$ to $k+1$. Assume that
Statements (ii) to (vii) hold for all values smaller or equal to $k$; Statement
(i) only applies to $k=1$ and was already shown. Define $l:=\lfloor k/M\rfloor
\ge 0$. Thus, if we can show Statement (ii) for $k+1$, we may infer that
\begin{eqnarray}
\!\!&&t_{k+1}\notag\\
\!\!&\stackrel{\!\!\!\!\!(i),(ii)\!\!\!\!\!\!\!}{\leq}&\!\!
t_{Ml}+(k+1-Ml)(T_1^++T_3^++\Sigma^++3d_{\max}^+)+3d\notag\\
\!\!&\leq & \!\!
t_{Ml}+M(T_1^++T_3^++\Sigma^++3d_{\max}^+)+3d\label{eq:k1_kp_pre}\\
\!\!&\sr{eq:M}{\leq}&\!\!
t_{Ml}+\frac{T_2+T_4}{\vartheta}.\label{eq:k1_kp}
\end{eqnarray}
In case $l = 0$, it holds that $k < M$ and we may deduce (P1) by the
same arguments as in the induction basis.

In case $l \ge 1$, we use Statement (v) for value $k$, and, by
analogous arguments as in the induction basis, deduce that at no
node in $W$, $(T^+_2,\acc)$ is expired during
$[t_{Ml}+3d,t_{Ml}+3d+2d^+_{\max})$, implying that all nodes in
$W$ are in $\accp$ during that time.
Repeating the reasoning of the induction basis before (P1) with $t_0$
replaced by $t_{Ml}$, $t_1$ replaced by $t_{k}$, and $t_M$
replaced by $t_{Ml+M}$ shows that:

\smallskip

(P1') No node in $W$ observes $(T_2^+,\acc)(t)=0$ at some time
$t\in [t_{Ml}+3d,\min\{t_{Ml+M},t_{Ml}+M(T_1^++T_3^++3d_{\max}^+)+3d\}]$ when
it is not in state \accp.

\smallskip


Since further $t_{Ml+M}\geq t_{k+1}$ by definition of $l$, we obtain
from (P1') that no node $i\in W$ will memorize $\Next[_i]=1$
earlier than time $\min\{t_{k+1},t_{Ml}+M(T_1^++T_3^++3d_{\max}^+)+3d\}$ (again
by reasoning analogously to the induction base).

By Statement (iv) for the value $k$, we know that each node $i\in W$
switches to \accp during $[t_k,t_k+\Sigma^+)$.
In particular, $i$ will increase its cycle counter at the respective
time, i.e., Statement (vi) for $k+1$ follows at once if we
establish Statement (vii) for $k+1$.
As Statement (iv) for the value $k$ together with Statement (ii) for
value $k+1$ imply that each node switches to \accp exactly once
during $[t_k,t_{k+1})$, Statement (vii) for $k+1$ follows,
provided that we can exclude that the counter is reset to $0$,
due to $(T_2^+,\acc)$ expiring, at a time when it holds a non-zero
value.

We now show that this never happens. By Statement (v) for value $k$ each node
$i\in W$ switches to \acc during
\begin{equation}\label{eq:acc_sw}
[t_{Ml},t_{Ml}+\Sigma^++2d)
\end{equation}
and this time is unique during $[t_{Ml},t_{k+1})$ due to
\eqref{eq:k1_kp_pre}.

Because of \eqref{eq:acc_sw} a node in $W$ will reset its timeout
$(T^+_2,\acc)$ during
\begin{equation}
[t_{Ml},t_{Ml}+\Sigma^++3d)\enspace,\label{eq:acc_re}
\end{equation}
and $(T_2^+,\acc)$ will expire within
\begin{eqnarray*}
&&\left[t_{Ml}+\frac{T_2^+}{\vartheta},
t_{Ml}+T_2^++\Sigma^++3d+d_{\max}^+\right)\\
&\sr{eq:T1+}{\subseteq}&\left[t_{Ml}+\frac{T_2^+}{\vartheta},
t_{Ml}+\frac{T_1^+}{\vartheta}\right)\\
&\sr{eq:T2+}{\subseteq} & [t_{Ml}+3d+2d_{\max}^+,t_{Ml+1}) \enspace.
\end{eqnarray*}

Thus, no node in $W$ leaves state \accp after switching there for the
$(Ml)^{th}$ time after $t_0+3d$ before observing that
$(T_2^+,\acc)$ is reset and expires again.
In particular, this shows that the counters are only reset to $0$ at
times when they are $0$ anyway. Granted that Statement (ii) holds for
$k+1$, Statement (vii) for $k+1$ follows.

\medskip

Next, we establish Statements (ii) to (iv) for $k+1$.
We reason analogously to the case of $k=1$, except that we have to
revisit the conditions under which state \accp is left.
As we have just seen, nodes switch from \accp to \rdyp upon $T_1^+$
expiring. Thus, as all nodes in $W$ switch to \accp during $[t_k,t_k+\Sigma^+)$,
they switch to $\rdyp$ within the time window
$[t_k+T_1^+/\vartheta, t_k+T_1^++\Sigma^++d_{\max}^+)$.
By time
\begin{equation*}
t_k+\frac{T_1^+}{\vartheta}\sr{eq:T1+}{\geq} t_k+\Sigma^++d_{\max}^+,
\end{equation*}
all nodes in $W$ will be observed in \accp (and therefore not
in \propp), together with (P1') preventing that the first node in $W$ that
(directly) switches from \rdyp to \propp afterwards does so without
$T_3^+$ expiring first.

More precisely, according to (P1') no node in $W$ observes $(T_2^+,\acc)$ to be
zero until time $\min\{t_{k+1},t_{Ml}+M(T_1^++T_3^++3d^+_{\max})+3d\}$. As
showing Statement (ii) for $k+1$ will imply \inequalityref{eq:k1_kp_pre}, we can
w.l.o.g.\ disregard the case that $t_{Ml}+M(T_1^++T_3^++3d^+_{\max})+3d<t_{k+1}$
in the following. Thus, each node $i\in W$ will be observed in state \propp no
later than time $t_k+T_1^++T_3^++\Sigma^++3d_{\max}^+$. As argued for $k=1$, it
follows that indeed $t_{k+1}\leq t_k+T_1^++T_3^++\Sigma^++3d_{\max}^+$, i.e.,
Statement (ii) for $k+1$ holds. Further each node $i$ switches to \accp for the
$(k+1)^{th}$ time after $t_0+3d$ during time $[t_{k+1},t_{k+1}+\Sigma^+)$, i.e.,
Statements (iv) for $k+1$ holds. Statement (iii) for $k+1$ is deduced from the
fact that it takes at least $(T_1^++T_3^+)/\vartheta$ time until the first node
from $W$ switching to \propp after $t_k+\Sigma^+$ does so, since timeouts
$T_1^+$ and $T_3^+$ need to be reset and expire first, one after the other.

\medskip

Finally, we need to establish Statement (v) for $k+1$. If $M$ does not divide
$k+1$, Statement (v) for $k+1$ follows from Statement~(v) for $k$. Otherwise $M$
does divide $k+1$ and we can bound\footnote{Note that we already build on
Statement (iii) for $k+1$ here.}
\begin{eqnarray*}
&&t_{k+1}+\Sigma^++d\\
&\stackrel{(i),(iii)}{\geq}&
t_{Ml}+\frac{M(T_1^++T_3^+)-T_1^++T_2^+}{\vartheta}+\Sigma^++d\\
&\sr{eq:M}{\geq}& t_{Ml}+T_2+T_3+\Sigma^++4d\enspace.
\end{eqnarray*}
As by Statement (v) for $k$ all nodes in $W$ switched to \acc during
$[t_{Ml},t_{Ml}+\Sigma^++2d)$, we conclude\footnote{This
statement relies on the constraints on the main state machines'
timeouts, which require that $T_2$ expiring is the critical
condition for switching to \rdy.} from the main state machines'
description that all nodes in $W$ are observed in state \rdy with
timeout $T_3$ being expired (or already switched to \prop or even
\acc) by time $t_{Ml}+T_2+T_3+\Sigma^++4d$.
Because all their \Next signals switch to one during
$[t_{k+1},t_{k+1}+\Sigma^+)$, all nodes in $W$ must therefore
have switched to \prop by time $t_{k+1}+\Sigma^++d$.
Consequently,\footnote{For details, we refer to the analysis of the
FATAL algorithm in the full paper.} as we stated that w.l.o.g.\
they do not switch to \acc again before time $t_{k+1}$, they do
so at times in $[t_{k+1},t_{k+1}+\Sigma^++2d)$ as claimed.

\medskip

This completes the induction. According to Statement (i), $t_1$ satisfies the
claimed bound on the stabilization time. With respect to this time, Statement
(iv) provides the skew bound, and combining it with Statements (ii) and (iii),
respectively, yields the stated accuracy bounds. Statements (vi) and (vii) show
the properties of the counters. Metastability-freedom of the state machine is
trivially guaranteed by the fact that each state has a unique successor state.
For the counter, we can infer metastability-freedom after stabilization from
the observation made in the proof that for times $t\geq t_1$, $(T_2,\acc)(t)=0$
at a non-faulty node implies that it is in state \accp with its cycle counter
equal to zero. This completes the proof.
\end{proof}

For some applications, one might require an even higher operational frequency
than provided by the quick cycle state machine. It turns out that there is a
simple solution to this issue.

\subsection*{Increasing the Frequency Further}
Given any pulse synchronization protocol, one can derive clocks operating at an
arbitrarily large frequency as follows. Whenever a pulse is triggered locally,
the nodes start to increase a local integer counter modulo some value $m\in \N$
at a speed of $\phi\in \R^+$ times that of a local clock, starting from~$0$.
Denote by $T^-$ the accuracy lower bound of the protocol and suppose that the
local clock controlling the counter runs at a speed between $1$ and $\rho\in
(1,\vartheta]$, i.e., its maximum drift is $\rho-1$.\footnote{We introduce
$\rho$ since one might want to invest into a single, more accurate clock source
per node in order to obtain smaller skews.} Once the counter reaches the value
$m-1$, it is halted until the next pulse. We demand that
\begin{equation}
m \leq \phi T^-. \label{eq:m}
\end{equation}
This approach is similar to the one presented in \cite{DD06}, enriched by
addressing the problem of metastability.

In the context of the FATAL$^+$ protocol, we get the following result.
\begin{corollary}\label{coro:clock}
Adding a counter as described above to the FATAL$^+$ protocol and concatenating
the counter values of the two counters at node $i\in V$ yields a bounded logical
clock $L_i\in \{0,\ldots,mM-1\}$. At any time $t$ when the protocol has
stabilized on some set $W$ (according to \theoremref{theorem:quick}), it holds
for any two nodes $i,j\in W$ that
\begin{equation*}
|L_i(t)-L_j(t) \operatorname{mod} mM|\leq \left\lceil
\phi\Sigma^++\left(1-\frac{1}{\rho}\right)m \right\rceil.
\end{equation*}
Once stabilized, these clocks do not ``jump'', i.e., they always increase by
exactly one mod $mM$, with at least $1/\rho$ time between any two
consecutive ``ticks''.

The amortized clock frequency is within the bounds
$m/(T_1^++T_3^++\Sigma^++3d_{\max}^+)$ and $\vartheta
m/(T_1^++T_3^+-\Sigma^+)$. Viewed as a state machine in our model, the clocks
$L_i$, where $i\in W$, are metastability-free after stabilization.
\end{corollary}
\begin{proof}
Observe that it takes at least $m/(\phi\rho)$ time for one of the new counters
to increase from $0$ to $m$. Since the counters are restarted at pulses, which
are triggered locally at most $\Sigma^+$ time apart, at the time when a ``fast''
node arrives at the value $m$, a ``slow'' node will have increased its clock by
at least $\lfloor m/\rho-\phi\Sigma^+\rfloor$. According to \inequalityref{eq:m},
slow nodes will be able to increase their counters to $m$ before the next pulse.
The claimed bound on the clock skew and the facts that clock increases are one
by one and at most every $1/\rho$ time follow.

The bound on the amortized clock frequency follows by considering the minimal
and maximal times $M$ iterations of the quick cycle may require.

The metastability-freedom of the clock is deduced from the metastability-freedom
of the individual counters. For the new counter this is guaranteed by
\inequalityref{eq:m}, since the counter is always halted at $0$ before it is
reset due to a new quick cycle pulse.
\end{proof}

We remark that in an implementation, one would probably utilize the better clock
source, if available, to drive $T_1^+$ and $T_3^+$ as well.\footnote{Since the
new counter is started together with $T_1^+$, this does not incur metastability.
Special handling is required for $T_3^+$ on the $M^{th}$ pulse of the quick
cycle, though.} Maximizing $m$ with respect to \inequalityref{eq:m} and choosing
$T_1^++T_3^+$ sufficiently large will thus result in clocks whose amortized
drift is arbitrarily close to $\rho$, the drift of the underlying local clock
source.

%% file: implementation.tex
\newcommand{\LOW}{{LOW}}
\newcommand{\HIGH}{{HIGH}}
\newcommand{\true}{{true}}
\newcommand{\false}{{false}}
\newcommand{\Tr}{\mbox{\emph{Tr}}}
\newcommand{\TSMClock}{\emph{TSMClock}}
\newcommand{\TSMCStop}{\emph{TSMCStop}}
\newcommand{\InCLK}{\emph{InCLK}}

\section{Implementation}
\label{sec:implementation}

In this section, we provide an overview of our FPGA prototype implementation of
the FATAL$^+$ protocol. The purposes of this implementation are (i)
to serve as a proof of concept, (ii) to validate the predictions of 
the theoretical analysis, and (iii) to form a basis for the future development of 
protocol variants and engineering improvements. Rather than striving for
optimizing performance, area or power efficiency, our primary goal is hence
to essentially provide a direct mapping of the algorithmic description to 
hardware, and to evaluate its properties in various operating scenarios.

Our implementation does not follow the usual design practice, for several
reasons:

\emph{Asynchrony:} Targeting ultra-reliable clock generation in SoCs, the
implementation of FATAL$^+$ itself cannot rely on the availability of a
synchronous clock. Moreover, some performance-critical guards, like the one of
the transition from \prop to \acc in \figureref{fig:main}, are purely
asynchronous and should hence not be synchronized to a local clock. Even worse,
testing for activated guards synchronized to a local clock source bears the risk
to generate metastability, as remote signals originate in different clock
domains. On the other hand, conventional \emph{asynchronous state machines}
(ASM) are not well-suited for implementing
\figureref{fig:main}--\figureref{fig:top_alg} due to the possibility of choice
of successor states and continuously enabled (i.e., non-alternating) guards. Our
prototype relies on \emph{hybrid state machines} (HSM) that combine an ASM with
synchronous \emph{transition state machines} (TSM) that are started on demand
only.

\emph{Fault tolerance:} The presence of Byzantine faulty nodes
forced us to abandon the classic ``wait for all'' paradigm 
traditionally used for enforcing the indication principle 
in asynchronous designs: Failures may easily inhibit 
the completion of the request/acknowledge cycles typically used for
transition-based flow control. Timing constraints, established 
by our theoretical analysis, in conjunction with state-based communication
are resorted to in order to establish event ordering and synchronized 
executions in FATAL$^+$.

\emph{Self-Stabilization:} In sharp contrast to non-stabilizing algorithms,
which can always assume that there is a (substantial) number of non-faulty
nodes that run approximately synchronously and hence adhere to certain 
timing constraints, self-stabilizing algorithms cannot even assume this. 
Although FATAL$^+$ guarantees that non-faulty nodes will eventually
execute synchronously, even when started from an arbitrary state, 
the violation of timing constraints and hence metastability 
\cite{Mar81} cannot be avoided during stabilization. 
For example, state \acc in \figureref{fig:main} 
has two successors \slp and \rec, the guards of which could become true arbitrarily 
close to each other in certain stabilizing scenarios.
This is acceptable, though, as long as such
problematic events are neither systematic nor frequent, which is ensured
by the design and implementation of FATAL$^+$ (see \sectionref{sec:meta}).

\medskip

Inspecting \figureref{fig:main}--\figureref{fig:top_alg} reveals that the
state transitions of the FATAL$^+$ state machines are triggered by AND/OR combinations
of the following different types of conditions:
\begin{itemize}
  \item[(1)] A watchdog timer expires [``$(T_2,\acc)$''].
  \item[(2)] The state machines of a certain number ($1$, $\geq f+1$, or $\geq
  n-f$) of nodes reached a particular (subset of) state(s) at least once since
  the reset of the corresponding memory flags [``$\geq n-f~\acc$''].
  \item[(3)] The state machines of a certain number ($1$, $\geq f+1$, or $\geq
  n-f$) of nodes are currently in (one of) a particular (subset of) state(s)
  [``in \res''].
  \item[(4)] Always [``\true''].
\end{itemize}

The above requirements reveal the need for the following major building blocks:
\begin{itemize}
  \item Concurrent HSMs, implementing the states and transitions specified in the protocol.
  \item Communication infrastructure between those state machines, continuously conveying the state information.
  \item Watchdog timers (also with random timeouts) for implementing type (1) guards.
  \item Threshold modules and memory flags for implementing type (2) and type (3) guards.
\end{itemize}
Obviously, all these building blocks require implementations that match the assumptions
of the formal model in \sectionref{sec:model}. Apart from maintaining timing assumptions like
an end-to-end communication delay bound $t-\tau_{i,j}^{-1}(t) < d$, this also
includes the need to implement all stateful components in a self-stabilizing
way: They must be able to eventually recover from an arbitrary erroneous
internal state, including metastability, when operating in the specified
environment.

Before we proceed with a description of the implementations of these
components, we discuss how FATAL$^+$ deals with the threat of
metastability arising from our extreme fault scenarios.

\subsection{Metastability issues}
\label{sec:meta}

Reducing the potential for both metastability generation and metastability
propagation are important goals in the design and implementation of FATAL$^+$.
Although it is impossible to completely rule out metastability generation in the
presence of Byzantine faulty nodes (which may issue signal transitions at
arbitrary times) and during self-stabilization (where all nodes may be
completely asynchronous), we nevertheless achieved the following properties:
\begin{itemize}
  \item[(I)] Guaranteed metastability-freedom in fault-free executions after
  stabilization.
  \item[(II)] Non-faulty nodes are safeguarded against ``attacks'' by faulty
  nodes that aim at inducing metastability, in particular once the system has
  stabilized.
  \item[(III)] Metastable upsets at non-faulty nodes are rare during
  stabilization, therefore delaying stabilization as little as possible.
  \item[(IV)] Very small windows of vulnerability and the possibility to
  incorporate additional measures for decreasing the upset probability further.
\end{itemize}

The following approaches have been used in FATAL$^+$ to accomplish these goals
(additional details will be given in the subsequent sections):

(I) is guaranteed by our proofs of metastability-freedom, which exploit the fact
that all non-faulty nodes run approximately synchronously after stabilization.
It is hence relatively straightforward to ensure, via timing constraints, that
some data from remote ASMs does not change while it is used.

(II) is accomplished by several means, which make it very difficult (albeit not 
impossible) for a faulty node to generate/propagate metastability. Besides
avoiding any explicit control flow between ASMs by communicating states only,
which greatly reduces the dependency of a non-faulty receiver node from a
faulty sender, several forms of logical masking of metastability are employed.
One example is the combination of memory flags and threshold gates, which ensure
that possibly upset memory flags are always overruled quickly by correct ones at the
threshold output. A different form of logical masking occurs due to the fact
that, after stabilization, all non-faulty nodes execute the outer cycle of
the main state machine (\figureref{fig:main}) only: Since the outer cycle does
not involve any type (3) guard once stabilization is achieved, any metastability
originating from the (less metastability-safe) resynchronization algorithm
(\figureref{fig:resync}) and its extension (\figureref{fig:extended}) is
completely masked.

To accomplish (III), the measures outlined in (2) are complemented by adding time masking using 
randomization: The resynchronization routine (\figureref{fig:resync}) 
tries to initialize recovery from arbitrary states at random, sufficiently 
sparse points in time. It is hence very unlikely that non-faulty nodes are kept from
stabilizing due to metastable upsets. Moreover, if at the beginning of the
stabilization process $f'<f<n/3$ nodes are faulty, up to $f-f'$ metastable upsets can be
tolerated without keeping the remaining nodes from stabilizing; the nodes that
became subject to newly arising transient faults will stabilize quickly once
$n-f$ nodes established synchronization (cf.~\theoremref{theorem:constant}).

Finally, (IV) is achieved by implementing all building blocks that are
susceptible to metastable upsets, like memory flags, in a way that minimizes
the window of vulnerability. Moreover,  elastic
pipelines acting as metastability filters \cite{FFS09:ASYNC09}
or synchronizers can be added easily to further protect such elements.

\subsection{State machine communication}
\label{sec:communication}

According to our system model, an HSM must be able to continuously communicate
its current state system-wide: It is requested
that every receiver is informed of the sender's current state within $d$ time
(resp.\ within $d_{\min}^+$ and $d_{\max}^+$ for the quick cycle algorithm).
For simplicity, we use parallel communication, by means of a suitably sized data
bus, in our implementation.\footnote{It is, however, reasonably easy to
replace parallel communication by serial communication, e.g., by extending the 
(synchronous) TSM appropriately.} Since a node treats itself like any other
node in type (2) and type (3) guards with thresholds, it comprises a complete
receiver as described below for every node in the system (including itself).

\figureref{fig:SendRec} shows the circuitry
used for communicating the current state of the main algorithm in
\figureref{fig:main}.
\begin{figure}
\centering
\includegraphics[width=\linewidth]{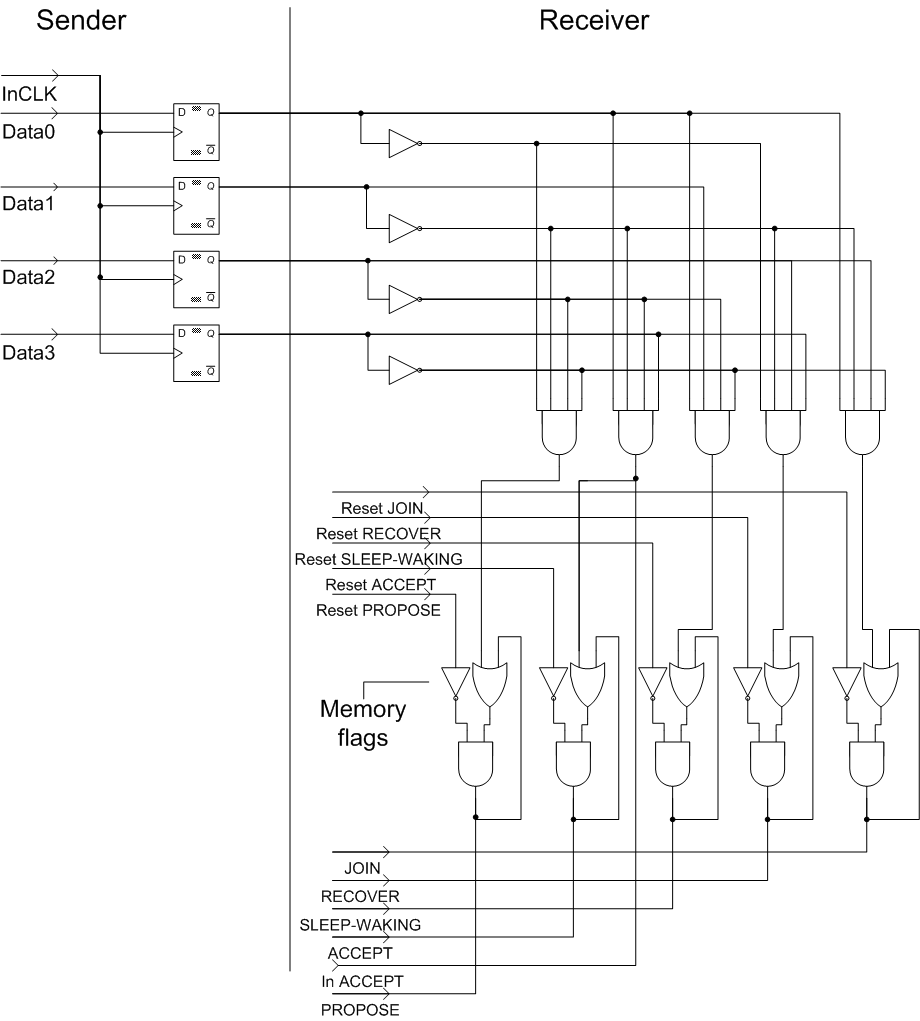}
\caption{Sender and single receiver (including memory flags) for the ASM of the main algorithm.}
\label{fig:SendRec}
\end{figure}
The sender consists of a simple array of flip-flops, which drive the parallel
data bus that thus continuously reflects the current state of the sender's HSM. 
In sharp contrast to handshake-based communication, reading at the receiver 
occurs without any coupling to the sender here. As argued in \sectionref{sec:meta},
the synchrony between non-faulty nodes guaranteed by the FATAL$^+$ protocol guarantees 
that the sender state data will always be stable when read after stabilization. For the
stabilization phase, we cannot give such a guarantee but take some (acceptable) 
risk of metastability. 

To avoid the unacceptable risk of reading and capturing false intermediate sender states
due to different delays on the wires of the data bus, delay-insensitive \cite{Ver88} state coding must be used. We have chosen the following encoding for the main state machine in \figureref{fig:main}:
\begin{center}
\begin{tabular}{llp{1cm}ll}
propose & 0000 &&
accept & 1001\\
sleep & 1011 &&
sleep $\to$ waking & 0011\\
waking & 0101 &&
ready & 0110\\
recover & 1100 &&
join & 1010\\
\end{tabular}
\end{center}
For the other state machines making up FATAL$^+$, it suffices to communicate only a
single bit of state information (\supp or \none in \figureref{fig:extended}, \init or \emph{wait}
in \figureref{fig:resync}, and \propp or \emph{none}$^+$ in \figureref{fig:top_alg}). Hence, every
bus consists of a single wire here, and the decoder in the receiver becomes trivial.

The receiver consists of a simple combinational decoder consisting of AND gates,
which generate a 1-out-of-$m$ encoding of the binary representation of the state
communicated via the data bus. The decoded signals correspond to a single sender
state each. This information is directly used for type (3) guards, and fed into
memory flags for type (2) guards. Every memory flag is just an SR-latch with
dominant reset, whose functional equivalents are also included in \figureref{fig:SendRec}.
Note that a memory flag is set depending on the state communicated by the sender, but 
(dominantly) cleared under the receiver's control. 

A memory flag may become metastable when the inputs change during stabilization
of its feedback loop, which can occur due to (a) input glitches and/or (b)
simultaneous falling transitions on both inputs. However, for correct receivers,
(a) can only occur in case of a faulty sender, and (b) is again only
possible during stabilization: Once non-faulty nodes execute the outer cycle of \figureref{fig:main}, it is guaranteed that e.g.\ all non-faulty nodes enter \acc\ before 
the first one leaves. The probability of an upset is thus very
small, and could be further reduced by means of an elastic pipeline acting as
metastability filter (which must be accounted for in the delay bounds). 

The most straightforward implementation of the threshold modules used for
generating the $\geq f+1$ and $\geq n-f$ thresholds in type (2) and type (3)
guards is a simple sum-of-product
network, which just builds the OR of all AND combinations of
$f+1$ resp.\ $n-f$ inputs. In our FPGA implementation, a threshold module
is built by means of lookup-tables (LUT); some dedicated experiments
confirmed that they work glitch-free for monotonic inputs (as provided
by the memory flags).

\subsection{Hybrid state machines}
\label{sec:SM}

Our prototype implementation of FATAL$^+$ relies on \emph{hybrid state machines} (HSM):
An ASM is used for determining, by asynchronously evaluating the guards, the points 
in time when a state transition shall occur. Our ASMs have been built by
deriving a \emph{state transition graph} (STG) specification directly\footnote{Note 
that the STG specification had to be extended slightly in order to transform our 
possibly non-alternating guards (which might be continuously enabled in some cycle, 
in particular during stabilization) into strictly alternating ones.} from 
\figurerefs{fig:main}--\ref{fig:top_alg} and generating the delay-insensitive
implementation via Petrify \cite{CKKLY02}. 
The actual state transition of an HSM is governed by an underlying synchronous 
\emph{transition state machine} (TSM). The TSM resolves a possibly 
non-deterministic
choice of the successor state and then performs the required transition actions:
\begin{enumerate}
\item Reset of memory flags and watchdog timers
\item Communication of the new state
\item Actual transition to the new state (i.e., enabling of further transitions
of the ASM)
\end{enumerate}
The TSM is driven by a pausible clock (see \sectionref{sec:osc}), which is started 
dynamically by the ASM before the transition. Note that this avoids the need for
synchronization with a free-running clock and hence preserves the ASMs
continuous time scale. 

\begin{figure}[t!]
\centering
\begin{tikzpicture}
\draw (0,0) node[circle,minimum width=1.5cm,draw] (s) {$A$};
\draw (1.5,0) node[circle,minimum width=1cm,draw] (s1) {\emph{Syn}};
\draw (3,0) node[circle,minimum width=1cm,draw] (res) {\emph{Cmt}};
\draw (4.5,0) node[circle,minimum width=1cm,draw] (s2) {\emph{Trm}};
\draw (6,0) node[circle,minimum width=1.5cm,draw] (sp) {$B$};
\draw (3,0) node[minimum width=4.0cm,minimum height=1.1cm,draw] (box) {};

\draw (-0.2,2) node {\TSMClock};
\draw (0.9,1.5) node[circle,minimum width=0.5cm,draw] (clk1) {\clock};
\draw (2.25,1.5) node[circle,minimum width=0.5cm,draw] (clk2) {\clock};
\draw (3.75,1.5) node[circle,minimum width=0.5cm,draw] (clk3) {\clock};

\path[->] (s) edge node[below=0.5cm] {$G'$} (s1);
\path[->] (s1) edge (res);
\path[->] (res) edge (s2);
\path[->] (s2) edge node[below=0.5cm] {$\true$} (sp);
\path[->] (s) edge  node[left] {{\footnotesize $\neg$\TSMCStop}}(clk1);
\path[->] (s2) edge  node[right] {{\footnotesize \TSMCStop}}(clk3);
\path[->] (clk1) edge (clk2);
\path[->] (clk2) edge (clk3);
\path[->] (clk1) edge[dotted] ($(clk1.south)+(0,-0.8cm)$) ;
\path[->] (clk2) edge[dotted] ($(clk2.south)+(0,-0.8cm)$);
\path[->] (clk3) edge[dotted] ($(clk3.south)+(0,-0.8cm)$);
\end{tikzpicture}
\caption{Example state transition, including the corresponding TSM.}\label{fig:TSM}
\end{figure}
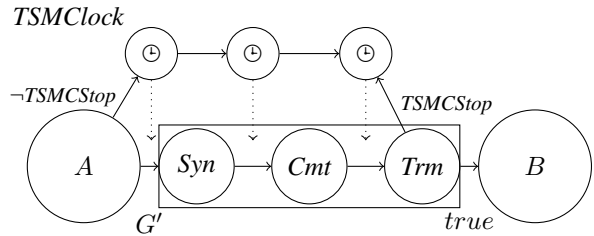

The TSM works as follows (see \figureref{fig:TSM}): Assume that the ASM is in
state $A$, and that the guard $G$ for the transition from $A$ to $B$ becomes
true. In the absence of an inhibit signal (indicating that another transition is
currently being taken, see below), the TSM clock is started. With every rising
edge of $\TSMClock$, the TSM unconditionally moves through a sequence of three
states: \emph{synchronize} (\emph{Syn}), \emph{commit} (\emph{Cmt}), and
\emph{terminate} (\emph{Trm}) shown in the rectangular box in
\figureref{fig:TSM}. In \emph{Syn}, the inhibit signal is activated to prevent
other choices from being executed in case of more than one guard becoming true.
Whereas any ambiguity can easily be resolved via some priority rule,
metastability due to (a) enabled guards that become immediately disabled again
or (b) new guards that are enabled close to transition time cannot be ruled out
in general here. However, as argued in \sectionref{sec:meta}, (a) could only do
harm to FATAL$^+$ during stabilization, due to type (3) guards; recall that type
(1) and type (2) guards are always monotonic, with the reset (of watchdog timers
and memory flags) being under the control of the local state machine. Similarly,
our proofs reveal that upsets due to (b) are fully masked after stabilization.
Thus, after stabilization, metastability of the TSM can only occur due to
unstable inputs, i.e., upsets in memory flags. Given the small window of
vulnerability of the synchronizing stage for \emph{Syn}, the resulting very low
probability of a metastable upset is considered acceptable.

Once the TSM has reached \emph{Syn}, it has decided to actually take the
transition to $B$ and hence moves on to state \emph{Cmt}. Here the watchdog
timer associated with $B$ and possibly some memory flags are cleared according
to the FATAL$^+$ state machine, and the new state B is captured by the output
flip-flops driving the state communication data bus (recall
\sectionref{sec:communication}). Note that the resulting delay must be accounted
for in the communication delay bounds $d$, $d_{\max}^+$ and $d_{\min}^+$.
Finally, the TSM moves on to state \emph{Trm}, in which the reset signals are
inactivated again and the TSM clock is halted. The inhibit signal is also
cleared here, which effectively moves the ASM to state $B$. It is only now that
guards pertaining to state $B$ may become true.

\subsection{Pausible oscillator}
\label{sec:osc}

The TSM clock is an asynchronously startable and synchronously
stoppable ring oscillator, which provides a clock signal \TSMClock\ that is \LOW\
when the clock is stopped via an input signal \TSMCStop. Note that
copies of this oscillator are used for driving the watchdog timers 
presented in \sectionref{sec:timers}.

\begin{figure}
\centering
\includegraphics[width=0.8\linewidth]{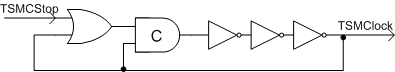}
\caption{Pausible ring oscillator implementing the TSM clock.}
\label{fig:osc}
\end{figure}

The operation of the TSM clock circuit shown in \figureref{fig:osc} is
straightforward: In its initial state, \TSMCStop=\HIGH\ and the Muller C-gate
has \HIGH\ at its output, such that \TSMClock=\LOW. Note that the circuit also
stabilizes to this state if the Muller C-gate was erroneously initialized to \LOW,
as the ring oscillator would eventually generate \TSMClock=\HIGH, enforcing
the correct initial value \HIGH\ of the C-gate.

When the ASM requests a state transition, at some arbitrary time when a
transition guard became true, it just sets \TSMCStop=\LOW. This
starts the TSM clock and produces the first rising edge of \TSMClock\ half a clock 
cycle time later. As long as \TSMCStop\ remains \LOW, the ring oscillator runs freely. 

The frequency of the ring oscillator is primarily determined by the (odd)
number of inverters in the feedback loop.\footnote{In our FPGA implementation,
the oscillator frequency is so high that we also employ a frequency divider
at the output.} It varies heavily with the operating conditions, in
particular with supply voltage and temperature: The resulting (two-sided)
clock drift $\xi$ is typically in the
range of $7\% \dots 9\%$ for uncompensated ring oscillators like ours;
in ASICs, it could be lowered down to $1\% \dots 2\%$ by special compensation
techniques \cite{SAA06}. Note that the two-sided clock drifts map to
$\vartheta = (1+\xi)/(1-\xi)$ bounds
of $1.15\dots 1.19$ and $1.02\dots 1.04$, respectively.

The stopping of \TSMClock\ is regularly initiated by the TSM itself:
With the rising edge of \TSMClock\ that moves
the TSM into \emph{Trm}, \TSMCStop\ is set to \HIGH. Since \TSMClock\
is also \HIGH\ after the rising edge,\footnote{Obviously, we only have to take 
care in the timing analysis that setting \TSMCStop=\HIGH\ occurs well within 
the first half period.} the C-gate output is also forced to \HIGH. Hence,
after having finished the half period of this final clock cycle, the feedback
loop is frozen and \TSMClock\ remains \LOW.

For metastability-free operation of the C-gate in \figureref{fig:osc}, (a) the
falling transition of \TSMCStop\ must not occur simultaneously with a rising
edge of \TSMClock, and (b) the rising transition of \TSMCStop\ must not occur
simultaneously with the falling edge of \TSMClock. (a) is guaranteed by stopping
the clock in state \emph{Trm} of the TSM, since the output of the C-gate is
permanently forced to \HIGH\ on this occasion; \TSMClock\ cannot hence generate
a rising transition before \TSMCStop\ goes to \LOW\ again. Whereas this
synchronous stopping normally also ensures (b), we cannot always rule out the
possibility of getting \TSMCStop=\HIGH\ close to the \emph{first} rising edge of
\TSMClock: (b) could thus occur due to prematurely disabled type (3) guards,
which we discussed already with respect to their potential to create
metastability in the TSM, recall \sectionref{sec:SM}. Besides being a rare
event, this can only do harm during stabilization, however.

\subsection{Watchdog Timers}
\label{sec:timers}

Recall that every ASM state, except for \acc\ in \figureref{fig:main},
is associated with at most one watchdog timer required for type (1) 
guards; \acc\ is associated with
three timers (for $T_1$ and $T_2$ as well as for $T_2^+$ in
\figureref{fig:top_alg}). A timer is reset by the TSM when its
associated state is entered.

According to \figureref{fig:timer}, every watchdog timer consists of a
synchronous
resettable up-counter that is clocked by some oscillator, and a timeout register that holds the timeout value. A
comparator raises an output signal if the counter value is
greater or equal to the register value. An SR latch with dominant reset memorizes 
this \emph{expired} condition until the timer is re-triggered.

\begin{figure}[ht!]
\centering
\includegraphics[width=\linewidth]{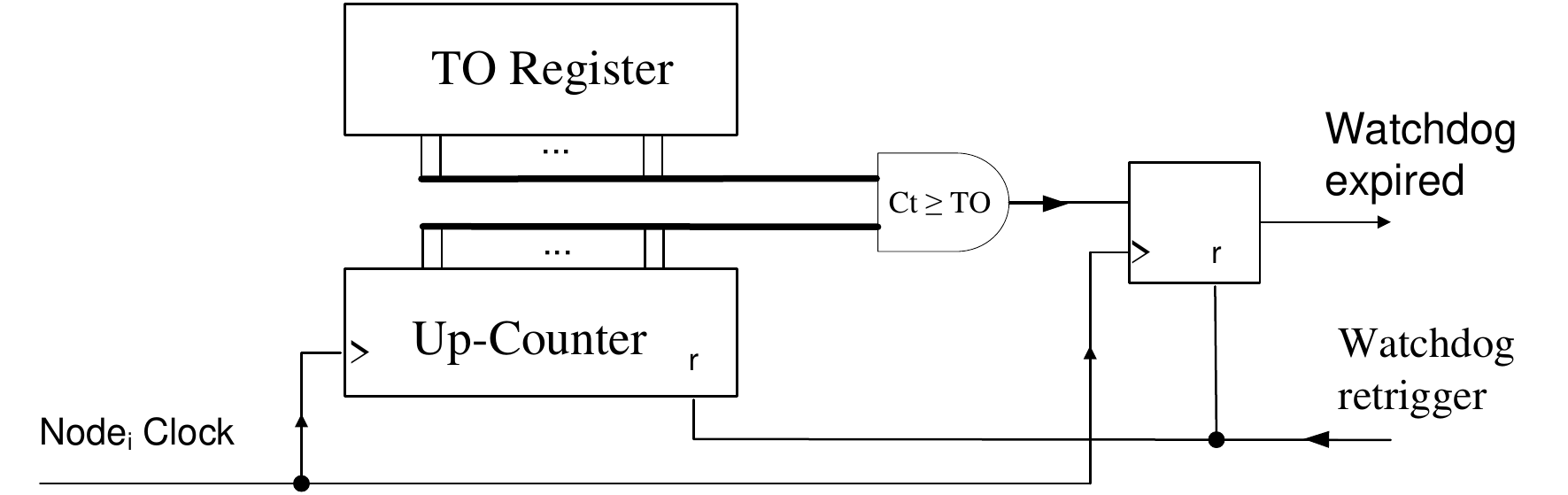}
\caption{Implementation principle of watchdog timers.}
\label{fig:timer}
\end{figure}

Like the TSM, timers are driven by pausible oscillators, which are started by the
TSM after resetting the timer and stopped synchronously upon timer expiration. 
Note that every
timer (except for the multiple \acc\ timers, which share a common
oscillator that is stopped when the largest timeout expires) 
is provided with a dedicated oscillator in our implementation for simplicity.
This not only
avoids quantization errors in the continuous timing of the ASM state transitions, but is
also mandatory for avoiding the potential of metastability due to timer resets
colliding with the transitions of a free-running clock. In our implementation, the 
timer reset takes place in TSM state \emph{Cmt}, while the oscillator is started 
in state \emph{Trm}. This well ordered sequence rules out all metastability issues.

As for the watchdog timer with random timeout~$R_3$ in \figureref{fig:resync},
our implementation uses an \emph{linear feedback shift register}
(LFSR) clocked by a dedicated oscillator: A uniformly
distributed random value, sampled from the LFSR, is loaded into the
timeout register whenever the watchdog timer is re-triggered. Note that for many
settings, it is reasonable to assume that the new random value remains a secret
until the timeout expires, as it is not read or in any other way considered by
the node until then. As our prototype implementation is not meant for studying
security issues, this simple implementation is thus sufficient.

%% file: evaluation.tex
\section{Experimental Evaluation}
\label{sec:experiments}

Our prototype implementation has been written in VHDL and compiled
for an Altera Cyclone IV FPGA using the Quartus tool. 

Apart from standard functional and timing verification via Modelsim, we conducted
some preliminary experiments for verifying the assumed properties (glitch-freeness, monotonicity, etc.)
of the synthetisized implementations of our core building blocks: Since
FPGAs do not natively provide the basic elements required for asynchronous designs, and we have no control 
over the actual mapping of functions to the available LUTs (e.g.\ our threshold modules are implemented via LUT
instead of the intended combinational AND-OR networks), we had
to make sure that properties that hold naturally in ``real'' asynchronous
implementations also hold here. Backed up by the (positive) results of
these experiments, a complete system consisting of $n=4$ resp.\ $n=8$ nodes 
(tolerating at most $f=1$ resp.\ $f=2$
Byzantine faulty nodes) has been built and verified to work as
expected; overall, they consume 23000 resp.\ 55000 logic blocks.
Note however, that both designs include the test environments which
makes up a significant part of the designs.

To facilitate systematic experiments, we also developed a custom test bench 
 that provides the following functionality:
\begin{enumerate}
\item[(1)] Measurement of pulse frequency and skew at different nodes.
\item[(2)] Continuous monitoring of the potential of non-deter\-ministic HSM
state transitions.
\item[(3)] Starting the entire system from an arbitrary state (including memory
flags and timers), both deterministically and randomly chosen.
\item[(4)] Resetting a single node to some initial state, at arbitrary times.
\item[(5)] Varying the clock frequency of any oscillator, at arbitrary times.
\item[(6)] Varying the communication delay between any pairs of sender and
receiver, at arbitrary times.
\end{enumerate}
All these experiments can be done with and without up to $f$ (actually, $f+1$ to also include
excessively many) Byzantine nodes. To this end, the HSMs of
at most $f+1=3$ nodes can be replaced by special devices that allow to (possibly inconsistently)
communicate, via the communication data buses, any HSM state to any receiver HSM at any time.

(1) is accomplished using standard measurement equipment (logic analyzer, oscilloscope, 
frequency counter) attached to the appropriate signals routed via output pins. (2) is
implemented by memorizing any event where more than one guard
is enabled when the TSM performs its first state transition, in a flag that can be externally monitored.

(3) is realized by adding a scan-chain to the implementation, which allows
to serially shift-in arbitrary initial system states at run-time. Repeated
random experiments are controlled via a Python script executed at a PC workstation,
which is connected via USB to an ATMega 16 microcontroller (uC) that acts as a scan-controller
towards the FPGA: The uC takes a bit-stream representing an initial 
configuration, sends it to the FPGA via the serial scan-chain interface, 
and signals the FPGA to start execution of FATAL$^+$. When
the system has stabilized, the
uC informs the Python script which records the stabilization time and proceeds with sending the
next initial configuration.

To enable (4)--(6), the testbench provides a global high-resolution clock that can
be used for triggering mode changes. To
ensure its synchrony w.r.t.\ the various node clocks, all start/stoppable ring oscillators
are replaced by start/stoppable oscillators that derive their output from the global
clock signal. (4) is achieved by just forcing a node to reset to its
initial state for this run at any time during the current execution.
In order to facilitate (5), dividers combined with clock multipliers (PLLs) are 
used: For any oscillator, it is possible to choose
one of five different frequencies (0, excessively slow, slow, fast, excessively fast) at
any time.
For (6), a variable delay line implemented as a synchronous shift 
register of length $X\in[0,15]$, driven by the global clock, can be inserted in any 
data bus connecting different HSMs individually. 

In order to exercise also complex test scenarios in a reproducible way, a dedicated
\emph{testbed execution state machine} (TESM), driven by the global clock, is used
to control the times and nodes when and where clock speeds, transmission delays and 
communicated fault states are changed and when a single node is reset
throughout an execution of the system. 
Transition guards may involve global time and any combinatorial expression
of signals used in the implementation of FATAL$^+$, i.e., any predicate on the
current system state.\footnote{To decrease the experiment setup time
(after all, changing the TESM requires recompilation of the entire system),
the TESM is gradually changed to also incorporate additional parameters and configuration
information downloaded at run-time via the uC.}

\medskip

Using our testbench, it was not too difficult to get our FATAL$^+$ up and 
running. As expected, we spotted several hidden design errors that showed
up during our experiments, but also some errors (like a missing factor
of $\vartheta$ in one of our timeouts due to a typo) in the initial version of
our theoretical analysis, which caused deviations of the measured w.r.t.\
the predicted performance.

Finally, using the implementation parameters $\vartheta=1.3$, $d=13T$,
     $d_{\min}^+ = d_{\max}^+ = 3T$, where $T$ is the experimental
     clock period~$T = 400$ns, and minimal timeouts according to the
     constraints, we conducted the following experiments, observing
     the behavior of both, the FATAL$^+$ as well as the underlying
     FATAL system:  

\emph{(A)} Maximum skew scenarios, including effects of excessively
     small/fast  clocks and message delays: The experimental results
     confirmed the analytic predictions as being tight: As shown in
     \figureref{fig:screenshot1}, pulses of the 8 node FATAL resp.\
     FATAL$^+$ system occur at a frequency of about 62Hz resp.\ 10kHz.
Note that the quite low values for the frequency stem from the fact
     that we were intentionally slowing down the system in order to
     carry out our worst-case experiments.
     
The figure further clearly demonstrates the capability of FATAL$^+$ to
     generate pulses with significantly less skew ($1\mu s$) on top of
     the FATAL pulses.

\begin{figure}[bt]
  \centering
  \includegraphics[width=0.9\columnwidth]{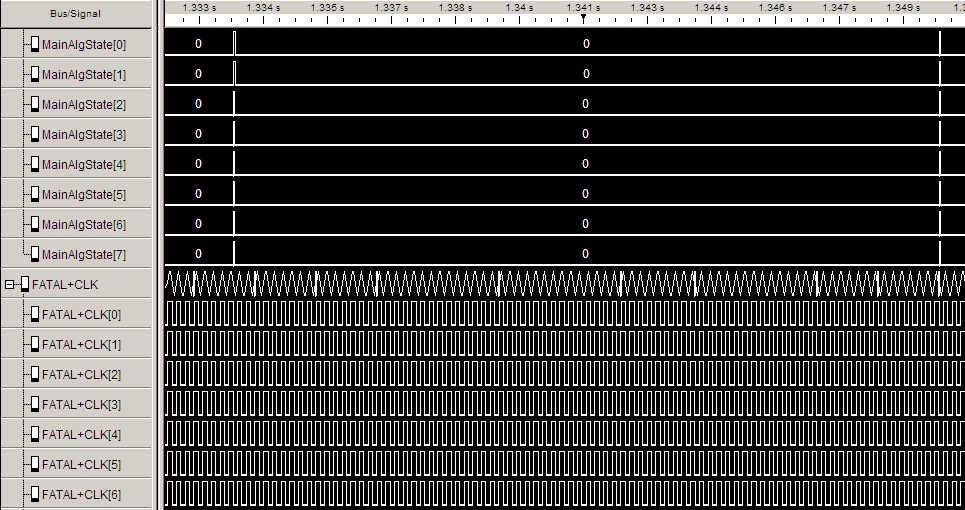}
  \caption{FATAL and FATAL$^+$ clocks: {\tt MainAlgState[i]}=1 iff $i$ is in \acc,
           and {\tt FATAL+CLK[i]} is $i$'s FATAL+
           signal.}\label{fig:screenshot1}
\end{figure}


Further experiments, involving $f=2$ Byzantine nodes, were used to
     produce a worst-case scenario for the FATAL skew ($6\mu s$).
The resulting waveform is depicted in \figureref{fig:screenshot2}.

\begin{figure}[bt]
  \centering
  \includegraphics[width=0.9\columnwidth]{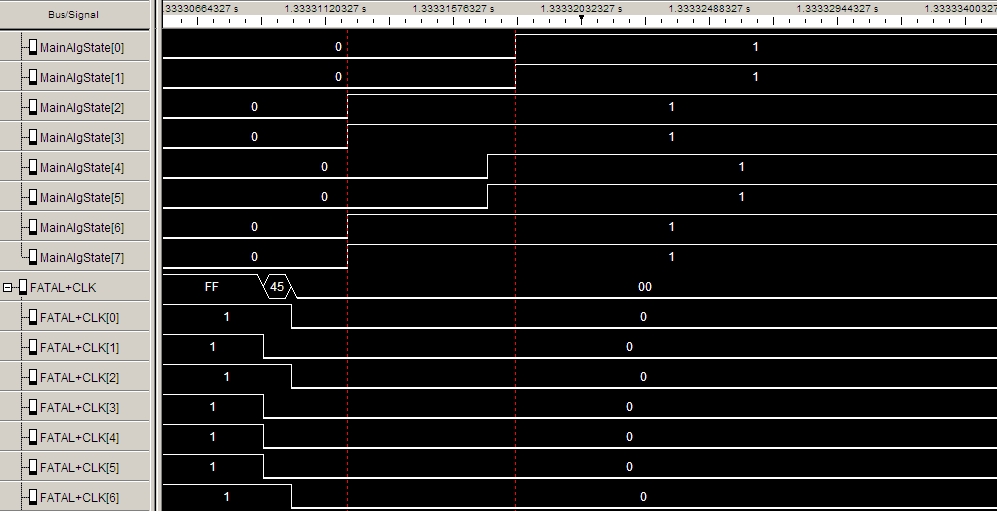}
  \caption{Worst-case skew scenario for FATAL clocks.}\label{fig:screenshot2}
\end{figure}

\emph{(B)} Scenarios leading to the potential of non-deterministic HSM state
     transitions in the absence of Byzantine nodes (which would invalidate our
     proof of metastability-freedom if happening after stabilization): We run
     17000 experiments, in each of which the 8 node system was set
     up with randomly chosen message delays between nodes and random
     clock speeds and stabilized from random initial states. Within 10
     seconds from stabilization on, not a single upset was encountered in any
     instance.

\emph{(C)} Stabilization of an 8-node system from random initial states, with
randomly chosen clock speeds and message delays (without Byzantine
nodes). Over 4000 runs have been performed. A considerable fraction of the setups stabilizes
within less than 0.035 seconds, which can be credited to the fast stabilization
mechanism intended for individual nodes resynchronizing to a running system (see
\figureref{fig:e1} and \figureref{fig:e2}). The remaining runs stabilize, supported by the
resynchronization routine, in less than 10 seconds, which is less than
the system's upper bound on $T(1)$. Note that the stabilization time is
inversely proportional to the frequency, i.e., in a system that is not slowed
down stabilization is orders of magnitude faster.

\begin{figure}[bt]
  \begin{minipage}[t]{0.48\columnwidth}
    \centering
    \includegraphics[width=0.95\columnwidth]{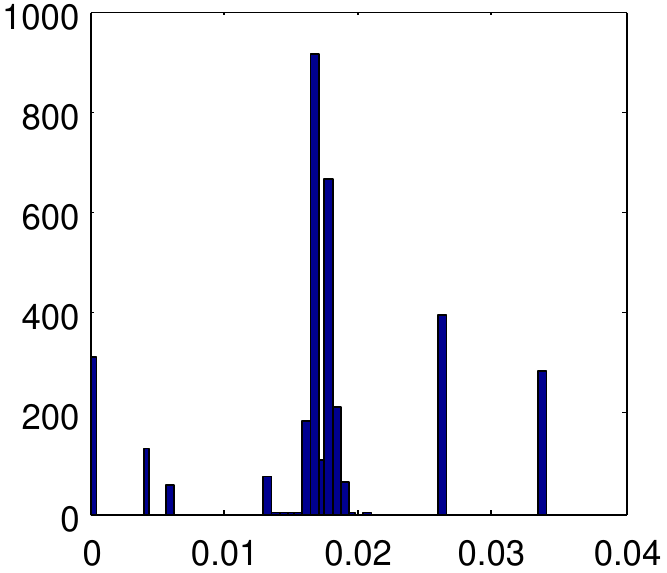}
    \caption{Head of distribution of stabilization times (in s) for
      over 6500 randomly initialized 8-node instances.}\label{fig:e1}
  \end{minipage}
  \hfill
  \begin{minipage}[t]{0.48\columnwidth}
    \centering
    \includegraphics[width=0.87\columnwidth]{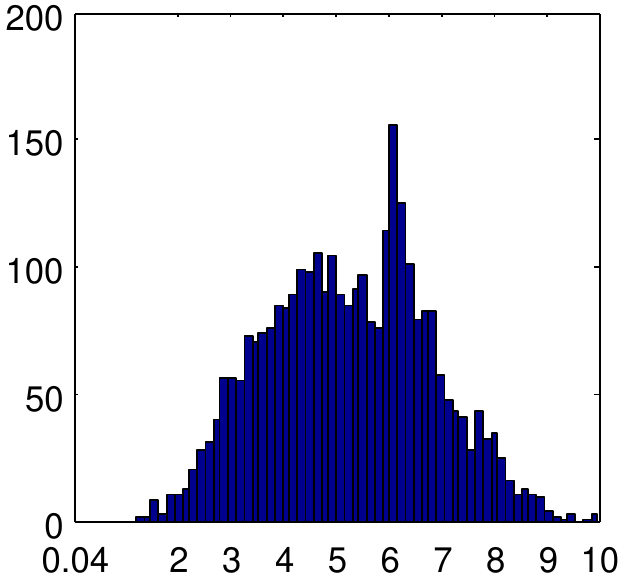}
    \caption{Tail of distribution of stabilization times (in s) for
      over 6500 randomly initialized 8-node instances.}\label{fig:e2}
  \end{minipage}
\end{figure}

%% file: conclusions.tex
\section{Conclusions}\label{sec:conclusions}

We conclude with a few considerations regarding the asymptotic complexity of
implementations of FATAL+ and future work. The algorithm has the favorable
property that nodes broadcast a constant number of bits in constant time, which
clearly is optimal. While it would be beneficial to reduce node degrees, this
must come at the price of reducing the resilience to faults~\cite{PSL80,DHS86}.
In terms of the number of Byzantine faults the algorithm can sustain in relation
to node degrees, the algorithm is asymptotically optimal as well. It is subject
to future work to extend the algorithm to be applicable to networks of lower
degree in a way preserving resilience to a (local) number of faults that is
optimal in terms of connectivity.

Furthermore, it is not difficult to see that except for the threshold modules,
each node comprises a number of basic components that is linear in $n$
(cf.~\cite{FSFK06:edcc}, where similar building blocks were used). In an ASIC
implementation, one could implement the threshold modules by sorting networks,
resulting in a latency of $\BO(\log n)$ and a gate complexity of $\BO(n \log
n)$~\cite{ajtai83}. Clearly, it is necessary to have conditions involving more
than $f$ nodes in order to overcome $f$ Byzantine faults. Hence, assuming
constant fan-in of the gates, both the current and envisioned solutions are
asymptotically optimal with respect to latency. Optimality of an
implementation relying on sorting networks with respect to gate complexity is
not immediate, however there is at most a logarithmic gap to the trivial lower
bound of $\Omega(n)$.